\documentclass[12pt]{article}

\newcommand{\blind}{1}

\usepackage[toc,page]{appendix}

\RequirePackage{amsthm,amsmath,amsfonts,amssymb,statmath}
\usepackage[authoryear]{natbib}
\RequirePackage[colorlinks,citecolor=blue,urlcolor=blue]{hyperref}
\RequirePackage{graphicx}

\usepackage{enumitem}
\usepackage{sidecap}
\usepackage{booktabs}
\usepackage{longtable}
\usepackage{multirow}
\usepackage{float}
\usepackage{mathtools}
\usepackage{color}
\usepackage{statmath}
\usepackage{xfrac}
\usepackage{bigints}
\usepackage{subfigure}
\usepackage{bbm}
\usepackage{array}
\usepackage{booktabs}
\usepackage{siunitx, tabularx}
\usepackage{adjustbox}
\usepackage{xr}
\usepackage{arydshln,leftidx}
\usepackage{verbatim}
\usepackage{upgreek}
\usepackage{algorithm,algpseudocode}
\usepackage{amssymb}
\usepackage{epstopdf}
\usepackage{bm}
\usepackage{makecell}  
\usepackage{bigints}
\usepackage{threeparttable} 

\renewcommand{\tilde}{\widetilde}

\newcommand{\cR}{\mathcal{R}}
\newcommand{\cL}{\mathcal{L}}
\newcommand{\cT}{\mathcal{T}}
\newcommand{\ct}{\textbf{\emph{{t}}}}

\newcommand{\cX}{\mathcal{X}}
\newcommand{\cK}{\mathcal{K}}
\newcommand{\cI}{\mathcal{I}}
\newcommand{\cG}{\mathcal{G}}

\newcommand{\bP}
{\overline{P}}
\newcommand{\bp}
{\overline{p}}

\newcommand{\bS}
{\overline{S}}
\newcommand{\bs}
{\overline{s}}

\newcommand{\sminus}{\scalebox{0.50}[1.0]{\( - \)}}

\newcommand{\nbracket}[1]{\left( #1 \right)}
\newcommand{\cbracket}[1]{\left\{ #1 \right\}}
\newcommand{\rbracket}[1]{\left[ #1 \right]}

\newcommand{\norm}[1]{\left\vert \left\vert #1 \right\vert \right\vert }

\newtheorem{proposition}{Proposition}[section]
\newtheorem{theorem}{Theorem}[section]
\newtheorem{lemma}[theorem]{Lemma}
\newtheorem{result}{Result}

\newtheorem{remark}{Remark}

\newtheorem{assumption}{Assumption}
\newtheorem{definition}{Definition}[section]

\newcommand{\ind}[1]{\mathbf{1}_{\{#1\}}}

\newcommand{\grad}{{\nabla}}

\def\argmin{\mathop{\rm argmin}\limits}

\def\1{\bm{1}}

\newcommand{\Cov}{\mathrm{Cov}}

\newcommand{\Pp}{{\mathbb P}}
\newcommand{\Ee}{{\mathbb E}}
\newcommand{\cP}{{\mathcal{P}}}

\newcommand{\bn}{\mathbf{n}}

\newcommand{\RD}{\operatorname{Diff}}

\newcommand{\abs}[1]{\lvert #1 \rvert}
\renewcommand{\hat}[1]{\widehat{#1}}

\renewcommand{\ind}[1]{\mathbbm{1}[#1]}

\newcommand{\twofigs}[2]{
\hbox to\hsize{\hss
\vbox{\psfig{figure=#1,width=2.7in,height=2.0in}}\qquad
\vbox{\psfig{figure=#2,width=2.7in,height=2.0in}}
\hss}}

\usepackage{setspace}
\begin{document}

\date{}

\def\spacingset#1{\renewcommand{\baselinestretch}%
{#1}\small\normalsize} \spacingset{1}


\if1\blind
{
  \title{\bf Classification Trees with Valid Inference \\ via the Exponential Mechanism}
    \author{Soham Bakshi \\
    Department of Statistics,
		University of Michigan,
         MI, USA.\\
    and  \\
    Snigdha Panigrahi \\
    Department of Statistics,
		University of Michigan,
         MI, USA.\\}
  \maketitle
} \fi

\if0\blind
{
  \bigskip
  \bigskip
  \bigskip
  \begin{center}
    {\LARGE\bf Classification Trees with Valid Inference\\ via the Exponential Mechanism}
\end{center}
  \medskip
}\fi

\begin{abstract}
Decision trees are widely used for non-linear modeling, as they capture interactions between predictors while producing inherently interpretable models.
Despite their popularity, performing inference on the non-linear fit remains largely unaddressed. 
This paper focuses on classification trees and makes two key contributions.
First, we introduce a novel tree-fitting method that replaces the greedy splitting of the predictor space in standard tree algorithms with a probabilistic approach.
Each split in our approach is selected according to sampling probabilities defined by an exponential mechanism, with a temperature parameter controlling its deviation from the deterministic choice given data.
Second, while our approach can fit a tree that with high probability coincides with the fit produced by standard tree algorithms at low temperatures, it is not merely predictive; unlike standard algorithms, it enables inference by taking into account the highly adaptive tree structure.
Our method produces pivots directly from the sampling probabilities in the exponential mechanism.
In theory, our pivots allow asymptotically valid inference on the parameters in the predictive fit, and in practice, our method delivers powerful inference without sacrificing predictive accuracy, in contrast to data splitting methods.
\end{abstract}

\noindent%

\spacingset{1.7}

\section{Introduction}
\label{sec:intro}

Breiman, in his influential paper \cite{breiman2001statistical} on the two modeling cultures, emphasized the \textit{algorithmic} approach to fitting predictive models while treating the underlying data-generating mechanism as a black box.
Although predictive accuracy is, and should remain, an important measure of a model’s success, it offers only partial insight into the black box.
Valid inference for algorithmic models equips practitioners with tools, such as hypothesis tests and interval estimates, to probe, estimate and understand various aspects of the underlying mechanism.
In this work, we focus on one such algorithmic modeling approach, wherein a decision tree is fitted to data with binary outcomes, commonly called a classification tree.
Rather than treating prediction and inference as mutually exclusive goals, we propose a method for fitting predictive trees with valid inference on the fit.

\begin{figure}[tp]
    \centering
    \includegraphics[width=1\textwidth, height=5cm]{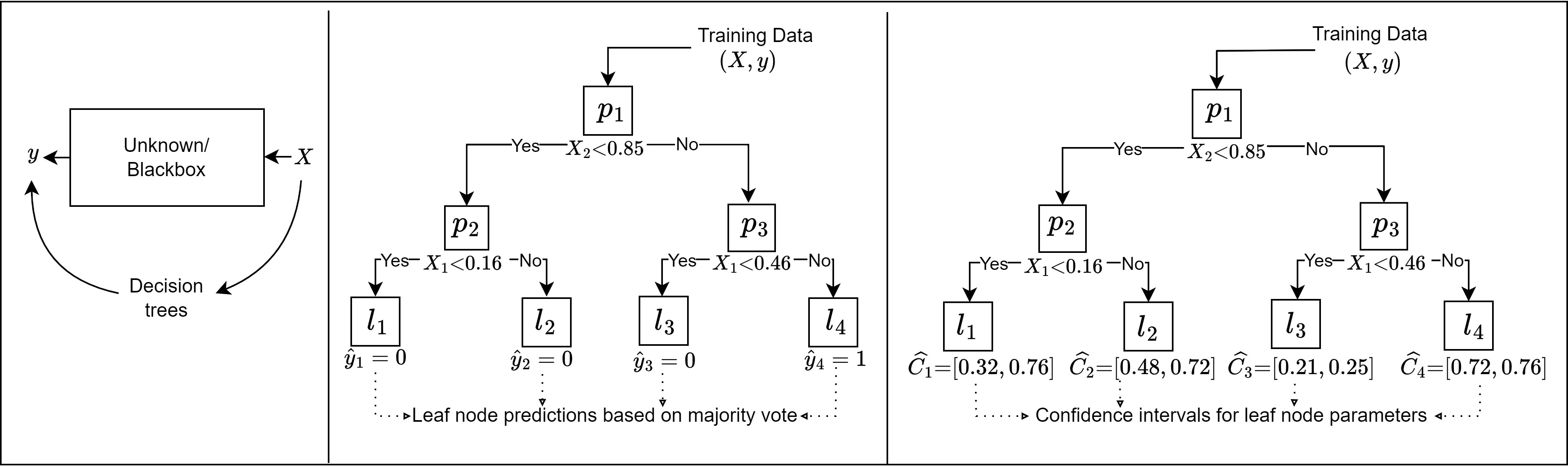} 
    \vspace{-0.5cm}
    \caption{
    Illustration of inference for decision trees. Left: the
    \emph{algorithmic modeling} approach described in \cite{breiman2001statistical}, in which a decision tree is used as a predictive model to explain the black box linking predictors $X$ to outcome $y$; Middle: a classification tree generated by our tree-fitting algorithm using the proposed exponential mechanism, with leaf nodes reporting predictions based on a majority vote; Right: the same tree fit, now equipped with inference (in particular, confidence intervals) on the mean parameters of subpopulations identified by the rule-based partition of the classification tree. 
    }
    \label{fig:plot1}
\end{figure}

Decision trees can capture highly non-linear relationships between predictors and outcomes while remaining largely interpretable.
Splits in a tree fit are easily visualized as rule-based partitions of the predictor space, dividing the sample space into groups of observations that are usually described by simple models.
For this reason, such rule-based partitions of the data find wide applications in domains such as business and biomedical sciences, where interpreting the algorithmic model is often prioritized, sometimes even over maximizing predictive accuracy.
Inference is naturally of interest in these applications, as it explains the black box mechanism captured by the fit.

Despite the ability to produce interpretable fits, standard algorithms for classification trees, such as CART \cite{breiman1984classification} and C4.5 \cite{quinlan2014c4}, have remained primarily predictive.
Classic inference tools do not apply to a tree fit, as the practitioner does not start with a predefined model.
Instead, an algorithmic model is obtained from a predictive fit to the data—here, resulting from a greedy partitioning of the predictor space.

This work makes two key contributions.
\begin{enumerate}[leftmargin=*]
\item First, we introduce a novel probabilistic tree-fitting algorithm, in which the usual deterministic splitting scheme is replaced with a sampling scheme based on the exponential mechanism.
These sampling probabilities include a temperature parameter that controls the degree of randomization in the selection of splits.
In particular, our approach fits a predictive tree that, with high probability, coincides with the tree fit produced with deterministic splitting rules when the temperature parameter is small.

This type of temperature-controlled sampling mechanism is commonly used in  generative modeling and large language models for prediction, e.g.,  \cite{brown2020language, liang2022holistic, wu2025role}, where the temperature parameter controls the balance between output diversity and coherence.
In our work, we demonstrate that this scheme serves purposes beyond prediction, highlighting its ability to generate valid inference while accounting for the highly adaptive and non-linear tree structure.

\item Second, we demonstrate how pivotal quantities can be derived from the sampling probabilities in the exponential mechanism, which then allows for valid inference for the algorithmic model. 
Our pivots are constructed using the full dataset, eliminating the need for any form of data splitting and, importantly, enabling reliable and powerful inference without compromising predictive accuracy.
Moreover, the construct of these pivots is agnostic to the  choice of information measure for splitting the predictor space, allowing flexible inference akin to data splitting methods.
\end{enumerate}

Deferring a more technical account of our contributions to the next section, where we review connections with various research areas, we present Figure \ref{fig:plot1}.
The left panel illustrates Breiman's depiction of a black box modeled with a decision tree, as described in \cite{breiman2001statistical} .
The middle panel shows a classification tree produced by our new probabilistic algorithm, visualized as a rule-based partition of the predictor space on simulated data from a Bernoulli outcome variable.
Details of the simulation scheme are provided later in Section \ref{sec:simulations}.
The right panel shows confidence intervals for the means in each terminal region of the adaptive partition generated by the algorithm, allowing practitioners to report inference for the subgroups identified by the fit alongside the prediction rule.
Finally, although we focus on classification trees, the proposed randomization approach readily applies to regression trees with continuous outcomes; we focus on the former due to the lack of inference methods in this domain.

The rest of the paper is organized as follows. 
In Section~\ref{sec:contributions}, we discuss the connections between our work and the related literature. 
In Section~\ref{sec:treegrowing}, we introduce a temperature-controlled sampling scheme based on the exponential mechanism for selecting splits in our tree-fitting procedure. Section~\ref{sec:inference} develops a conditional inference method to construct pivotal quantities for parameters based on the data-adaptive tree fit, while Section~\ref{sec:asymptotictheory} establishes the conditional validity of these pivotal quantities. Section~\ref{sec:adaptivetemp} presents data-adaptive strategies for setting the temperature parameters at each split, along with the corresponding inferential adjustments.
In Sections~\ref{sec:simulations} and \ref{sec:dataapps}, numerical experiments demonstrate the ability of our pivots to produce valid and powerful inference in the algorithmic model. 
Section \ref{sec:conclusion} concludes the paper with a summary of our contributions and potential extensions.

\section{Review of related literature}
\label{sec:contributions}

\paragraph{Inference with decision trees.} This work builds on the idea proposed by \cite{bakshi2024inference} of using external randomization to enable inference conditional on non-linear tree fits.
Our current method, however, differs from the earlier approach in two important respects.
First, our work explores an entirely new randomization approach based on sampling splits according to an exponential mechanism. 
These sampling probabilities allow us to directly account for the adaptive tree structure and eliminate the need of any further numerical approximation to compute inference, as required by \cite{bakshi2024inference}.
Second, we develop a novel theoretical framework using statistics whose asymptotic distribution delivers valid inference for classification trees with binary outcomes. 
Existing methodology designed for regression trees with normal data, such as \cite{bakshi2024inference} or their non-randomized predecessor in \cite{neufeld2022tree}, do not apply to binary outcomes.
Our inferential theory, leveraging properties of the exponential mechanism, may be of independent interest to readers, and could motivate the application of this randomization approach to other inference problems involving data-adaptive targets.

\cite{mentch2016quantifying} developed theory and methodology to test hypotheses about predictions aggregated across many trees built on subsampled data. 
More recently, \cite{naef2023distributional} introduced bootstrap-based techniques to quantify uncertainty in the estimated conditional distribution obtained via distributional random forests.
However, both lines of research differ from the focus of our work, as such approaches do not provide estimates of population parameters associated with the tree fit, such as the means of subpopulations identified by the adaptive partition of the predictor space.
Recent work \cite{klusowski2021sharp, klusowski2024large} advances the theory of tree learners such as CART and C4.5, establishing risk consistency properties.
Our contribution complements this line of work by providing a framework for inference in an algorithmic model obtained from the adaptive tree fit, rather than aiming for predictive guarantees.

\paragraph{Exponential mechanism: connections with privacy and stability.} 
The exponential mechanism, introduced by \citet{mcsherry2007mechanism}, is a fundamental tool in differential privacy that selects $r$ from a discrete set of possible  outcomes $\mathcal{R}$ with probability proportional to $\exp\{u(r)/\tau\,\Delta u\}$, where $\tau$ is a temperature parameter controlling the randomness in selection, $u(r)$ represents a utility score and $\Delta u$ denotes its sensitivity. Our tree-fitting algorithm employs a similar Gibbs-type randomization, not for privacy but to enable valid inference that accounts for the highly adaptive nature of the tree fit.

The exponential mechanism has recently been applied beyond privacy in settings conceptually related to our goal.
For example, \cite{karwa2016inference} demonstrate that exponential mechanism can yield asymptotically normal estimators even after randomization.
\cite{awan2019benefits} develop methodology for differentially private functional data analysis, proving a central limit theorem for the exponential mechanism.
Differential privacy has been shown to imply strong uniform stability \citep{bassily2016algorithmic}, and stability itself has been linked to robust risk estimation and cross-validation \citep{lei2025stability}. 
In \cite{zhang2024winners}, the exponential mechanism has been applied to construct confidence sets for model selection problems.
By contrast, we utilize the exponential mechanism to perform inference while making no stability assumptions about the algorithm.

\paragraph{Selective inference with conditional guarantees.} 
While our method shares the conceptual idea behind conditional approaches in selective inference (e.g., \cite{Lee_2016, fithian2017optimal, randomizedresponse, snigdhaintergative, gao2024selective, perry2025inference}), we demonstrate how properties of the exponential mechanism can be exploited to circumvent several well-known difficulties associated with these methods.

First, in constructing pivotal quantities for inference, our approach avoids explicitly characterizing the selection event or relying on geometric simplifications, unlike \cite{Lee_2016, neufeld2022tree, panigrahi2023approximate, appromixateMLE}.
Instead, the sampling probabilities in the exponential mechanism directly yield adjustments for the adaptive fitting process and can apply to any information gain measure used in standard tree-fitting algorithms.
Second, the inferential guarantees developed in our work make no assumptions about how rare or likely the tree fit is, remaining valid even when selection probabilities approach zero by leveraging properties of the exponential mechanism.
By contrast, much of the existing work on asymptotic selective inference (e.g., \cite{randomizedresponse}) relies on restrictive assumptions, requiring the selection event to have a probability bounded away from zero, or under Gaussian randomization schemes, such as those in \cite{rasines2023splitting, rasines2023local, panigrahi2023carving, bakshi2024selectiveinferencetimevaryingeffect}, to have probability vanishing to zero only at specific rates.

\section{The tree-growing algorithm}
\label{sec:treegrowing}

In this section, we briefly review the standard tree-growing algorithm and then introduce our proposed algorithm using the exponential mechanism.

\subsection{A brief review}

Suppose our data consists of $n$ observations $(Y_i, X_i)$ for $i \in [n] = \{1, 2, \ldots, n\}$, where $Y_i \in \cbracket{0,1}$ and $X_{i}=(X_{i,1}, \ldots, X_{i,p}) \in \cX=\cX_{1} \times \cdots \times \cX_{p}$ represent the binary outcome and the $p$ predictors for the $i$-th sample, respectively. 
We denote the outcome vector by $Y$ and the predictor matrix by $X$, obtained by stacking the $n$ observations in our data row-wise.
For any region $\cR \subseteq \cX$ and a $p$-dimensional vector $x\in \cX$, let $\ind{x\in \cR}$ denote the indicator function that takes the value $1$ when $x \in \cR$ and $0$ otherwise.

The standard tree algorithm proceeds by recursively splitting regions $\cP \subseteq \cX$, starting from the full predictor space $\cX$.
At each step of the algorithm, when a region $\cP$ is to be divided, a split $s \in \cK(\cP)$ is \emph{selected} based on the data, where $\cK(\cP)$ denotes the set of all permissible splits that divide the region $\cP$. 
When the split is binary (i.e., two-way), it partitions the parent region $\cP$ into two subregions: $\cbracket{\cP(s), \cP'(s)}$, where $\cP'(s)= \cP \setminus \cP(s)$. Since binary splits are more commonly used than multi-way splits in practice, we focus on binary splits for the rest of this paper.

The optimal split, which partitions the parent region $\cP$, is selected in a purely data-driven manner by maximizing an information gain measure evaluated across all candidate splits according to a chosen impurity or information criterion.
For a region $\cR$, we denote the data-dependent information measure associated with the region by $I(\cR; Y)$ and the number of observations within the region by $n_{\cR}$. 
Now, for a parent region $\cP$ and a possible split $s \in \cK(\cP)$ on it, we define the  information gain as
\begin{equation}
G_{\cP}(s; Y)= I(\cP; Y)- \left\{\frac{n_{\cP(s)}}{n_{\cP}} I(\cP(s); Y)+\frac{n_{\cP'(s)}}{n_{\cP}} I(\cP'(s); Y)\right\}.
\label{defn: gain}
\end{equation}
In the standard tree-growing algorithm, the greedily selected split $s^*= s^*(\cP)$, obtained as
\vspace{-0.5cm}
$$
s^*= \underset{s \in \cK(\cP)}{\text{argmax}}\; G_{\cP}(s; Y),
\vspace{-0.5cm}
$$  
yields the partition $\cbracket{\cP(s^*), \cP'(s^*)}$ of the region $\cP$.


\subsection{Exponential mechanism for selecting splits}

Unlike the deterministic split in the standard algorithm, our method, for a given parent region $\mathcal{P}$, samples a realization of a random split $S(\mathcal{P})$ from the permissible splits $\mathcal{K}(\mathcal{P})$.
Next we define the split variable $S(\mathcal{P})$.

Given $Y = y$, the split variable $S(\mathcal{P})$ in our method follows a discrete distribution supported on $\mathcal{K}(\mathcal{P})$, with probability mass function:
\vspace{-0.5cm}
$$
\Pp_{S(\cP)}= \{\Pp_{y}(s; \cP): s\in \mathcal{K}(\mathcal{P})\},\vspace{-0.5cm}
$$
where
\vspace{-0.5cm}
\begin{equation}
\Pp_{y}(s; \cP)=\mathbb{P}\rbracket{S(\cP) = s \mid Y=y}= \frac{\exp \nbracket{\epsilon_{\cP}^{-1} G_{\cP}(s; y) }}{\displaystyle\sum_{s'\in \cK(\cP)}\exp \nbracket{\epsilon_{\cP}^{-1} G_{\cP}(s'; y)}}, \text{ for } s\in \cK(\cP).
\label{exp:mechanism}
\end{equation}
The probability on the right-hand side of \eqref{exp:mechanism} arises from an exponential mechanism using the information gains of all candidate splits.
Here, $\epsilon_{\mathcal{P}}$ is the \emph{temperature parameter} that determines the amount of randomization in the split at $\mathcal{P}$; larger values of $\epsilon_{\mathcal{P}}$ result in greater randomness in split selection, while smaller values yield selections closer to the standard algorithm.

Our split-selection algorithm is simple: at each internal region $\mathcal{P}$, we sample a realization of the random variable $S(\mathcal{P})$ according to the distribution in \eqref{exp:mechanism}, and apply this procedure recursively to partition the predictor space with independent splits at each region.

Next, we make a few important observations about the randomization introduced by the exponential mechanism.
First, it follows directly from the definition of the sampling probabilities that the random split variable $S(\cP)$ is independent of $Y$ given the gain functions at parent region $\cP$, denoted by 
\vspace{-0.5cm}
$$
\textsc{Gains}_{\cP}(Y)=\{G_{\cP}(s; Y): s\in \cK(\cP)\},
$$
i.e., $\Pp_{y}(s; \cP)= \mathbb{P}\rbracket{S(\cP) = s \mid \textsc{Gains}_{\cP}(Y)= \textsc{Gains}_{\cP}(y)}$. 
As a consequence, the proposed sampling scheme can be executed using only the information gains at each parent region, without requiring access to the raw outcome values.

Second, Proposition~\ref{prop: exp conv} formalizes that for a sufficiently small temperature parameter $\epsilon_{\mathcal{P}}$, the probability that our randomly chosen split matches the optimal split returned by the standard tree algorithm approaches $1$. 
\begin{proposition}
For any split $s \in \cK(\cP)$ on an interior region $\cP$, where $s\notin \underset{s'\in \cK(\cP)}{\text{argmax}} \; G_{\cP}(s'; y)$, it holds that
\vspace{-0.5cm}
$$
\Pp_{y}(s; \cP) \to 0 \; \text{ as } \; \epsilon_{\cP} \to 0.
\vspace{-0.5cm}
$$
\label{prop: exp conv}
\end{proposition}

Section~\ref{sec:simulations} discusses our choice of $\epsilon_{\mathcal{P}}$ in detail; here we summarize the tree-growing method based on \eqref{exp:mechanism} with $\epsilon_{\mathcal{P}}=\epsilon$, maximum depth $d_{\text{max}}$, and minimum leaf size $n_{\text{min}}$ in Algorithm~\ref{alg:RCT} to simplify the exposition.

\begin{algorithm}[h]
\setstretch{1.3}
\caption{Tree-growing algorithm via the exponential mechanism}
\begin{algorithmic}[1]
    \State \textbf{Input:} Training data \( (X, y) \), maximum depth \( d_{\text{max}} \), minimum samples \(n_{\text{min}}\), temperature parameter $\epsilon$
    \State \textbf{Initialize:} Root region \( \cP_{1} = \cX \), current depth \( d = 0 \)
    \For{each region \( \cP \) with \( n_{\cP} \) samples at depth \( d \)}
        \If{$d \leq d_{\text{max}}$ and $n_{\cP} > n_{\text{min}}$}
            \State Compute the gain function \( G_{\cP}(s; y) \) for all possible splits \( s \in \cK(\cP)\)
            \State Draw a realization of the split variable \( S(\cP)\) according to the distribution \\ \hspace*{1.2cm} in \eqref{exp:mechanism} with $\epsilon_{\cP}= \epsilon$. 
            \State Partition \( \cP \) based on selected split  $\{S(\cP)=s\}$ into \( \{\cP(s), \cP'(s)\} \)
            \State depth $d \gets d + 1$ 
        \Else
            \State Return \( P \) as a terminal region
        \EndIf
    \EndFor 
\end{algorithmic}
\label{alg:RCT}
\end{algorithm}

\noindent \textbf{Other extensions.} The tree-growing procedure introduced in Algorithm~\ref{alg:RCT} naturally induces a new ensemble method obtained by aggregating multiple randomized trees fitted using the proposed splitting scheme. We refer to this approach as an \emph{ensemble RCT}, whose predictive performance  is compared in Section~\ref{appendix: ensembleRCT} with that of standard random forests based on bootstrap aggregation. The key distinction lies in the form of randomization used: in our approach, no data point is ever excluded from any tree.

For a single tree, pruning can be performed after growing it to a fixed depth. Our approach, based on the exponential mechanism, provides a new pruning scheme inspired by cost-complexity pruning, with valid inference in the pruned tree described in Appendix~\ref{appendix:ccp}.



\subsection{Connection to additive randomization}

At first glance, the randomization from the exponential mechanism appears quite different in form from the additive noise used in \cite{bakshi2024inference}. 
However, Proposition~\ref{prop: exp gumbel} shows that this randomization scheme is, in fact, equivalent to an additive randomization mechanism, in which independent Gumbel noise is added to the gain functions underlying the splitting rules. 
The equivalence in our next result, often referred to as the \emph{Gumbel–max trick}, is adopted from the literature on differential privacy and discrete sampling \citep{dwork2014algorithmic,maddison2014asampling}. 
\begin{proposition}
Let $V_{\cP}(s; Y) = \epsilon_{\cP}^{-1} G_{\cP}(s; Y) + Z_s$, where $Z_s \overset{\text{i.i.d.}}{\sim} \text{Gumbel}(0,1)$ for $s\in \cK(\cP)$. Then, 
$\Pp_{y}(s; \cP) = \mathbb{P}\left[ V_{\cP}(s; Y)  \geq V_{\cP}(s'; Y) \, \text{ for all } \; s' \in \cK(\cP)\mid Y=y\right].$
\label{prop: exp gumbel}
\end{proposition}

\section{Our inferential method}
\label{sec:inference}

In this section, we develop our inference method for the algorithmic model derived from the tree fitted using Algorithm~\ref{alg:RCT}.
We begin by stating the model and establishing the necessary notations. 
Next, we introduce the statistics whose distributions provide pivotal quantities for inference. 
Finally, we present a roadmap of the main inferential results, which lead to pivotal quantities for parameters in the fitted non-linear model.

\subsection{Algorithmic model and notations}
\label{subsec:treefit}
Suppose the terminal regions, or leaves, obtained from the tree fit $\cT$ are denoted by $\{L_{1}, \ldots, L_{M}\}$, representing a partition of the predictor space $\cX$ into $M$ disjoint regions.
For our realization of the outcome $Y = y$, assume that we observed a partition of $\cX$ into $M=m$ regions:
$$
\left\{L_1 = l_1, L_2 = l_2, \ldots, L_m = l_m\right\},
$$ 
where the symbols in lowercase represent the observed realizations of the corresponding random variables on the observed data.

A simple model after observing the tree fit is given by: 
\begin{equation}
Y_{i} \stackrel{\text{ind}}{\sim} \text{Bernoulli}(\theta_{i}), \text{ where } \theta_{i} = \sum_{m'=1}^{m} \pi_{m'} \mathbbm{1}[X_{i} \in l_{m'}],  \ \text{ for } i\in [n],
\label{tree:model}
\end{equation}
where $\pi_{m'}$ represents the unknown population mean for the group of observations in $l_{m'}$. 
Without loss of generality, we focus on $\pi_1$, the mean parameter associated with the observed leaf $L_1 = l_1$, for the remainder of this section.
We make a few remarks on the interpretation of parameters in the algorithmic model and the scope of inference with our approach.

\begin{remark}
The target of inference $\pi_{m'}$, although in practice interpreted with respect to the algorithmic model, is a well-defined object even when the fitted model does not exactly coincide with the underlying black box. 
More specifically, the mean parameter associated with the leaf region $l_{m'}$, written as $\mathbb{E}\left[(n_{l_{m'}})^{-1}\sum_{i=1}^{n} Y_i \ \mathbbm{1}[X_{i} \in l_{m'}]\right]$ 
where $n_{l_{m'}}$ denotes the number of observations in $l_{m'}$, corresponds to the population mean of the subgroup of observations that fall into that leaf region.
\label{rem:interpretation}
\end{remark}

\begin{remark}
The scope of inference with our approach is not limited to the mean parameters $\pi_{m'}$, for $m'\in [m]$.
Although we focus on inference for the mean parameters for clarity, our approach can readily extend to target parameters that are monotone transformations of these parameters, such as log-odds ratios, as well as to linear combinations, for example the difference in means between two sibling regions.
\label{rem:scope}
\end{remark}


Before constructing inference, we introduce notations for the tree fit and its observed values.
Let $\cT_{1}$ be the ancestral subtree of $L_1$, consisting of all internal regions and branches along the unique path from the root $\cX$ to $L_1$.
Further, let $\bP_{1}:=\cbracket{P_{1}=\cX,\ldots,P_{D}}$ be the internal regions in  the subtree $\cT_{1}$ and let $\bS_{1}:=\cbracket{S_{1},\ldots,S_{D}}$ denote the sequence of splits, where $S_h=S(P_h)$ is the random split made on $P_h$ for $h\in [D]$. Given $Y = y$, we denote the realized fit by $\cT_{1} = \ct_{1}$.
Suppose that in $\ct_{1}$, we observe $D=d$ many splits $\{S_{h}=s^1_h: h\in [d]\}$, made on a corresponding set of observed internal regions $ \bp_{1}:=\cbracket{p_{1}, \ldots, p_{d}}$. 
For each $h \in [d]$, let $S_h = s^1_h$ denote the selected split at node $p_h$, chosen from the set of candidate splits at $p_h$, denoted by $\cK_h := \cK(p_h) = \cbracket{s^1_h, s^2_h, \ldots, s^{k_h}_h}$, where $|\cK_h| = k_h$; e.g., $S_{1} = s^1_1$ is the first split made at the root $p_{1} = \cX$, and $S_{d} = s^1_d$ is the final split made at $p_d$, leading to the leaf $l_1$.

\subsection{Inference conditional on tree fit}
\label{subsec:inferenceguarantees}

As described in the proposed methodology, our inferential method constructs an asymptotic confidence interval $\widehat{C}_{1}(Y)$ for $\pi_1$ by conditioning on the sequence of greedy splits that lead to the associated leaf region $l_1$.
Specifically, our intervals satisfy the following guarantee 
\begin{equation}
\label{cond:guarantee}
\lim_{[\bn]  \to \infty} \Big|\mathbb{P}\left[\pi_{1} \in \widehat{C}_{1}(Y) \; \Big\lvert \;  \{S_{h}=s^1_h: h\in [d]\} \right]- (1-\alpha)\Big| =0,
\end{equation}
as the number of samples in the tree fit is allowed to grow.

Deferring the precise asymptotic framework to the next section, we focus here on providing intuition for the validity of inference implied by this guarantee.
In particular, for sufficiently large samples, \eqref{cond:guarantee} implies that we have coverage guarantees conditional on the observed tree fit.
Put differently, the interval estimates possess the standard frequentist coverage guarantees under repeated runs of the algorithm that yield the same fitted tree as observed on $Y=y$.

Furthermore, it follows trivially that this guarantee also implies a weaker unconditional coverage guarantee, which is marginal over all possible tree fits, i.e.,
\vspace{-0.5cm}
$$
\mathbb{P}\left[\pi_{1} \in \widehat{C}_{1}(Y) \; \Big\lvert \;  \{S_{h}=s^1_h: h\in [d]\} \right] \geq 1-\alpha 
\implies \mathbb{P}\left[\pi_{1} \in \widehat{C}_{1}(Y) \right] \geq 1-\alpha.
$$

\subsection{Key statistics for inference}
\label{subsec:keystatistics}

Let the sample proportion within the observed leaf $l_1$ be
\begin{equation}
\hat{\pi}_{1} = \frac{1}{n_{1}}\sum_{i=1}^n Y_i \cdot \ind{X_{i}\in l_{1}},
\label{mean:stat}
\end{equation}
where $n_{1}:=n_{l_{1}}$ denotes the number of observations in $l_1$. 
The statistic $\hat{\pi}_{1}$ alone cannot provide valid inference for $\pi_1$.
This is because the leaf $l_1$ results from all upstream splits within the subtree $\ct_1$, which depend on observations beyond those in $l_1$.
We now define the relevant statistics needed to construct valid inference for $\pi_1$. 


Formally, for each internal region $\cR$ in $\ct_1$, let $\hat{\pi}_{\cR} = \dfrac{1}{n_{\cR}}\sum_{i=1}^n Y_i \cdot \ind{X_{i}\in \cR}$.
We let
\begin{equation}
\label{orth:stat1}
    \widehat{\pi}_{\sminus 1}(\cR) := \sqrt{n_{\cR}}\widehat{\pi}_{\cR} - \frac{n_{\cR}^{1}}{\sqrt{n_{1}n_{\cR}}} \sqrt{n_1}\widehat{\pi}_{1},
\end{equation}  
where $n_{\cR}^{1}:=n_{\cR \cap l_1}$ is the number of observations that lie in $\cR \cap l_1$.
With this definition, note that $\widehat{\pi}_{\cR}$ can be decomposed into a component involving $\widehat{\pi}_1$ and another involving $\widehat{\pi}_{\sminus 1}(\cR)$, as follows:
$$
\widehat{\pi}_{\cR} = \dfrac{n_{\cR}^{1}}{n_{\cR}} \widehat{\pi}_{1}+ \dfrac{1}{\sqrt{n_{\cR}}}\widehat{\pi}_{\sminus 1}(\cR).
$$
Now, for each internal region $p_h$ in the subtree $\ct_1$ and its candidate split region $p_h\nbracket{s_k^h}$, we define:
\vspace{-0.5cm}
\begin{equation*}
    \hat{T}_{n,\sminus 1}^{h} :=  \widehat{\pi}_{\sminus 1}\nbracket{p_{h}}, \quad \hat{T}^{h,k}_{n,\sminus 1} :=  \widehat{\pi}_{\sminus 1}\nbracket{p_{h}\nbracket{s_{h}^{k}}}, \  \text{ for }\  k \in \{2,\ldots, k_h\},  h \in [d].
    \vspace{-0.5cm}
\end{equation*} 
We include $n$ in the subscript of these statistics to highlight their dependence on the number of samples in the data.

Finally, using the statistics defined in \eqref{mean:stat} and \eqref{orth:stat1}, we let
\vspace{-0.5cm}
\begin{equation*}
\hat{T}_{n} := \begin{bmatrix}
 \hat{T}_{n,1}   & (\hat{T}_{n,\sminus 1})^\top 
\end{bmatrix}^\top\in \mathbb{R}^{\bar{k}}, 
\vspace{-0.5cm}
\label{defn: T}
\end{equation*}
where $\bar{k} := 1+ \sum_{h=1}^{d} k_h$, and 
\vspace{-0.5cm}
\begin{align*}
\begin{gathered}
\hat{T}_{n,1} := \sqrt{n_1}\hat{\pi}_{1} \in \mathbb{R}, \ \hat{T}_{n,\sminus 1} := \text{Vec}\nbracket{\begin{bmatrix} 
          \hat{T}^{h}_{n,\sminus 1} &   \hat{T}^{h,2}_{n,\sminus 1} & \cdots & \hat{T}^{h,k_{h}}_{n,\sminus 1} \end{bmatrix}^{\top} \in \mathbb{R}^{k_{h}}: \ h \in [d]} \in \mathbb{R}^{\bar{k}-1}.
\end{gathered}
\end{align*}

We refer to $\hat{T}_n$ as the \textit{key statistics} for inference on the parameter $\pi_1$. 
As shown next, the gain measure at each candidate split in the fit $\ct_1$ depends on data solely through these statistics.
Therefore, the distribution of $\hat{T}_n$, conditional on the tree fit, yields valid inference for $\pi_1$, justifying our use of the term key statistics for these variables.
For the reader’s convenience, the symbols used to represent the adaptive tree structure and construct the statistics of interest are summarized in Table \ref{tab:notations}.

\begin{table}[!h]
\label{tab:notations}
\centering
\renewcommand{\arraystretch}{1.25}
\begin{tabular}{|p{3.6cm}|p{11.2cm}|}
\hline
\multicolumn{2}{|l|}{\textbf{Notations for adaptive tree structure}} \\
\hline
$L_1=l_1$ & Observed leaf of interest. \\
$\mathcal{T}_1=\ct_1$ & Ancestral subtree of $l_1$; all internal nodes along the path from $\mathcal{X}$ to $l_1$. \\
$D=d$ & Total number of internal splits on the path to $L_1=l_1$. \\
$\bS_1=\bs_1$ & Sequence of random splits $\{S_h=s^1_h: h\in[d]\}$ along $\ct_1$. \\
$\bp_1=\{p_1,\ldots,p_d\}$ & Ordered internal regions from root to the leaf $l_1$. \\
$\hat{\pi}_1$ & Sample proportion in the leaf $l_1$. \\
$\widehat{\pi}_{-1}(\cR)$ & Orthogonalized proportion in region $\cR$, removing dependence on $\hat{\pi}_1$. \\
\hline
\multicolumn{2}{|l|}{\textbf{Node-specific notations}} \\
\hline
$S_h=s_h^1$ & Observed split at node $p_h$. \\
$\mathcal{K}_h=\{s_h^1,\ldots,s_h^{k_h}\}$ & Candidate split set at $p_h$ (size $k_h$). \\
$\hat{T}_{n,\sminus 1}^{h} \in \mathbb{R}$ & Orthogonalized statistic $\widehat{\pi}_{\sminus 1}(p_h)$ for internal node $p_h$. \\
$\hat{T}_{n,\sminus 1}^{h,k}\in \mathbb{R}$ & Orthogonalized statistic $\widehat{\pi}_{\sminus 1}(p_h[s_h^{k}])$ for $p_h[s_h^{k}]$, for $k=2,\ldots,k_h$. \\
\hline
\multicolumn{2}{|l|}{\textbf{Key statistics for inference}} \\
\hline
$\hat{T}_{n,1}=\sqrt{n_1}\hat{\pi}_1 \in \mathbb{R}$ & Main statistic involving the target parameter $\pi_1$. \\
$\hat{T}_{n,\sminus 1} \in \mathbb{R}^{\sum_{h=1}^d k_h}$ & Stacked vector of all orthogonalized statistics 
$\begin{bmatrix}
\hat{T}^{h}_{n,\sminus 1} &
\hat{T}^{h,2}_{n,\sminus 1} &
\cdots &
\hat{T}^{h,k_{h}}_{n,\sminus 1}
\end{bmatrix}^{\top}$ for $h \in [d]$. \\
$\hat{T}_n=[\hat{T}_{n,1},\hat{T}_{n,\sminus 1}^\top]^\top$ & Full vector of key statistics for selective inference. \\
$\bar{k}=1+\sum_{h=1}^d k_h$ & Total dimension of the statistic vector $\hat{T}_n$. \\
\hline
\end{tabular}
\caption{Summary of notation for the fitted tree, node-level quantities, and key statistics used in selective inference.}
\end{table}

\subsection{Information gains: a function of the key statistics}
\label{subsec:infogain}

Lemma~\ref{lem: suff stats}, the main result in this section, shows that the information gains within the relevant subtree are functions of the key statistics $\widehat{T}_n$ defined in \eqref{defn: T}.

To state this result, we rewrite the impurity measure $I(\cP; Y) := \cI(\hat{\pi}_{\cP})$ to make its dependence on the sample proportion $\hat{\pi}_{\cP}$ explicit and introduce the necessary definitions.
\begin{definition}
\label{defn: eta}
Consider a split $s$ made on the parent region $\cP$, partitioning it into the subregions $(\cP(s), \cP'(s))$. 
We define the scaling vector 
\vspace{-0.5cm}
\begin{align*}
   & \  \eta(s; \cP) := \\ & \begin{bmatrix} 
    \dfrac{\sqrt{n_{1}}}{n_{\cP}} & \dfrac{1}{\sqrt{n_{\cP}}}& \dfrac{n_{\cP(s)}^{1}}{n_{\cP(s)}\sqrt{n_{1}}}& \dfrac{1}{\sqrt{n_{\cP(s)}}} & \dfrac{n_{\cP}^{1}-n_{\cP(s)}^{1}}{\nbracket{n_{\cP}-n_{\cP(s)}}\sqrt{n_{1}}}  & \dfrac{\sqrt{n_{\cP}}}{n_{\cP}-n_{\cP(s)}} & -\dfrac{\sqrt{n_{\cP(s)}}}{n_{\cP}-n_{\cP(s)}} & \dfrac{n_{\cP(s)}}{n_{\cP}}\end{bmatrix}^{\top},
\end{align*} 
 where $n_{\cR}^{1}$ is the number of observations in $\cR \cap l_1$.
For brevity of notation, we let
\vspace{-0.5cm}
\begin{equation*}
 \eta^{k}_{h} := \eta(s_{h}^{k}; p_h) \in \mathbb{R}^{8}, \quad \text{ for } k \in [k_h],  h\in [d].
\end{equation*} 
\end{definition}

\begin{definition}
\label{def:2}
For a vector $\eta\in \mathbb{R}^{8}$, let 
\vspace{-0.5cm}
$$
D(\eta):=\cbracket{(x,y,z): (\eta_{1}x+\eta_{2}y, \eta_{3}x+\eta_4 z, \eta_5 x+ \eta_6 y+\eta_7 z) \in (0,1) \times (0,1)\times (0,1)} \subseteq \mathbb{R}^{3}.
\vspace{-0.5cm}
$$
We then define a mapping $g_{\eta}(.): D(\eta) \to \mathbb{R}$ as:
\vspace{-0.5cm}
\begin{equation*}
g_{\eta}\nbracket{x,y,z} := \cI\nbracket{\eta_{1}x+\eta_{2}y} -\eta_{8}\cI\nbracket{\eta_{3}x+\eta_{4}z}-(1-\eta_{8})\cI\nbracket{\eta_{5}x+\eta_{6}y+\eta_{7}z}.
\end{equation*} 
\end{definition}

\begin{definition}
\label{defn: g}
Let $v \in \mathbb{R}^{\bar{k}-1}$ be a $(\bar{k}-1)$-dimensional vector, which we decompose as $\text{Vec}([v^{h} \  v^{h,2}  \cdots  v^{h,k_{h}} ]^{\top} \in \mathbb{R}^{k_{h}}:  h \in [d]) \in \mathbb{R}^{\bar{k}-1}$, where $v^{h}$ and $v^{h,k}$, for $k \in \cbracket{2,\ldots,k_h}$, follow the ordering of the components of $\hat{T}_{n,\sminus 1}$ in \eqref{defn: T}.
Let 
\vspace{-0.5cm}
$$
\tilde{D}_{n}:= \cbracket{[u \ v]^\top \in \mathbb{R}^{\bar{k}}: 
 \nbracket{u,v^{h},v^{h+1}} \in D(\eta^1_h), \  \nbracket{u,v^{h},v^{h,k}}\in D(\eta^k_h), \ \forall h\in[d], k\in[k_h]},
$$
where  $\eta^{k}_{h}$ is as given in Definition \eqref{defn: eta}. 
Then, for all $h\in [d]$ and $ k \in \cbracket{2,\ldots,k_h}$, we define the non-negative functions $\cG^{k}_{h}(u,v): \tilde{D}_{n} \to \mathbb{R}^+$ as
\vspace{-0.5cm}
\begin{equation*}
    \cG^{1}_{h}(u,v) := g_{\eta^{1}_{h}}\nbracket{u,v^{h},v^{h+1}}, \quad  \cG^{k}_{h}(u,v) := g_{\eta^{k}_{h}}\nbracket{u,v^{h},v^{h,k}}.
    \vspace{-0.5cm}
\end{equation*} 
\end{definition}

We are now ready to state Lemma~\ref{lem: suff stats}.

\begin{lemma}
Let, for all $h \in [d]$ and $k \in [k_h]$, $\eta_{h}^{k}$ and $\cG_{h}^{k}(.,.)$ be as defined in \eqref{defn: eta} and \eqref{defn: g}, respectively.
Consider the information gain measure $G_{\cP}(s; Y)$ underlying the splitting rules. Then, it holds that
\vspace{-0.8cm}
$$
G_{p_h}(s^{k}_{h}; Y) = \cG^{k}_{h}\nbracket{\hat{T}_{n,1}, {\hat{T}_{n,\sminus 1}}}, \quad \text{ for all} \ \; k \in [k_h],\  h\in [d].
$$
 \label{lem: suff stats}
\end{lemma}

\begin{definition}
\label{defn: Lambda}
For each $h \in [d]$, we define the mapping $\lambda_{h}(., .): \tilde{D}_n\to (0,1)$ as
\begin{align*}
\lambda_{h}(u,v) &:= \dfrac{\exp \left(\epsilon_{h}^{-1} \cG^{1}_{h}(u,v)\right)}{\sum_{k=1}^{k_{h}}\exp \left(\epsilon_{h}^{-1} \cG^{k}_{h}(u,v)\right)},
\end{align*}
where $\epsilon_{h}:=\epsilon_{p_h}$ and $\tilde{D}_n$ is as given in Definition \ref{defn: g}.
\end{definition}

As an immediate consequence of Lemma~\ref{lem: suff stats}, we obtain the following result. Lemma~\ref{lem: connection sampling prob} links $\lambda_{h}(\hat{T}_{n,1}, \hat{T}_{n,\sminus 1})$ with the probability of the observed splits under the exponential mechanism in~\eqref{exp:mechanism}.

\begin{lemma}
For $h\in [d]$, we that $\lambda_{h}(\hat{T}_{n,1}, \hat{T}_{n,\sminus 1})= \Pp_{Y}(s_{h}^{1};p_{h})$.
\label{lem: connection sampling prob}
\end{lemma}
The proof of this lemma follows directly from Definition \ref{defn: Lambda} and is therefore omitted. As the next step, we characterize the distribution of the key statistics $\hat{T}_n$, adjusted for the selection of splits, to enable valid inference for our parameter of interest, $\pi_1$.

\subsection{Roadmap to pivots for valid inference}
\label{subsec:inferenceroadmap}

We summarize the main inferential results in three steps, outlining the pre-conditional and conditional distributions of key statistics. 
Formal statements along with the precise asymptotic analysis appear in the next section.

\noindent{\textbf{Step 1}}. \quad First, we show that if the adaptive nature of the tree fit is ignored---i.e., we treat the tree structure as fixed---then the key statistics defined in~\eqref{defn: T} follow an asymptotic normal distribution as the number of samples in each region of the fit increases.

\begin{result}{}{\normalfont{[An informal view of the pre-conditional distribution}].} Suppose that the tree fit is specified a priori and independent of the observed data. Let $\Pi_{n,1} = \sqrt{n_1} \pi_1$ and $\Pi_{n,\sminus 1}= \Ee[\hat{T}_{n,\sminus 1}]$. 
Then, 
\vspace{-0.5cm}
   $$ \Sigma_n^{-1/2}\begin{bmatrix} \hat{T}_{n,1}- \Pi_{n,1} \\ \hat{T}_{n,\sminus 1}- \Pi_{n,\sminus 1}\end{bmatrix} \indist \mathcal{N}_{\bar{k}}\nbracket{0_{\bar{k}}
    ,I_{\bar{k}, \bar{k}}},$$ 
where $\Sigma_n = \begin{bmatrix}
        \sigma_{n,1}^{2} & 0_{\bar{k}-1}^{\top}\\ 0_{\bar{k}-1} & \underset{(\bar{k}-1 \times \bar{k}-1)}{\Sigma_{n,\sminus 1}}
    \end{bmatrix}$, with the exact expressions for $\sigma_{n,1}^2$ and $\Sigma_{n,\sminus 1}$ provided later in Proposition \ref{prop: normaldistn}.
\label{prop: marginal}
\end{result}

To help build intuition behind our inferential approach, in the next two steps we assume for now that the pre-conditional distribution of the key statistics, as stated in Result \ref{prop: marginal}, holds exactly rather than asymptotically.

\noindent{\textbf{Step 2}}. \quad In reality, the tree fit, and the resulting algorithmic model depends on the observed data. To obtain the guarantees specified in Section~\ref{subsec:inferenceguarantees}, we condition the pre-conditional distribution on the observed event
\vspace{-0.5cm}
$$
\cbracket{\bS_{1}=\bs_{1}}=\underset{h\in [d]}{\cap}\left\{S_{h} = s^1_h\right\}
$$
where $s^1_h$ is the (selected) split sampled from the candidate set $\cK_h$ when splitting the internal region $p_h$, and $\bs_{1}= \{s^1_h: h\in [d]\}$. This yields the conditional distribution of the key statistics.

Result~\ref{prop: cond} states the conditional density of the key statistics given the observed  tree fit.
\begin{result}{}{\normalfont{[An informal view of the conditional distribution}].} 
Suppose that $$\begin{bmatrix} \hat{T}_{n,1}- \Pi_{n,1} \\ \hat{T}_{n,\sminus 1}- \Pi_{n,\sminus 1}\end{bmatrix} \sim \mathcal{N}_{\bar{k}}\nbracket{0_{\bar{k}}, \Sigma_n}.$$ 
Then, the conditional density of $(\hat{T}_{n,1}, \hat{T}_{n,\sminus 1})$  given the event $\cbracket{\bS_{1}=\bs_{1}}=\underset{h\in [d]}{\cap}\left\{S_{h} = s^1_h\right\}$, at $(u, v)\in \mathbb{R}^{\bar{k}}$, is:
\vspace{-0.5cm}
$$
\frac{ \phi(u;\Pi_{n,1},\sigma_{n,1}^{2})\times \phi(v;\Pi_{n,\sminus 1},\Sigma_{n,\sminus 1}) \times \prod_{h=1}^{d}\lambda_{h}(u,v)}{ \int\phi(u';\Pi_{n,1},\sigma_{n,1}^{2})\times \phi(v';\Pi_{n,\sminus 1},\Sigma_{n,\sminus 1}) \times \left\{\prod_{h=1}^{d}\lambda_{h}(u',v')\right\} \; du' dv'}.
\vspace{-0.5cm}
$$
\label{prop: cond}
\end{result}

First we make a few observations about the conditional density of our key statistics.
\begin{enumerate}[leftmargin=*]
\item[(i)] The function $\lambda_{h}(u,v)$ accounts for the selection of the random split at $p_h$, and therefore, we refer to it as the correction factor.
When applied to the pre-conditional density of the key statistics, this correction factor yields the conditional density of the key statistics given the fitted tree.
\item[(ii)] The correction factor $\lambda_{h}(\cdot, \cdot)$ at each split is available in closed form, as it coincides with the sampling probabilities based on the exponential mechanism in~\eqref{exp:mechanism}.
\end{enumerate}
We present a proof of Result \ref{prop: cond} below.

\begin{proof}
Applying the Bayes' rule, we obtain that the conditional density of $(\hat{T}_{n,1}, \hat{T}_{n,\sminus 1})$  given $\cbracket{\bS_{1}=\bs_{1}}$, at $(u, v)\in \mathbb{R}^{\bar{k}}$  is equal to 
 \begin{equation}
 \frac{ \phi(u;\Pi_{n,1},\sigma_{n,1}^{2})\times \phi(v;\Pi_{n,\sminus 1},\Sigma_{n,\sminus 1}) \times \mathbb{P}\rbracket{\bS_{1}=\bs_{1}\mid \hat{T}_{n,1}=u,\hat{T}_{n,\sminus 1}=v}}{ \int \phi(u';\Pi_{n,1},\sigma_{n,1}^{2})\times \phi(v';\Pi_{n,\sminus 1},\Sigma_{n,\sminus 1}) \times \mathbb{P}\rbracket{\bS_{1}=\bs_{1}\mid \hat{T}_{n,1}=u',\hat{T}_{n,\sminus 1}=v'} \; du' dv'}.
 \label{cond:density:BR}
\end{equation} 

Under the exponential mechanism in \eqref{exp:mechanism}, the probability of observing $s_{h}^{1}$ given the key statistics equals:
\vspace{-0.5cm}
$$\Pp_{Y}(s_{h}^{1};p_{h})=\lambda_{h}(\hat{T}_{n,1}, \hat{T}_{n,\sminus 1}),$$
as shown in Lemma \ref{lem: connection sampling prob}.
Furthermore, 
\vspace{-0.5cm}
$$
\mathbb{P}\rbracket{S_{h}^{1}=s_{h}^{1}\mid \hat{T}_{n,1},\hat{T}_{n,\sminus 1}}= \mathbb{P}\rbracket{S_{h}^{1}=s_{h}^{1}\mid \hat{T}_{n,1},\hat{T}_{n,\sminus 1}, \textsc{Gains}_{p_h}(Y)}=\Pp_{Y}(s_{h}^{1};p_{h})=\lambda_{h}(\hat{T}_{n,1}, \hat{T}_{n,\sminus 1}),
\vspace{-0.5cm}
$$
since $S_{h}^{1}$ is independent of $Y$ given $\textsc{Gains}_{p_h}(Y)$, the collection of gains at $p_h$.
Finally, noting that the random splits are conditionally independent given the key statistics, we have 
\vspace{-0.5cm}
\begin{align*}
\mathbb{P}\rbracket{\bS_{1}=\bs_{1}\mid \hat{T}_{n,1}=u,\hat{T}_{n,\sminus 1}=v} &= \prod_{h=1}^{d} \mathbb{P}\rbracket{S_{h}^{1}=s_{h}^{1}\mid \hat{T}_{n,1}=u,\hat{T}_{n,\sminus 1}=v}\\
&= \prod_{h=1}^{d} \lambda_{h}(u, v). 
\end{align*} 
By substituting the expression for $\mathbb{P}\rbracket{\bS_{1}=\bs_{1}\mid \hat{T}_{n,1}=u,\hat{T}_{n,\sminus 1}=v}$ into \eqref{cond:density:BR}, we obtain the conditional density stated in the claim.
\end{proof}

\noindent{\textbf{Step 3}}. \quad The conditional density in Result~\ref{prop: cond} involves the unknown parameters $\Pi_{n,1}$ and $\Pi_{n,\sminus 1}$. 
To construct a pivot for $\Pi_{n,1}=\sqrt{n_{1}}\pi_{1}$, a rescaled version of our  target parameter $\pi_1$, we further condition on $\hat{T}_{n,\sminus 1}$ to remove the nuisance parameters $\Pi_{n,\sminus 1}$ from this conditional density. 
In this final step, we apply a probability integral transform (PIT) to the conditional density of 
\begin{equation}
    \hat{T}_{n,1} \;  \Big\lvert \;  \cbracket{\bS_{1}=\bs_{1}, \hat{T}_{n,\sminus 1}}
    \label{cond:density:pit}
\end{equation}
to obtain a pivotal quantity for our target parameter.

\begin{result}{}{\normalfont{[Pivot}].} Consider the random variable 
\[
    \text{\normalfont Pivot}(\hat{T}_{n,1}; \Pi_{n,1})= \frac{\int_{-\infty}^{\hat{T}_{n,1}}\phi(u;\Pi_{n,1},\sigma_{n,1}^{2})\times \prod_{h=1}^{d}\lambda_{h}(u,\hat{T}_{n,\sminus 1})  du}{\int_{-\infty}^{\infty}\phi(u;\Pi_{n,1},\sigma_{n,1}^{2})\times \prod_{h=1}^{d}\lambda_{h}(u,\hat{T}_{n,\sminus 1})  du}.
\]
Under the conditional density stated in Result~\ref{prop: cond}, 
\vspace{-0.5cm}
$$
\text{\normalfont  Pivot}(\hat{T}_{n,1}; \Pi_{n,1})\; \Big\lvert \;  \cbracket{\bS_{1}=\bs_{1}}\sim \text{\normalfont Uniform}[0,1].
\vspace{-0.5cm}
$$
\label{prop: pivot}
\end{result}

\begin{proof}
Result \ref{prop: cond} implies directly that this density function, evaluated at $u$, is proportional to
\vspace{-0.5cm}
$$\phi(u;\Pi_{n,1},\sigma_{n,1}^{2})\times \prod_{h=1}^{d}\lambda_{h}(u,\hat{T}_{n,\sminus 1}).$$
Applying the PIT to this density function yields the desired pivot, which satisfies
\vspace{-0.5cm}
$$
\text{\normalfont  Pivot}(\hat{T}_{n,1}; \Pi_{n,1})\; \Big\lvert \;  \cbracket{\bS_{1}=\bs_{1}, \hat{T}_{n,\sminus 1}}\sim \text{\normalfont Uniform}[0,1].
\vspace{-0.5cm}
$$
Applying the law of iterated expectation establishes the claim in the result.
\end{proof}

As emphasized earlier, Results \ref{prop: cond} and \ref{prop: pivot} in Steps 2 and 3 rely on the assumption that the pre-conditional distribution of the key statistics is exactly normal. 
However, as stated in Result \ref{prop: marginal}, when dealing with Bernoulli outcome variables, this distribution is only asymptotically normal and does not hold exactly.
In the following section, we show that while the pivot in Result \ref{prop: pivot} may not be exact, it yields asymptotically valid inference since it converges in distribution to a $\text{Uniform}[0,1]$ random variable.

We conclude this section with a remark on the variance parameter $\sigma_{n,1}^{2}$ that appears in the expression of the pivot.
In practice, this quantity is unknown, as it depends on $\pi_1$, as we derive later. 
However, we demonstrate in the subsequent theoretical development that replacing this parameter with a consistent estimator yields asymptotically valid inference.
The simulation results in Section \ref{sec:simulations} corroborate this theoretical result.

\section{Asymptotic theory for inference}
\label{sec:asymptotictheory}

We theoretically justify the pivotal quantity introduced in Section \ref{sec:inference}, showing that it converges to a uniform random variable.

\subsection{Weak limit for our pivot}
\label{subsec:weaklimit}

We begin by identifying a linear representation of our key statistics in Proposition \ref{prop: linrep}, expressing them as weighted sums of independent Bernoulli response variables.

For any region $R \subseteq \cX$ and $j \in [m]$, let $n_R = \lvert \cbracket{ i : X_{i,n} \in R } \rvert$ and $n_R^j = n_{R \cap l_j}$. 
Furthermore, for simplicity of notation, we let $p_h^k=p_h(s^k_h)$, $n_{h,0} = n_{p_h}$, $n_{h,k} = n_{p_h^k}$, and $n_{h,0}^j = n_{p_h \cap l_j}$, $n_{h,k}^j = n_{p_h^k \cap l_j}$.
In this section, we make the dependence of the model parameters on the sample size explicit by writing $\theta_i= \theta_{i,n}$, for $i\in [n]$, and $\pi_j = \pi_{j,n}$ for $j \in [m]$.

\begin{definition}
\label{defn: weights}
Define $a_{i,n} \in \mathbb{R}^{\bar{k}}$ with its first entry denoted by $a_{i,n,1}=\sqrt{\frac{n}{n_1}}\mathbbm{1}[X_{i,n}\in l_{1}]$ and the remaining components collected as $a_{i,n,\sminus 1}= \text{Vec}([a_{i,n,\sminus 1}^{h} \ a_{i,n,\sminus 1}^{h,2} \cdots a_{i,n,\sminus 1}^{h,k_{h}} ]^{\top} \in \mathbb{R}^{k_{h}}: \ h \in [d]) \in \mathbb{R}^{\bar{k}-1}$, where the entries are set to be
\vspace{-0.5cm}
\begin{equation*}
\begin{gathered}
\ a_{i,n,\sminus 1}^{h} =  \sqrt{\frac{n}{n_{h,0}}}\nbracket{\mathbbm{1}[X_{i,n}\in p_{h}]-\mathbbm{1}[X_{i,n}\in l_{1}]},\ 
a_{i,n,\sminus 1}^{h,k} =  \sqrt{\frac{n}{n_{h,k}}}\nbracket{\mathbbm{1}[X_{i,n}\in p_{h}^{k}]-\frac{n_{h,k}^{1}}{n_1}\mathbbm{1}[X_{i,n}\in l_{1}]}.
\end{gathered} 
\end{equation*}
\end{definition}

\begin{proposition}
\label{prop: linrep}
It holds that $\hat{T}_{n} = \frac{1}{\sqrt{n}} \sum_{i\in [n]} a_{i,n}Y_{i,n}$ where $a_{i,n}$ is as given in Definition \ref{defn: weights} and $\{Y_{i,n}\}_{i\in[n]}$ are independent $\text{\normalfont Bernoulli}(\theta_{i,n})$ variables, with 
\vspace{-0.5cm}
$$\theta_{i,n} = \sum_{j \in [m]} \pi_{j,n} \mathbbm{1}[X_{i,n} \in l_{j}].$$ 
\end{proposition}

Equipped with this linear representation, we apply the Lindeberg Central Limit Theorem to establish that, for a fixed tree fit, these key statistics are normally distributed, as described in Step 1 of our overview. This result is formalized in Proposition \ref{prop: normaldistn}, and the conditions required for its validity are given below. 

\begin{assumption}
\label{ass:1}
Assume that, for all observed regions in $\ct_1$ and for all candidate regions obtained by splitting an observed parent region, the number of samples grows proportionally with the total sample size $n$.
Notationally, we denote this asymptotic regime as $[\bn] \to \infty$ in our limits, i.e., 
\vspace{-0.5cm}
$$\lim _{[\bn] \to \infty} \frac{n}{n_{\cR}} < \infty,  \text{ for } \ \cR \in \cbracket{l_1, p_h^k, p_h\backslash p_h^k: h\in [d], \ k \in [k_h]}.$$
\end{assumption}

\begin{assumption}
\label{ass:2}
Assume that the proportion parameters $\nbracket{\pi_{1,n},\ldots,\pi_{m,n}}\in(0,1)^{m}$ are bounded away from 0 and 1, i.e.,
\vspace{-0.5cm}
$$\liminf_{[\bn] \to \infty} \pi_{j,n}>0; \quad  \limsup _{[\bn] \to \infty} \pi_{j,n}<1 \ , \ \forall j \in [m]$$
\end{assumption}

\begin{assumption}{}{[Uniform non-degeneracy and boundedness]}  
\label{ass:3}
The covariance matrix $\Sigma_n$ in the pre-conditional normal distribution of our key statistics is positive definite, and all its eigenvalues are bounded away from $0$ and $\infty$; i.e., there exist constants $b, B > 0$ such that 
 \vspace{-0.5cm}
$$
0<\frac{1}{b^2}<\lambda_{\min}(\Sigma_n)\leq\lambda_{\max}(\Sigma_n)<B^2<\infty
$$
\end{assumption}

\begin{remark}
\label{remark:ass:3}
Under Assumption~\ref{ass:3}, the norms of the square root and the inverse square root of the covariance matrix are uniformly bounded. Since $\Sigma_n$ is symmetric and positive definite, we have
$\|\Sigma_n^{1/2}\|^2=\lambda_{\max}(\Sigma_n)$ and
$\|\Sigma_n^{-1/2}\|^{-2}=\lambda_{\min}(\Sigma_n)$, so that we obtain
 \vspace{-0.5cm}
 \[
\|\Sigma_n^{1/2}\|<B<\infty, 
\qquad 
\|\Sigma_n^{-1/2}\|<b<\infty
\quad \forall n
\]
\end{remark}

\begin{proposition}
Suppose that the tree fit is specified a apriori and independent of the observed data. 
Under Assumptions \ref{ass:1}, \ref{ass:2} and \ref{ass:3}, it follows that 
\vspace{-0.5cm}
$$
\Sigma_{n}^{-\frac{1}{2}}\nbracket{\hat{T}_{n}-\Pi_{n}}\indist \mathcal{N}_{\bar{k}}\nbracket{0_{\bar{k}}
    ,I_{\bar{k}, \bar{k}}},
$$
where $\Pi_{n} := \mathbb{E} [\hat{T}_n] = \dfrac{1}{\sqrt{n}} \displaystyle\sum_{i \in [n]} a_{i,n} \theta_{i,n}$, and $\Sigma_{n} := \Cov (\hat{T}_n) = \dfrac{1}{n} \displaystyle\sum_{i \in [n]} \theta_{i,n}(1 - \theta_{i,n}) a_{i,n} a_{i,n}^{\top}$.    
\label{prop: normaldistn}
\end{proposition}

\begin{assumption}
\label{ass:infodevbound}
The function measuring information, $\cI(.):(0,1)\to \mathbb{R}$, is bounded and thrice differentiable, with all its first through third order derivatives also bounded uniformly over $(0,1)$.
\end{assumption}

We now state our main result, Theorem \ref{thm: asymptotic validity}, which shows that the pivot introduced informally in Step 3 of Section~\ref{subsec:inferenceroadmap} converges in distribution to a uniform random variable.

\begin{theorem}
Under Assumptions~\ref{ass:1},\ref{ass:2},\ref{ass:3} and \ref{ass:infodevbound}, we have
 \vspace{-0.5cm}
$$
 \text{\normalfont Pivot}(\hat{T}_{n,1}; \Pi_{n,1})\;\lvert \;  \cbracket{\bS_{1}=\bs_{1}} \indist \text{\normalfont Uniform}[0,1].
 \vspace{-1cm}
$$
\label{thm: asymptotic validity}
\end{theorem}

Theorem \ref{thm: asymptotic validity} establishes that the weak limit of the $\text{\normalfont Pivot}(\hat{T}_{n,1}; \Pi_{n,1})$ is a uniform random variable; in this sense, it serves as an asymptotic pivot, enabling valid tests and confidence intervals for $\Pi_{1,n}=\sqrt{n_{1}}\pi_{1,n}$.

In the next section, we provide the key results for proving Theorem \ref{thm: asymptotic validity} while deferring all remaining details to the Appendix.

\subsection{Key results and proof idea}
\label{subsec:proofidea}
\paragraph{Notations.} For a vector $v$, let $\norm{v}$ denote its $\ell_2$ norm, and for a matrix $A$, let $\norm{A}$ denote the spectral norm of $A$. For a function $\mathcal{A}(.): \mathbb{R}^{d}\to\mathbb{R}$, let $\|\nabla^{3} \mathcal{A}(v)\|$ denote the spectral trilinear operator norm $\displaystyle\sup_{\|x\|=\|y\|=\|z\|=1} \big|\langle \nabla^{3}\mathcal{A}(v), x\otimes y\otimes z\rangle\big|$, where $\nabla^{3}\mathcal{A}(v)$ is the third-order derivative tensor with entries $\partial_{ijk}\mathcal{A}(v)$ (unless specified otherwise).

\paragraph{Representations using standardized variables.} 
To develop the proof of Theorem \ref{thm: asymptotic validity}, we work with the standardized version of the statistic $\hat{T}_n$, defined as $\zeta_{n}:= \Sigma_{n}^{-1/2}\nbracket{\hat{T}_{n}-\Pi_{n}}$, where $\Sigma_{n}^{-1/2}$ is the inverse of the square root of the invertible matrix $\Sigma_{n}$. 

Two immediate observations hold for the standardized variable $\zeta_n$, which are summarized in Lemma \ref{lem: linear rep stan} and Lemma \ref{lem: pivot: stan}. The first result provides a linear representation of $\zeta_n$ as a weighted sum of independent standardized variables with mean $0$ and variance $1$. The second result presents the pivot, introduced in Result \ref{prop: pivot}, in terms of the standardized variable $\zeta_n$.

\begin{lemma}
Define $\tilde{Y}_{i,n}:= \dfrac{Y_{i,n}-\theta_{i,n}}{\sqrt{\theta_{i,n}(1-\theta_{i,n})}}$ and $b_{i,n} := \sqrt{\theta_{i,n}(1-\theta_{i,n})}\Sigma_{n}^{-1/2}a_{i,n}$. Let  $B_{(n)}:= (b_{1,n}, \ldots, b_{n,n})^\top$ and $ \tilde{Y}_{(n)} := \nbracket{\tilde{Y}_{1,n}, \ldots, \tilde{Y}_{n,n}}^\top$. Then it holds that:
\vspace{-0.5cm}
$$
\zeta_{n} = \dfrac{1}{\sqrt{n}} \sum_{i\in [n]} b_{i,n} \tilde{Y}_{i,n}=  \dfrac{1}{\sqrt{n}}B_{(n)}^\top \tilde{Y}_{(n)}.
\vspace{-1cm}
$$
\label{lem: linear rep stan} 
\end{lemma} 

\begin{definition}
\label{pivot: std}
We redefine the pivotal quantity of inferential interest in terms of the standardized variables as
\vspace{-0.5cm}
$$P(\zeta_{n}; \Pi_n)= \frac{\int_{-\infty}^{\sigma_{n,1} \zeta_{n,1}+\Pi_{n,1}}\phi(u;\Pi_{n,1},\sigma_{n,1}^{2})\prod_{h=1}^{d}\lambda_{h}(u,\Sigma_{ n,\sminus 1}^{1/2}\zeta_{n, \sminus 1} + \Pi_{n, \sminus 1} )  du}{\int_{-\infty}^{\infty}\phi(u;\Pi_{n,1},\sigma_{n,1}^{2})\prod_{h=1}^{d}\lambda_{h}(u,\Sigma_{ n,\sminus 1}^{1/2}\zeta_{n, \sminus 1} + \Pi_{n, \sminus 1})  du}.$$
For our theoretical analysis, we study $P(\zeta_{n}; \Pi_n)$ as a function of the data variables $\zeta_n$.
For notational convenience, we hereafter denote $P(\zeta_{n}; \Pi_n)$ by $P(\zeta_{n})$, suppressing its dependence on $\Pi_n$; that is, let $P(\zeta_{n}) := P(\zeta_{n}; \Pi_n)$.
\end{definition}

\begin{lemma}
We have the following equivalence:
$P(\zeta_{n})= \text{\normalfont Pivot}(\hat{T}_{n,1}; \Pi_{n,1})$.
\label{lem: pivot: stan} 
\end{lemma}

The proofs of Lemma \ref{lem: linear rep stan} and Lemma \ref{lem: pivot: stan} follow by straightforward algebra and are therefore omitted from the paper.

\paragraph{Defining our objective.}
We formally state our objective in order to prove Theorem \ref{thm: asymptotic validity}.

\begin{definition}
Define $Z_{n} = \dfrac{1}{\sqrt{n}} \sum_{i\in [n]} b_{i,n} Z_{i,n}$, where the standardized Bernoulli random variables $\tilde{Y}_{i,n}$ have been replaced by independent Gaussian variables $Z_{i,n}$ with mean $0$ and variance $1$.
If we let $\tilde{Z}_{(n)} := \nbracket{Z_{1,n}, \ldots, Z_{n,n}}^\top$, then we can equivalently express $Z_{n} = \dfrac{1}{\sqrt{n}}B_{(n)}^\top \tilde{Z}_{(n)}$. 
By definition, $Z_{n}$ follows a $\bar{k}$-dimensional standard normal distribution, i.e., 
\vspace{-0.5cm}
$$
Z_{n} \sim \mathcal{N}\nbracket{0_{\bar{k}}, I_{\bar{k}, \bar{k}}}. \vspace{-0.5cm}
$$

\end{definition}

Observe that the pivot $P(Z_n)$, when evaluated with the normal random variable $Z_n$, is distributed as a $\text{Uniform}(0,1)$ random variable. 
This follows from the observation that the key statistics constructed from $Z_n$, given by $\hat{U}_n := \Sigma_{n}^{1/2}Z_n + \Pi_n$, follow an exact normal distribution, as assumed in Result \ref{prop: cond}.

Our objective, which is to establish the weak limit of our pivot as a $\text{Uniform}(0,1)$ random variable, can be equivalently framed as follows: we aim to prove that
\vspace{-0.5cm}
\begin{equation}
\label{WeakConv}
\lim_{[\bn]  \to \infty} \Big \lvert \mathbb{E}  \left[h \circ P\left(\zeta_{n}\right) \mid \cbracket{\bS_{1}=\bs_{1}} \right]  -\mathbb{E}\left[h \circ P(Z_{n}) \mid\cbracket{\bS_{1}=\bs_{1}}\right] \Big \rvert=0, \vspace{-0.5cm}
\end{equation}
for every real-valued $h \in \Bigl\{\tilde{h} \in \mathbb{C}^3(\mathbb{R}): \; \sup_{x \in \mathbb{R}} \big| \grad^{(k)}\tilde{h}(x) \big| < \infty \;\; \text{for } k = 0,1,2,3 \Bigr\}$. Here,
\vspace{-0.5cm}
$$\mathbb{E}\left[h \circ P(Z_{n}) \mid\cbracket{\bS_{1}=\bs_{1}}\right] = \int_{0}^{1} h(u) du, \vspace{-0.5cm}$$ 
and $[\bn] \to \infty$ denotes our asymptotic regime in which the samples in the tree fit grow proportionally, as specified in Assumption \ref{ass:1}.

\paragraph{Sufficient conditions for establishing weak limit of pivot.} 

Below, we redefine the sampling probabilities in the proposed exponential mechanism in terms of the standardized variables.

\begin{definition}
\label{defn: Lam}
Consider the gain functions over domain $\tilde{D}_n$ from Definition~\ref{defn: g}, then define function $f:\tilde{D}_{n}\to \mathbb{R}$ as
\vspace{-0.5cm}
$$
f(t) := \sum_{h=1}^{d} \rbracket{\epsilon_{h}^{-1} \cG^{1}_{h}(t) - \log \sum_{k=1}^{k_{h}} \exp \epsilon_{h}^{-1} \cG^{k}_{h}(t)}.
\vspace{-0.5cm}
$$ Define a function $\Lambda: D_n \to (0,1)$ as:
\vspace{-0.5cm}
\begin{align*}
\Lambda(\zeta) :=  \exp f(\Sigma_{n}^{1/2}\zeta+\Pi_{n}), \text{ where } D_n:= \cbracket{\zeta \in \mathbb{R}^{\bar{k}}: \Sigma_{n}^{1/2}\zeta+\Pi_{n} \in \tilde{D}_{n}}\subset \mathbb{R}^{\bar{k}}.
\label{rel:lambda:f}
\end{align*}
\end{definition} Based on these definitions, check that: 
\vspace{-0.5cm}
\begin{equation*}
\log \Lambda(\zeta) = f(\Sigma_{n}^{1/2}\zeta+\Pi_{n})=\sum_{h=1}^{d}\log \lambda_{h} (\Sigma_{n}^{1/2}\zeta+\Pi_{n}).\vspace{-0.5cm}
\label{rel:lambda:f}
\end{equation*} and observe $\Lambda(\zeta_n)$ calculates the selection weight from the master statistics
\vspace{-0.5cm}
$$\Lambda(\zeta)= \prod_{h=1}^{d} \lambda_h\nbracket{\Sigma_{n}^{1/2}\zeta+\Pi_{n}}=\mathbb{P}\rbracket{S_{h}^{1}=s_{h}^{1}\mid \hat{T}_{n}=\Sigma_{n}^{1/2}\zeta+\Pi_{n}}$$

\begin{remark}
Note that the samples proportions in any subregion of the predictor space take values in $(0,1)$ with probability $1$, which implies that $\hat{T}_n = \rbracket{\hat{T}_{n,1} \ \hat{T}_{n,\sminus 1}}^\top \in \tilde{D}_n$ almost surely.
However, for the reconstructed statistics $\hat{U}_n$ from normal variables, this may not hold. 
We therefore extend the functions $\cG_h^k$ and $f$ outside $\tilde{D}_n$ to take the value 0, and, consequently, extend $\lambda_h$, $\Lambda$, and $P$ to take the value 0 outside $D_n$.
To avoid notational clutter, however, we continue to use the same symbols for the extended functions everywhere.
\end{remark}

\begin{proposition}{}{(Relative Differences)}
Define the difference terms
\vspace{-0.5cm}
\begin{align*}
\begin{gathered}
\RD_n^{(1)}=\frac{\left|\mathbb{E}\left[\Lambda(\zeta_n)\right]-\mathbb{E}\left[\Lambda(Z_n) \right]\right|}{\mathbb{E}\left[\Lambda(Z_{n})\right]} \\
\RD_n^{(2)}=\frac{\left|\mathbb{E}\left[h \circ P\left(\zeta_{n}\right) \times \Lambda(\zeta_n) \right]-\mathbb{E}\left[h \circ P(Z_{n}) \times \Lambda(Z_n) \right]\right|}{\mathbb{E}\left[\Lambda(Z_n)\right]}.
\vspace{-0.5cm}
\end{gathered}
\end{align*} 
where $h\in \mathbb{C}^3(\mathbb{R})$ and $P$,$\Lambda$ are as defined in Definitions \ref{pivot: std} and \ref{defn: Lam}.
Then, if 
\vspace{-0.5cm}
$$
\lim_{[\bn]  \to \infty} \RD_n^{(1)}=0, \quad \lim_{[\bn]  \to \infty}  \RD_n^{(2)}=0,
\vspace{-0.5cm}
$$
the weak convergence in \eqref{WeakConv} holds.
\label{prop: relativediff}
\end{proposition}

Such sufficient conditions are commonly used to establish asymptotic guarantees for conditional inference (see, e.g., \cite{panigrahi2023carving, bakshi2024inference}). However, as will be evident from our proof technique, our approach to verifying the sufficient conditions in Proposition~\ref{prop: relativediff} is fundamentally different from prior work: we directly exploit in our work properties of the exponential mechanism in~\eqref{exp:mechanism} to obtain these conditions.

\paragraph{Verifying the sufficient conditions.}
We now turn to verifying the sufficient conditions in Proposition \ref{prop: relativediff}, which form the core of our proof of Theorem \ref{thm: asymptotic validity}.
The supporting results used in its proof are provided in Appendix~\ref{appendix:supporting}.

\begin{theorem}
   Under Assumptions~\ref{ass:1}, \ref{ass:2}, \ref{ass:3}, and \ref{ass:infodevbound}, we have 
   \vspace{-0.5cm}
   $$
   \lim_{[\bn]  \to \infty} \RD_n^{(1)}=0, \quad \lim_{[\bn]  \to \infty}  \RD_n^{(2)}=0. 
   \vspace{-0.5cm}
   $$ 
   \label{thm: verifying suff cond}
\end{theorem}

\begin{proof}
First, we derive a lower bound for $\mathbb{E}\left[\Lambda(Z_{n})\right]$ in the denominator of the difference terms $\RD_n^{(j)}$ for $j \in \{1,2\}$, then derive an upper bound for the numerators in these difference terms, and finally show that the resulting ratios converge to $0$.

\noindent{\textbf{Bounding the denominator}}. Observe that 
\vspace{-0.5cm}
 \begin{equation*}
 \begin{aligned}
 \mathbb{E}\left[\Lambda(Z_{n})\right] &= \mathbb{E}\left[\Lambda(Z_{n})\times \mathbbm{1}_{D_{n}}(Z_{n})\right] \\  &= \mathbb{E}\left[\exp \left(f\nbracket{\Sigma_{n}^{1/2}Z_{n}+\Pi_{n}}\right)\times \mathbbm{1}_{D_{n}}(Z_{n})\right]    \\
 &= \exp f\nbracket{\Pi_{n}}\times \mathbb{E}\left[\exp \nbracket{ f\nbracket{\Sigma_{n}^{1/2}Z_{n}+\Pi_{n}}-f\nbracket{\Pi_{n}}}\times \mathbbm{1}_{D_{n}}(Z_{n})\right]  \\
 &\geq \exp{f\nbracket{\Pi_{n}}}\times \mathbb{E}\rbracket{\exp (-L\norm{\Sigma_{n}^{1/2}Z_{n}})\times \mathbbm{1}_{D_{n}}(Z_{n})},
 \vspace{-0.5cm}
  \end{aligned}
  \end{equation*}
 where in the last display we use the fact that the log-probability from the exponential mechanism, $f(\cdot)$, is Lipschitz with constant $L$, as shown in Proposition~\ref{prop:supporting1}.

 Now, using the bound on $\norm{\Sigma_{n}^{1/2}}$ from Remark~\ref{remark:ass:3} (which follows from Assumption~\ref{ass:3}), we obtain
 \vspace{-0.5cm}
 $$
 \mathbb{E}\rbracket{\exp -L\norm{\Sigma_{n}^{1/2}Z_{n}}\times \mathbbm{1}_{D_{n}}(Z_{n})}\geq \mathbb{E}\rbracket{\exp (-LB\norm{Z_{n}})\times \mathbbm{1}_{D_{n}}(Z_{n})}.
 \vspace{-0.5cm}
 $$
Finally, from Lemma~\ref{lemma:supporting2}, we obtain that, for sufficiently large $n$,
 \vspace{-0.5cm}
 $$\mathbb{E}\rbracket{\exp (-LB\norm{Z_{n}})\times \mathbbm{1}_{D_{n}}(Z_{n})} \geq C_{-},
  \vspace{-0.5cm}
  $$ 
 and therefore
  \vspace{-0.5cm}
 \begin{equation}
\mathbb{E}\left[\Lambda(Z_{n})\right] \geq C_{-}\exp{f\nbracket{\Pi_{n}}}. 
 \label{lower:bdd}
 \end{equation}

\noindent{\textbf{Bounding the numerators}}. 
Let $\Psi^{(1)}(z)= \Lambda(z), \ \Psi^{(2)}(z)= h \circ P\left(z\right) \times \Lambda(z)$. 
Furthermore, we define
 \vspace{-1cm}
\begin{align*}
\begin{gathered}
U_{(i)}:= \nbracket{\tilde{Y}_{1,n}, \ldots, \tilde{Y}_{i,n}, Z_{i+1,n}, \ldots, Z_{n,n}}^\top,\
V_{(i)}:= \nbracket{\tilde{Y}_{1,n}, \ldots, \tilde{Y}_{i-1,n},0, Z_{i+1,n}, \ldots, Z_{n,n}}^\top,
 \vspace{-0.5cm}
\end{gathered}
\end{align*}
where $\{Z_{i,n}: i\in [n]\}$ and $\{\tilde{Y}_{i,n}: i\in [n]\}$ are as defined earlier. An application of Lemma~\ref{lemma:lindeberg} yields the upper bound 
\begin{equation*}
\begin{aligned}
    &\left |\Ee\rbracket{\Psi^{(j)}\nbracket{\frac{1}{\sqrt{n}}B_{(n)}^\top\tilde{Y}_{(n)}}}-  \Ee\rbracket{\Psi^{(j)}\nbracket{\frac{1}{\sqrt{n}}B_{(n)}^\top \tilde{Z}_{(n)}}}\right|\\
    & \;\;\;\;\;\;\;\;\;\;\;\;\;\;\;\;\;\;\;\;\;\;\;\;\;\;\;\;\;\;\leq 2 \frac{\|B_{(n)}\|_{\infty}^3}{\sqrt{n}}\sup_{\alpha\in(0,1)}\max_{i\in [n]}\Ee\!\Big[\norm{\nabla^3 \Psi^{(j)}\big(\tfrac{1}{\sqrt{n}} B_{(n)}^\top \big((1-\alpha) V_{(i)}+\alpha U_{(i)}\big)\big)}\Big].
\end{aligned}
\end{equation*} 
for $j\in \{1,2\}$. 
Now, using Proposition~\ref{prop:supporting2}, it follows that
 \vspace{-0.5cm}
$$
\max_{i\in [n]}\Ee\!\Big[\norm{\nabla^3 \Psi^{(j)}\big(\tfrac{1}{\sqrt{n}} B_{(n)}^\top \big((1-\alpha) V_{(i)}+\alpha U_{(i)}\big)\big)}\Big]\leq C_j\times \exp f(\Pi_n).
 \vspace{-0.5cm}
 $$
Moreover, since $\|B_{(n)}\|_{\infty}\leq \bar{C}_{+}$ for all $n$ (by Proposition~\ref{prop:supporting0}), we have that, for both $j\in\{1,2\}$, 
 \vspace{-0.5cm}
 \begin{equation}
\left |\Ee\rbracket{\Psi^{(j)}\nbracket{\frac{1}{\sqrt{n}}B_{(n)}^\top\tilde{Y}_{(n)}}}-  \Ee\rbracket{\Psi^{(j)}\nbracket{\frac{1}{\sqrt{n}}B_{(n)}^\top \tilde{Z}_{(n)}}}\right|\leq C_{+} \times \frac{\exp f(\Pi_n)}{\sqrt{n}}, \vspace{-0.5cm}
 \label{upper:bdd}
 \end{equation}
where $C_{+}= \max{\{C_1, C_2\}}\times 2\bar{C}_{+}$.

\noindent{\textbf{Combining the two bounds}}. Our claim for the ratios in the difference terms follows immediately from \eqref{lower:bdd} and \eqref{upper:bdd}.
\end{proof}

\subsection{Weak limit of pivot with consistent plug-in for variance}
\label{sec:plugin-validity}

Recall that the pivot introduced in Proposition~\ref{prop: pivot}
\vspace{-0.5cm}
\[
\mathrm{Pivot}(\hat T_{n,1};\Pi_{n,1},\sigma_{n,1}^{2})
:= \mathrm{Pivot}(\hat T_{n,1};\Pi_{n,1})=
\frac{\displaystyle \int_{-\infty}^{\hat T_{n,1}}
\phi(u;\Pi_{n,1},\sigma_{n,1}^{2})
\prod_{h=1}^{d}\lambda_{h}(u,\hat T_{n,-1})\,du}
{\displaystyle \int_{-\infty}^{\infty}
\phi(u;\Pi_{n,1},\sigma_{n,1}^{2})
\prod_{h=1}^{d}\lambda_{h}(u,\hat T_{n,-1})\,du},
\]
involves the variance parameter $\sigma_{n,1}^{2} = \pi_{1,n}(1 - \pi_{1,n})$. 
In practice, this parameter is unknown and must be replaced by an estimate $\hat{\sigma}_{n,1}^{2}$. 

To complete our inferential framework, Theorem \ref{thm:plugin} shows that the plug-in version of our pivot, $\mathrm{Pivot}(\hat T_{n,1};\Pi_{n,1}, \hat\sigma_{n,1}^{2})$ remains asymptotically valid as long as a consistent estimator is used for the unknown variance. 
In simulations, we find that the simple plug-in estimate $$\hat{\sigma}_{n,1}^{2} = \hat{\pi}_{1}(1 - \hat{\pi}_{1})= \dfrac{\hat T_{n,1}}{\sqrt{n_1}}\nbracket{1-\dfrac{\hat T_{n,1}}{\sqrt{n_1}}},$$ using the sample proportion in $l_1$, provides valid inference.

\begin{theorem}
\label{thm:plugin}
Let $\hat\sigma_{n,1}^{2} \to \sigma_{n,1}^{2}$ in probability, without conditioning on the fitted tree.  
Then, under Assumptions~\ref{ass:1}, \ref{ass:2}, \ref{ass:3}, \ref{ass:infodevbound} as stated in Theorem~\ref{thm: asymptotic validity}, we have the following.
\begin{enumerate}[leftmargin=*]
    \item[(i)] $\hat\sigma_{n,1}^{2}$ is a consistent estimator of $\sigma_{n,1}^{2}$ after conditioning on $\{\bS_1=\bs_1\}$, i.e., for any $\delta>0$,
    \vspace{-0.5cm}
    $$
    \lim_{[\bn]  \to \infty}\mathbb{P}\!\left(\,|\hat\sigma_{n,1}^2-\sigma_{n,1}^2|>\delta\ \big|\ \bS_1=\bs_1\,\right) = 0.
    \vspace{-0.5cm}
    $$
    \item[(ii)] Moreover, the pivot with the consistent plug-in estimator, $\mathrm{Pivot}(\hat T_{n,1};\Pi_{n,1}, \hat\sigma_{n,1}^{2})$, remains asymptotically valid conditional on the sequence of greedy splits leading to leaf $l_1$, i.e.,
    \vspace{-0.5cm}
    $$
     \mathrm{Pivot}(\hat T_{n,1};\Pi_{n,1}, \hat\sigma_{n,1}^{2}) \;\lvert \;  \cbracket{\bS_{1}=\bs_{1}}\;\indist\; \mathrm{Uniform}[0,1].
    $$
\end{enumerate}
\end{theorem}

The proof of Theorem \ref{thm:plugin} is provided in Appendix~\ref{appendix:proofs5}.

\section{Adaptive choices of the temperature parameter}
\label{sec:adaptivetemp}

The temperature parameter $\epsilon_h$ in Algorithm~\ref{alg:RCT} controls the degree of randomization used to select the split on $p_h$, and has so far been fixed at a common value $\epsilon$. In this section, we describe a data-adaptive strategy for choosing the amount of randomization at each internal node. This yields a variant of Algorithm~\ref{alg:RCT} in which the temperature parameter is selected adaptively. For inference under this choice, the pivot requires a slight modification of that introduced in Proposition~\ref{prop: pivot}. We outline the resulting approach below.

For each $h \in [d]$, we define temperature functions with a fixed hyperparameter $\tau>0$ as:
\begin{align}
\label{eq:dataadaptivetemp}
 \; \overline{\epsilon}_h(t) 
\;:=\;
\tau\times\frac{1}{k_h}\sum_{k=1}^{k_h}\cG_h^k(t), \qquad
\forall t\in\tilde D_n,
\end{align} and set the data-adaptive temperature parameter at node $p_h$ to be $\overline{\epsilon}_h(\hat{T}_{n})$. Recall that $\cG^{k}_{h}(\hat{T}_{n}) =  \cG^{k}_{h}(\hat{T}_{n,1}, \hat{T}_{n,\sminus 1})= G_{p_h}(s_h^k; Y)$ (see Definition \ref{defn: g}) denotes the information gain at $p_h$ computed from $\hat{T}_{n}$.
Hereafter, we refer to $\tau$ as the \emph{temperature scale} parameter, which globally rescales the node-specific temperature parameters, and controls the overall level of randomization across all splits; smaller values of $\tau$ correspond to lower randomization. 


Accounting for the tree fit obtained with this data-driven temperature parameter requires only a minor modifications to our pivot. 
For each $h \in [d]$,  define the correction factor  $\overline{\lambda}_{h}:\tilde{D}_n\to(0,1)$  as:
\begin{align}
\label{lambda:adjusted}
 \overline{\lambda}_{h}(t)
:= 
\dfrac{\exp \!\left(\overline{\epsilon}_{h}(t)^{-1}\,\cG^{1}_{h}(t)\right)}
      {\sum_{k=1}^{k_{h}} 
        \exp \!\left(\overline{\epsilon}_{h}(t)^{-1}\,\cG^{k}_{h}(t)\right)}
\end{align} and note that the sampling probabilities of the random split variable at each parent $p_h$ takes the form $\mathbb{P}\!\rbracket{S_{h}^{1}=s_{h}^{1}\mid \hat{T}_{n}}
= 
\overline{\lambda}_{h}(\hat{T}_{n})$. Therefore, it follows that $\mathbb{P}\!\left[\bS_{1}=\bs_{1}\mid \hat{T}_{n}\right]
= 
\prod_{h=1}^{d}
\overline{\lambda}_{h}(\hat{T}_{n})
$, and the pivot accounting for the tree fit—using the same steps of the proof for Proposition~\ref{prop: pivot}—becomes 
\begin{align}
\label{pivot:adjusted}
    \overline{\text{Pivot}}(\hat{T}_{n,1}; \Pi_{n,1})
    =
    \frac{
        \int_{-\infty}^{\hat{T}_{n,1}}
        \phi(u;\Pi_{n,1},\sigma_{n,1}^{2})
        \prod_{h=1}^{d}\overline{\lambda}_{h}(u,\hat{T}_{n,\sminus 1})\, du
    }{
        \int_{-\infty}^{\infty}
        \phi(u;\Pi_{n,1},\sigma_{n,1}^{2})
        \prod_{h=1}^{d}\overline{\lambda}_{h}(u,\hat{T}_{n,\sminus 1})\, du
    }.
\end{align}
Notably, the pivot under the new sampling scheme retains the same form as the pivot described in Proposition~\ref{prop: pivot} and achieves the same asymptotic guarantee as in Theorem~\ref{thm: asymptotic validity}, provided the temperature parameters are bounded away from zero.

\begin{assumption}
\label{ass:extra}
For the set of temperature functions $\{\overline{\epsilon}_h:{h\in[d]}\}$ with temperature scale $\tau>0$, there exists a non-negative constant $m_\tau>0$ such that
$\overline{\epsilon}_h(v)\;\ge\; m_\tau$
for all  $v\in\tilde D_n$  and all $h\in[d]$.
\end{assumption}

\begin{theorem}
Under Assumptions~\ref{ass:1},\ref{ass:2},\ref{ass:3},\ref{ass:infodevbound} and \ref{ass:extra}, we have
 \vspace{-0.5cm}
$$
 \overline{\text{\normalfont Pivot}}(\hat{T}_{n,1}; \Pi_{n,1})\;\lvert \;  \cbracket{\bS_{1}=\bs_{1}} \indist \text{\normalfont Uniform}[0,1].
 \vspace{-1cm}
$$
\label{thm: adjusted asymptotic validity}
\end{theorem}

Details for the proof and further discussion on data-adaptive temperatures are given in Appendix~\ref{appendix:dataadaptivetemp}. While we focus on a specific data-adaptive temperature parameter, our inference method can accommodate other data-driven choices as well, provided they satisfy Assumption~\ref{ass:temp_regular}. 




\section{Simulation studies}
\label{sec:simulations}

In this section, we present simulations demonstrating that our method delivers valid, powerful inference without sacrificing predictive accuracy, and we compare its performance with existing alternatives.

\subsection{Setting, methods and metrics}  We generate independent binary responses $Y_i\stackrel{\text{ind}}{\sim} \text{Bernoulli}(\theta_i)$ with
\begin{equation}
\label{gen:model}
\theta_i
=
\sigma\!\nbracket{
m \nbracket{
s\,X_{i1}
-
s\,X_{i2}
+
0.3\,s\,X_{i1}X_{i2}}},
\end{equation}
where $X_{i1}, X_{i2} \sim \text{Uniform}(-1,1)$ are i.i.d covariates and $\sigma(t) = (1+e^{-t})^{-1}$ is the logistic sigmoid function. The parameter $s>0$ controls signal strength: larger values of $s$ increase class separation by amplifying both the main effects and their interaction. Specifically, we consider $s \in \{1,2,3\}$, with a fixed margin parameter $m=0.5$ to avoid near-deterministic separation even at higher signal levels. This design yields a correctly specified, low-dimensional logistic model with moderate interaction.

We compare inference procedures for the mean parameters associated with the terminal regions of depth-3 trees, using an information criterion based on the \emph{gini impurity} measure:
\vspace{-0.5cm}
$$
I(\cP; Y)= \cI(\hat{\pi}_{\cP})=2\widehat{\pi}_{\cP}(1 - \widehat{\pi}_{\cP}),
$$
and evaluate average coverage, interval length, and predictive accuracy of results generated by different procedures.
Each Monte Carlo iteration generates $n=500$ observations, from which a fixed subset of $100$ observations is held out as \emph{test} data used solely to compute the log-loss. 
This test set is shared across all methods. 
The remaining $400$ observations constitute the \emph{train+inference} data, used to fit trees and to construct confidence intervals for the model parameters in the algorithmic model.
Across all methods, inference involves constructing $100\times (1-\alpha)\%$ confidence intervals for the mean parameters in the terminal regions or leaves of the tree fit, with $\alpha = 0.1$.
As noted earlier in Remark \ref{rem:interpretation}, the target of inference is a well-defined quantity: it corresponds to population mean of the subgroup of observations that fall into that leaf region, and can be computed directly in simulations.

We compare the following three strategies for training the algorithmic tree model and performing inference on $400$ observations.
\begin{enumerate}[leftmargin=*]
    \item[(i)] \textit{Naive}. Fit a standard depth-3 CART to all $400$ observations and perform inference on the same data, ignoring the adaptivity in the tree structure of the model.
    \item[(ii)] \textit{Data Splitting} (DS). Partition the $400$ observations into train and inference subsets according to a prefixed split ratio. The depth-3 CART is formed on the train subset of the data, and naive inference is conducted on the inference subset of the data.
    \item[(iii)] \textit{Randomized classification trees using exponential mechanism} (RCT). Fit a depth-3 RCT using the temperature-controlled sampling scheme based on the exponential mechanism on all $n=400$ observations, and then perform inference using the pivot, accounting for the adaptive tree fit, on the full dataset. Note that both model training and inference are conducted without further splitting the dataset with $400$ observations into separate train and inference subsets.
    
    For the simulations, we implement the variant of Algorithm~\ref{alg:RCT} introduced in Section~\ref{sec:adaptivetemp}, in which the temperature parameter is chosen according to Equation~\ref{eq:dataadaptivetemp}, and inference is performed using the pivot in Equation~\ref{pivot:adjusted}.
\end{enumerate}

Each replication in our simulation experiments records:
\begin{enumerate}[label=(\arabic*)]
    \item Empirical coverage of the nominal $(1-\alpha)$ intervals.  
    For a tree with $m$ leaves, the coverage rate is defined as $\dfrac{\left|\{j \in [m] : \pi_j \in \hat{C}_j(Y)\}\right|}{m}$,
    where $\hat{C}_j(Y)$ denotes the confidence interval for leaf parameter $\pi_j$.
    \item Average confidence interval length. This metric reflects the inferential power of the algorithmic model: shorter confidence intervals indicate greater power.
        \item Held-out log-loss evaluated on test data, computed on the shared validation set of 100 observations, as  
        $$-\frac{1}{n_{\text{Test}}}\sum_{i\in \text{Test-data}}\{ y_i \log(\hat{p}_i) + (1-y_i)\log(1-\hat{p}_i)\},$$
       where $\hat{p}_i$ denotes the estimated probability that the outcome equals $1$.
This is a measure of the predictive accuracy of the algorithmic model or the adaptive tree fit.
\end{enumerate}

Results are aggregated over $200$ replications and visualized using boxplots, with triangles indicating the mean of each metric when computed across all replications.

\subsection{Experiments and simulation results}

We design three sets of experiments and report our findings for each set. 

In the first experiment, we vary the overall temperature parameter $\tau$ in our proposed method (RCT), while fixing the signal strength to a moderate value $s = 2$, to examine the impact of randomization on coverage, inferential power, and predictive accuracy.
The results are reported in Figure \ref{fig:vary_tau}.
We note that across all values of the temperature scale, which controls the amount of randomization, inferential validity is maintained. 
We refer to these three levels of randomization as RCT(1), RCT(2), and RCT(3) (ordered lowest to highest degree of randomization) using $\tau=10,15,20$ respectively.
In comparison, the naive method produces invalid inference, with coverage rates well below the desired target of $0.90$. 
Data splitting, denoted as DS, performed with $30\%$ of the data reserved for inference, yields valid inference but results in much longer confidence intervals, and the model trained on only $70\%$ of the data exhibits substantially worse predictive accuracy compared to our method. 
Most notably, our randomized method even achieves slightly better predictive performance than the naive method, which we attribute to improved stability and reduced generalization error due to the added external randomization.

In an additional experiment, reported in Appendix~\ref{appendix: ensembleRCT}, we evaluate the predictive performance of an ensemble version of RCT, formed by aggregating multiple trees constructed using the proposed splitting scheme, with predictions obtained by majority vote. From a purely predictive standpoint, the performance of ensemble RCT method remains comparable to that of the standard random forest. Notably, with a smaller number of trees, our ensemble method achieves improved accuracy relative to the conventional bootstrap-aggregated random forest. Further details on this experiment are provided in Appendix~\ref{appendix: ensembleRCT}, and more investigation on the predictive properties of the new random forest method is deferred to future work.

\begin{figure}[h]
    \centering
    \includegraphics[width=1\textwidth]{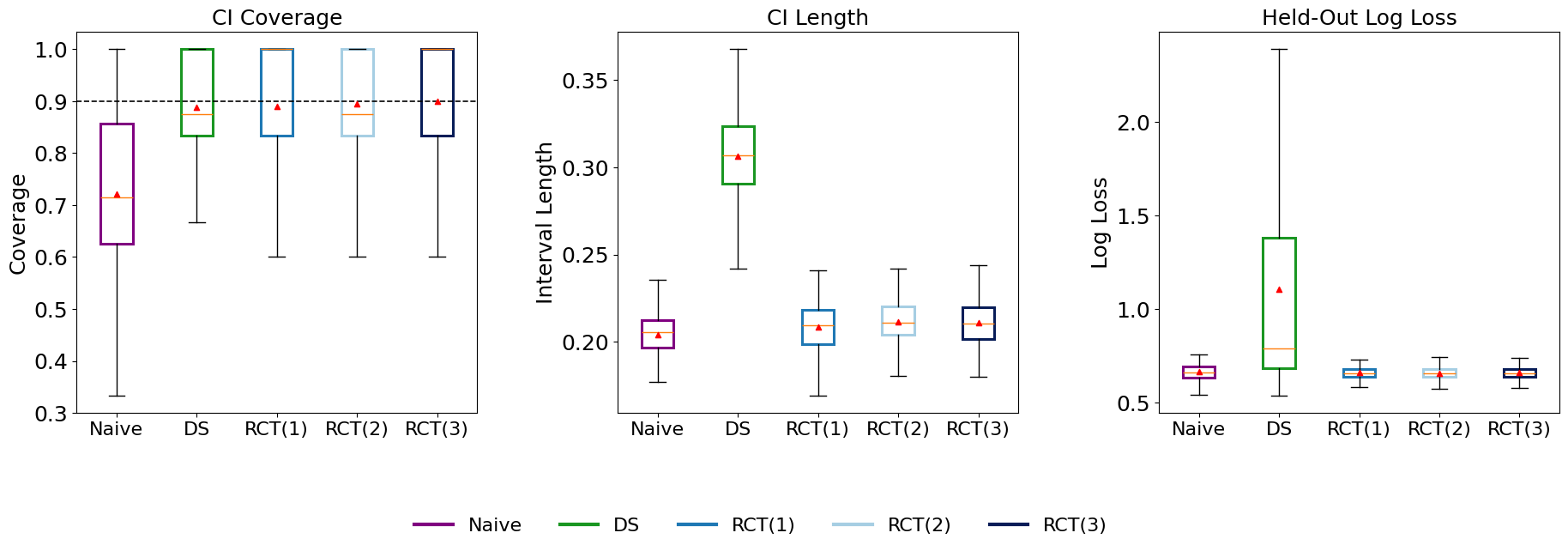}
    \vspace{-0.5cm}
    \caption{Effect of overall temperature on coverage, interval length, and held-out log-loss. RCT(1), RCT(2), RCT(3) represent RCT with $\tau = 10,15,20$ respectively.}
    \label{fig:vary_tau}
\end{figure}

The second experiment varies the data-splitting proportion and compares it with RCT(1) under fixed moderate signal strength $s=2$.
In particular, we vary the \emph{inference fraction} parameter in data splitting (DS) over the set $\{0.5, 0.4, 0.3, 0.2\}$, where each value indicates the proportion of data reserved for inference, with the remaining data used to train the algorithmic model.
As shown in Figure \ref{fig:vary_split}, even when DS uses $80\%$ of the $400$ observations to train the tree, its predictive accuracy is substantially worse than that of the fit produced by RCT. 
Conversely, DS reserving $50\%$ of the samples for inference still results in much wider intervals than RCT.
Overall, the trade-off between predictive accuracy and inferential power inherent in data splitting is suboptimal compared to our proposed method, which achieves both high predictive accuracy and valid inference simultaneously.

\begin{figure}[h]
    \centering
    \includegraphics[width=1\textwidth]{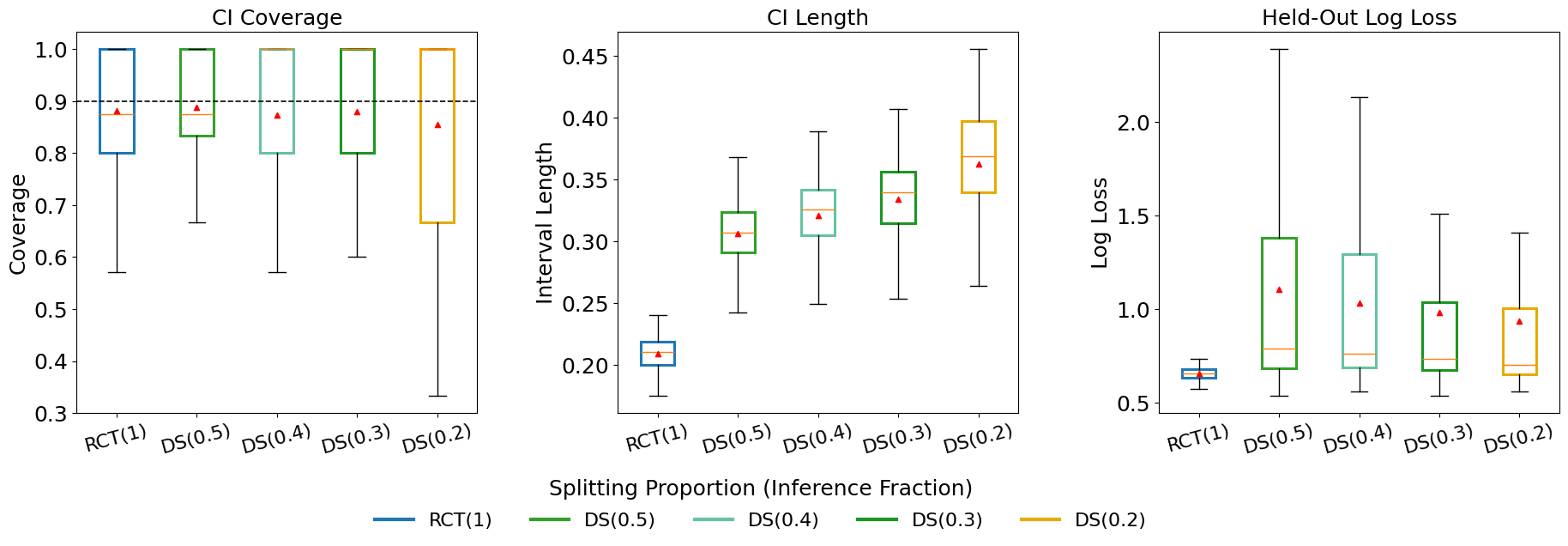}
    \vspace{-0.5cm}
    \caption{Comparison with data splitting at varying inference fraction s: coverage, interval length, and held-out log-loss}
    \label{fig:vary_split}
\end{figure}

\begin{figure}[h]
    \centering
    \includegraphics[width=1\textwidth]{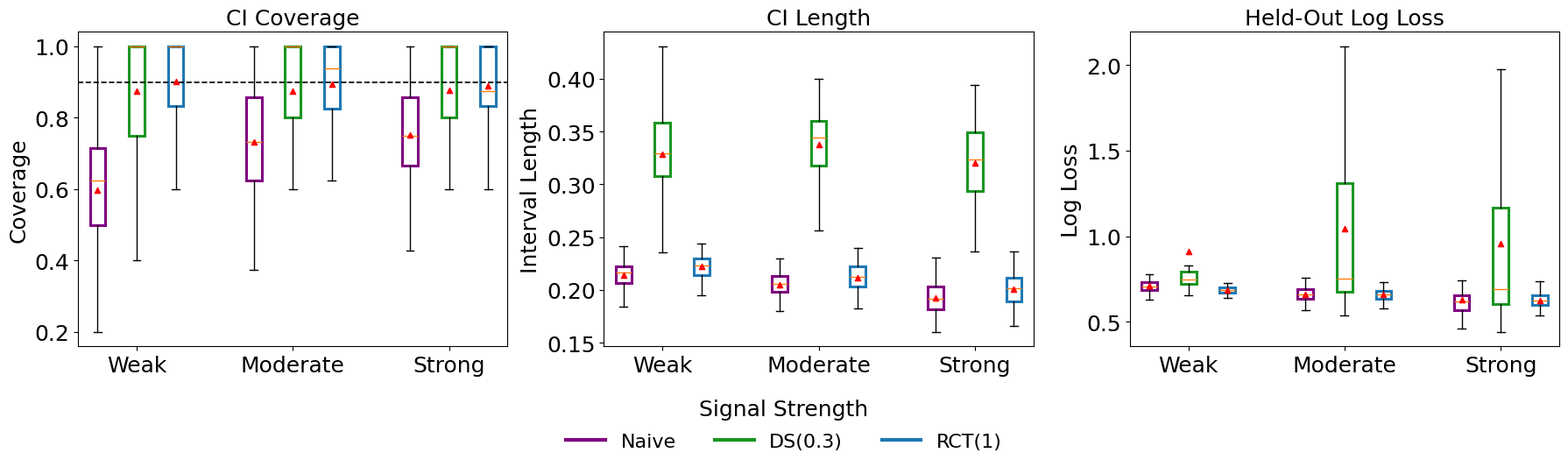}
    \vspace{-0.5cm}
    \caption{Effect of signal strength on CI quality and predictive performance.}
    \label{fig:vary_signal}
\end{figure}

In the third experiment, we vary the signal strength with $s \in \{1,2,3\}$ in~\eqref{gen:model} (representing weak, moderate and strong signal respectively). The previously observed patterns remain consistent, as illustrated in Figure \ref{fig:vary_signal}.
Naive intervals exhibit severe undercoverage, while RCT produces valid inference across all regimes. RCT(1) not only generates substantially shorter intervals than DS conducted with $30\%$ reserved for inference, but also attains better predictive accuracy on the held-out test data.

\section{Real data applications}
\label{sec:dataapps}


We evaluate our method on two publicly available datasets. Each dataset is split once into a training set (70\%) and a held-out test set (30\%). The training set is used for tree fitting and, when applicable, post-fit inference, while the test set is used solely to assess predictive performance; all methods share the same test set.

We compare three approaches: (i) \emph{CART (full)}, which fits a deterministic classification tree on the full training data and serves solely as a predictive baseline; (ii) \emph{DS (data splitting)}, which partitions the training data into fitting and inference subsets, with a total proportion $\rho$ of samples allocated to inference, and constructs confidence intervals using the inference subset; and (iii) \emph{RCT (proposed method)}, which fits a randomized classification tree on the full training data using our randomization scheme, and provides confidence intervals conditional on the realized tree, with temperature scale \(\tau\) controlling the amount of randomization. For DS and RCT, we vary \(\rho\) and \(\tau\), respectively, to illustrate the trade-off between predictive accuracy and inferential efficiency. All methods use Gini gain as the split criterion and fix the maximum tree depth at 3. 

We report two metrics: (a) predictive performance, measured by test-set classification accuracy, and (b) inferential efficiency, measured by the average confidence-interval length for leaf-level class probabilities \(\pi_\ell = \mathbb{P}(Y=1 \mid X \in \ell)\), averaged over terminal leaves. Confidence intervals are reported only for methods that provide valid post-fit inference.

We consider two datasets from distinct domains: a moderate-sized chemoinformatics application presented here and a larger-scale clinical application deferred to Appendix~\ref{app:diabetes}.

\noindent\textbf{Predicting and estimating biodegradability.} \ We consider the Quantitative Structure--Activity Relationship (QSAR) Biodegradation dataset \citep{uci_qsar_biodeg}, which contains $1055$ chemical compounds characterized by $41$ molecular descriptors and a binary indicator of \emph{ready biodegradability}. After preprocessing (see details in Appendix~\ref{app:QSAR_preproc}), the dataset includes $356$ ready biodegradable and $699$ not ready biodegradable compounds. The descriptors capture structural and physicochemical properties (e.g., spectral moments, functional group counts, and topological indices) and are measured to understand QSAR modeling of biodegradability in environmental chemoinformatics. Decision trees are appealing in this setting because they provide interpretable rules linking molecular structure to biodegradability, while naturally capturing nonlinearities and interactions. Beyond prediction, interval estimates for leaf-level probabilities are scientifically valuable, as they provide uncertainty quantification for biodegradability rates within the resulting terminal regions. The results from the three approaches are summarized in Table~\ref{tab:qsar_results}.

\begin{table}[ht]
\centering
\label{tab:qsar_results}
\begin{tabular}{lcc}
\toprule
Method & Test Acc. & Avg. CI Length \\
\midrule
CART (full)  & \textbf{0.833} & -- \\
DS ($\rho = 0.1$)  & \textcolor{red}{\textbf{0.789}} & \textcolor{red}{\textbf{0.452}} \\
DS ($\rho = 0.3$) & 0.726 & 0.233 \\
DS ($\rho = 0.5$)  & 0.770 & 0.181 \\
DS ($\rho = 0.7$) & 0.710 & 0.122 \\
RCT ($\tau=0.001$) & \textcolor{cyan}{\textbf{0.833}} & \textcolor{cyan}{\textbf{0.586}} \\
RCT ($\tau=0.1$) & \textcolor{blue}{\textbf{0.814}} & \textcolor{blue}{\textbf{0.207}} \\
RCT ($\tau=0.5$)   & 0.804 & 0.148 \\
RCT ($\tau=1$) & 0.773 & 0.146 \\
\bottomrule
\end{tabular}
\end{table}


\noindent\textbf{Interpretation of results.} \ Table~\ref{tab:qsar_results} illustrates a clear trade-off between prediction and inference in the QSAR application. The full CART fit achieves the highest test accuracy but does not yield valid post-fit confidence intervals. For both DS (highlighted in red) and RCT (highlighted in blue), we compare methods with predictive accuracy closest to CART. Notably, RCT produces confidence intervals that are less than half as long as those from DS, while still attaining better predictive accuracy that is closer to the full CART fit. 



\section{Concluding remarks}
\label{sec:conclusion}

In this paper, we introduce a temperature-controlled sampling scheme for fitting decision trees, based on an exponential mechanism. 
Unlike standard tree-fitting algorithms that are purely predictive, our method enables valid inference by accounting for the highly adaptive nature of the tree fit through externally introduced randomization.
Two novel contributions of our inferential approach are as follows: first, we show that pivotal quantities can be constructed directly using the sampling probabilities from the exponential mechanism, and this applies to any information gain measure in the tree-splitting criterion; 
second, valid inference conditional on the tree fit can be guaranteed without imposing restrictive assumptions about the quality of the tree fit or about how likely or rare the observed fit is.

Several promising avenues for future research remain to be explored.
Joint inference for multivariate parameters, as well as point estimates derived from a joint conditional distribution that accounts for the adaptive tree structure, remain to be developed.
Goodness-of-fit tests for assessing model complexity are a promising direction for future research.
More broadly, we believe that the proof techniques developed for our new randomization scheme have the potential to be extended to inference in a wide range of other settings.
For example, while we primarily focused on binary data, a similar analysis can be generalized to multinomial data. 
More generally, our techniques, which leverage the properties of the exponential mechanism, can provide a theoretical basis for applying such randomization schemes to construct feasible  inference in other data-adaptive problems.

\section{Acknowledgements}
The research of S. Panigrahi was supported in part by the NSF CAREER Award DMS-2337882 and the NIH grant 1R01GM152549-01.

\bibliographystyle{plainnat}    
\bibliography{references}   

\begin{appendices}

\section{Proofs for results in Section~\ref{sec:treegrowing}}
\label{appendix:proofs3}

\begin{proof}[Proof of Proposition \ref{prop: exp conv}]
Let $s^*(\cP) \in \underset{s'\in \cK(\cP)}{\text{argmax}} \; G_{\cP}(s'; y)$. 
%
%
Then, for any $s \in \mathcal{K}(\cP)$ and $s\notin \underset{s'\in \cK(\cP)}{\text{argmax}} \; G_{\cP}(s'; y)$, observe that
\begin{align*}
\Pp_{y}(s; \cP)= \mathbb{P}\rbracket{S(\cP) = s \mid Y=y} &= \frac{\exp \nbracket{\epsilon_{\cP}^{-1} G_{\cP}(s; y)}}{\sum_{s'\in \mathcal{K}(\cP)}\exp \nbracket{\epsilon_{\cP}^{-1} G_{\cP}(s'; y)}} \\   
&=\frac{\exp \nbracket{\epsilon_{\cP}^{-1} G_{\cP}(s; y)}}{\exp \nbracket{\epsilon_{\cP}^{-1}  G_{\cP}(s^{*}(\cP); y)}+\displaystyle\sum_{s'\neq s^{*}(\cP)  \in \mathcal{K}(\cP)}\exp \nbracket{\epsilon_{\cP}^{-1} G_{\cP}(s'; y)}} \\
&< \exp \nbracket{\epsilon_{\cP}^{-1} \nbracket{G_{\cP}(s; y) - G_{\cP}(s^{*}(\cP); y)}}. 
\end{align*}
If $s\notin \underset{s'\in \cK(\cP)}{\text{argmax}} \; G_{\cP}(s'; y)$, then 

$$
G_{\cP}(s; y) - G_{\cP}(s^{*}(\cP); y)<0,
$$
which in turn implies that:
$$
\Pp_{y}(s; \cP) \to 0 \; \text{ as }\; \epsilon_{\cP} \to 0.
$$
\end{proof}



\begin{proof}[Proof of Proposition~\ref{prop: exp gumbel}]
First, note that the i.i.d. noise $Z_{s} \sim \text{\normalfont Gumbel}(0,1)$ has cumulative distribution function (CDF) given by $F(z) = \exp(-\exp(-z))$. 

To simplify notation, let $V_{s} := V_{\cP}(s; Y)$ for all $s \in \cK(\cP)$.
Conditioning on $Z_{s} = z_{s}$, and using its independence from $Y = y$, we have that 

\begin{align*}
    &\mathbb{P}\left(V_{s'} \leq V_s \ \text{ for all } s' \in \cK(\cP) \mid Z_{s} = z_s ,Y=y\right)\\ 
    &= \prod_{s'\neq s} \exp\left\{ -\exp\left( -(z_s+\epsilon_{\cP}^{-1}(G_{\cP}(s; y)-G_{\cP}(s'; y))) \right) \right\} \\
    &= \exp\left( -\exp(-z_s)\cdot \sum_{s'\neq s} \exp\left( -\epsilon_{\cP}^{-1}(G_{\cP}(s; y)-G_{\cP}(s'; y)) \right) \right).
\end{align*}
To compute the probability in the claim, we are left to calculate

\begin{align*}
    &\mathbb{P}\left( V_{\cP}(s; Y)  \geq V_{\cP}(s'; Y) \, \text{ for all } \; s' \in \cK(\cP)\mid Y=y\right) \\
    &= \int_{-\infty}^{\infty}\mathbb{P}\left(V_{s'} \leq V_s \ \text{ for all } s' \in \cK(\cP) \mid Z_{s} = z_s ,Y=y\right)\exp(-z_s) \exp(-\exp(-z_s)) dz_s.
\end{align*}

Now, let $C = 1 + \sum_{s'\neq s} \exp\left( -\epsilon_{\cP}^{-1}(G_{\cP}(s; y)-G_{\cP}(s'; y)) \right)$. 
Observe that the integral on the right-hand side of the previous display is equal to

\begin{align*}
& \int_{-\infty}^\infty \exp(-z_s) \exp(-\exp(-z_s)) \exp\left( -\exp(-z_s) \sum_{s'\neq s} \exp(-\epsilon_{\cP}^{-1}(G_{\cP}(s; y)-G_{\cP}(s'; y))) \right) dz_s\\
&= \int_{-\infty}^{\infty} \exp(-z_s) \exp\left( -\exp(-z_s) C \right) dz_s\\
&= \int_0^\infty \exp(-Cu_s) du_s = \frac{1}{C},
\end{align*}
where in the last step we perform the following change of variables $u_s = \exp(-z_s)$, $dz_s = -\dfrac{du_s}{u_s}$.
This leads us to the claim:

\begin{align*}
\mathbb{P}\left( V_{\cP}(s; Y)  \geq V_{\cP}(s'; Y) \, \text{ for all} \; s' \in \cK(\cP)\mid Y=y\right) &= \dfrac{1}{C}
 = \dfrac{ \exp(\epsilon_{\cP}^{-1} G_{\cP}(s; y)) }{ \sum\limits_{s'\in \mathcal{K}(P)} \exp(\epsilon_{\cP}^{-1} G_{\cP}(s'; y)) }.
\end{align*}
\end{proof}


\section{Proof for result in Section~\ref{sec:inference}}
\label{appendix:proofs4}

\begin{proof}[Proof of Lemma~\ref{lem: suff stats}]
 Let $\cP$ be a parent of $l_1$. Then, $n_{\cP}^{1} = n_1$.
 Now, using the decomposition:
$
\widehat{\pi}_{\cR} = \dfrac{n_{\cR}^{1}}{n_{\cR}} \widehat{\pi}_{1}+ \dfrac{1}{\sqrt{n_{\cR}}}\widehat{\pi}_{\sminus 1}(\cR),
$
 we obtain
 \begin{align*}
 \begin{gathered}
 \hat{\pi}_{\cP} =  \frac{\sqrt{n_{1}}}{n_{\cP}}\sqrt{n_{1}}\widehat{\pi}_{1}  + \frac{1}{\sqrt{n_{\cP}}} \widehat{\pi}_{\sminus 1}(\cP), \
 \hat{\pi}_{\cP(s)} = \frac{n^{1}_{\cP(s)}}{n_{\cP(s)}\sqrt{n_{1}}}\sqrt{n_{1}}\widehat{\pi}_{1}    + \frac{1}{\sqrt{n_{\cP(s)}}} \widehat{\pi}_{\sminus 1}(\cP(s)),\\
  \hat{\pi}_{\cP'(s)} = \frac{n_{\cP}}{n_{\cP'(s)}}\hat{\pi}_{\cP}   - \frac{n_{\cP(s)}}{n_{\cP'(s)}}\hat{\pi}_{\cP(s)}=  \frac{n_{1}-n^{1}_{\cP(s)}}{n_{\cP'(s)}\sqrt{n_{1}}}\sqrt{n_{1}}\widehat{\pi}_{1} + \frac{\sqrt{n_{\cP}}}{n_{\cP'(s)}} \widehat{\pi}_{\sminus 1}(\cP)  -\frac{\sqrt{n_{\cP(s)}}}{n_{\cP'(s)}} \widehat{\pi}_{\sminus 1}(\cP(s)).
 \end{gathered}
 \end{align*}
By substituting $n_{\cP'(s)} = n - n_{\cP(s)}$ into the above expressions, we have that:
 \begin{align*}
     G_{\cP}(s; Y)= \cI(\hat{\pi}_{\cP})- \left\{\frac{n_{\cP(s)}}{n_{\cP}} \cI(\hat{\pi}_{\cP(s)})+\frac{n_{\cP'(s)}}{n_{\cP}} \cI(\hat{\pi}_{\cP'(s)})\right\} 
     = g_{\eta(s;\cP)}(\sqrt{n_{1}}\widehat{\pi}_{1},\widehat{\pi}_{\sminus 1}(\cP),\widehat{\pi}_{\sminus 1}(\cP(s))).
 \end{align*} 
 
 Therefore, for $\cP=p_{h}, s=s_{h}^{k}$, we get 
 
 \begin{align*}
 \begin{gathered}
 G_{p_h}(s^{k}_{h}; Y) = g_{\eta^{k}_{h}}\nbracket{\hat{T}_{n,1}, \hat{T}_{n,\sminus 1}^{h},\hat{T}_{n,\sminus 1}^{h,k}} \; \text{ for } \; k\in\cbracket{2,\ldots,k_h}, \text{ and }\\
G_{p_h}(s^{1}_{h}; Y) = g_{\eta^{1}_{h}}\nbracket{\hat{T}_{n,1}, \hat{T}_{n,\sminus 1}^{h},\hat{T}_{n,\sminus 1}^{h+1}} \; \text{ as } \; p_{h+1}=p_{h}(s_{h}^{1})
 \end{gathered} 
 \end{align*} Note that for above computation, w.l.o.g we set $p_{h+1}=p_{h}(s_{h}^{1})$ as $p_{h+1}\in \cbracket{p_{h}(s_{h}^{1}),p'_{h}(s_{h}^{1})}$ and the choice of left (or) right child doesn't affect the gain of the split, $G_{p_h}(s^{1}_{h}; Y)$. Hence, it follows that $G_{p_h}(s^{k}_{h}; Y) = \cG^{k}_{h}\nbracket{\hat{T}_{n,1}, {\hat{T}_{n,\sminus 1}}}$.
\end{proof}



\section{Proofs for results in Section~\ref{sec:asymptotictheory}}
\label{appendix:proofs5}

\begin{proof}[Proof of Proposition~\ref{prop: linrep}]
By Definition~\ref{defn: weights}, each vector $a_{i,n}\in\mathbb{R}^{\bar{k}}$ collects the weights corresponding to the root region, internal regions, and their candidate split regions. We verify that every coordinate of $\hat{T}_n$ can be written as a weighted sum of the form $\frac{1}{\sqrt{n}}\sum_{i=1}^{n}a_{i,n,\ell}Y_{i,n}$, and stacking these components yields the desired linear representation.

For the terminal region of interest $l_1$, note that

\[
\hat{T}_{n,1}=\sqrt{n_1}\,\hat{\pi}_1
=\frac{1}{\sqrt{n_1}}\sum_{i=1}^{n}Y_{i,n}\,\mathbbm{1}[X_{i,n}\in l_1]
=\frac{1}{\sqrt{n}}\sum_{i=1}^{n}\sqrt{\frac{n}{n_1}}\,\mathbbm{1}[X_{i,n}\in l_1]\,Y_{i,n},
\]
so that $\hat{T}_{n,1}=\frac{1}{\sqrt{n}}\sum_{i=1}^{n}a_{i,n,1}Y_{i,n}$ with $a_{i,n,1}$ defined as in Definition~\ref{defn: weights}.

For any parent region $p_h$ in $\ct_1$, using the definition in Equation~\eqref{orth:stat1} with $n_{h,0}=n_{p_h}$ and $n_{h,0}^{1}=n_{p_h\cap l_1}$, we have

\begin{equation*}
\begin{aligned}
\hat{T}_{n,\sminus 1}^{h}
&=\sqrt{n_{h,0}}\widehat{\pi}_{p_h}
-\frac{n_{h,0}^{1}}{\sqrt{n_1n_{h,0}}}\,\sqrt{n_1}\widehat{\pi}_1\\
&= \frac{1}{\sqrt{n}}\sum_{i=1}^{n}\sqrt{\frac{n}{n_{h,0}}}\Big(
\mathbbm{1}[X_{i,n}\in p_h]-\mathbbm{1}[X_{i,n}\in l_1]\Big)\,Y_{i,n},
\end{aligned}
\end{equation*}
which coincides with $\frac{1}{\sqrt{n}}\sum_{i=1}^{n}a_{i,n,\sminus 1}^{h}Y_{i,n}$. 

Similarly, for each candidate split region $p_h^k=p_h(s_h^k)$ with $n_{h,k}=n_{p_h^k}$ and $n_{h,k}^{1}=n_{p_h^k\cap l_1}$, we obtain\[
\hat{T}_{n,\sminus 1}^{h,k}
=\sqrt{n_{h,k}}\widehat{\pi}_{p_h^k}
-\frac{n_{h,k}^{1}}{\sqrt{n_1n_{h,k}}}\,\sqrt{n_1}\widehat{\pi}_1
=\frac{1}{\sqrt{n}}\sum_{i=1}^{n}\sqrt{\frac{n}{n_{h,k}}}
\Big(\mathbbm{1}[X_{i,n}\in p_h^k]-\frac{n_{h,k}^{1}}{n_1}\mathbbm{1}[X_{i,n}\in l_1]\Big)\,Y_{i,n},
\]
which equals $\frac{1}{\sqrt{n}}\sum_{i=1}^{n}a_{i,n,\sminus 1}^{h,k}Y_{i,n}$.

Stacking all coordinates of $\hat{T}_n$ as

\[
\hat{T}_{n}
=\begin{bmatrix}\hat{T}_{n,1} & (\hat{T}_{n,\sminus 1})^\top \end{bmatrix}^\top,
\qquad 
a_{i,n}
=\begin{bmatrix}a_{i,n,1} & (a_{i,n,\sminus 1})^\top\end{bmatrix}^\top,
\]
yields the desired representation.
\end{proof}



\begin{proof}[Proof of Proposition~\ref{prop: normaldistn}]
First, observe that

\[
\hat T_n-\Pi_n=\frac{1}{\sqrt n}\sum_{i=1}^n a_{i,n}\big(Y_{i,n}-\theta_{i,n}\big),
\qquad
\Sigma_n=\frac{1}{n}\sum_{i=1}^n \theta_{i,n}(1-\theta_{i,n})\,a_{i,n}a_{i,n}^\top,
\] 
which follows from using the fact that the summands $a_{i,n}(Y_{i,n}-\theta_{i,n})/\sqrt n\in\mathbb R^{\bar k}$ are independent and mean zero, and that
\[
\Cov\!\Big(\frac{1}{\sqrt n}\sum_{i=1}^n a_{i,n}(Y_{i,n}-\theta_{i,n})\Big)
=\frac{1}{n}\sum_{i=1}^n \operatorname{Var}(Y_{i,n})\,a_{i,n}a_{i,n}^\top
=\Sigma_n .
\] 

Now, to show asymptotic normality, we verify the multivariate Lindeberg condition.
Fix $\varepsilon>0$ and observe that
$|Y_{i,n}-\theta_{i,n}|\le 1$, and hence
\vspace{-0.5cm}
\[
\big\|\tfrac{1}{\sqrt n}a_{i,n}(Y_{i,n}-\theta_{i,n})\big\|
\le \frac{\|a_{i,n}\|}{\sqrt n}. 
\vspace{-0.5cm}
\]
Under Assumption \ref{ass:1}, note that $\sup_{i,n}\|a_{i,n}\|\le M<\infty$ for some constant $M$ and  Assumption~\ref{ass:3} ensures that $\Sigma_n$ is nonsingular with $0<\frac{1}{b^2}<\lambda_{\min}(\Sigma_n)\leq\lambda_{\max}(\Sigma_n)<B^2<\infty$.
Then, for all $n>(Mb^2/\varepsilon)^2$, we have $\|a_{i,n}\|/\sqrt n\le \varepsilon/b^2< \varepsilon \lambda_{\min}^2(\Sigma_n)$ for every $i$, and thus we conclude that
\vspace{-0.5cm}
\[
\tfrac{1}{\lambda_{\min}^2(\Sigma_n)}\lim_{[\bn]  \to \infty}\sum_{i=1}^n \Ee\!\left[\big\|\tfrac{1}{\sqrt n}a_{i,n}(Y_{i,n}-\theta_{i,n})\big\|^2
\,\mathbbm{1}\!\left\{\big\|\tfrac{1}{\sqrt n}a_{i,n}(Y_{i,n}-\theta_{i,n})\big\|>\varepsilon\lambda_{\min}^2(\Sigma_n)\right\}\right]=0.
\]

Hence, the Lindeberg condition holds, and by the multivariate Lindeberg–Feller central limit theorem, we have that $\Sigma_n^{-1/2}(\hat T_n-\Pi_n) \;\indist\; \mathcal N_{\bar k}(0,I_{\bar k,\bar k}).$
\end{proof}




\begin{proof}[Proof of Proposition \ref{prop: relativediff}]
From Lemma~\ref{lemma:supporting1}, it follows that the convergence in~\eqref{WeakConv} holds if

\begin{equation}
\lim_{[\bn]  \to \infty}\left|\frac{\mathbb{E}\left[h \circ P\left(\zeta_{n}\right) \times \Lambda\left(\zeta_{n}\right)\right]}{\mathbb{E}\left[\Lambda\left(\zeta_{n}\right)\right]}-\frac{\mathbb{E}\left[h \circ P(Z_{n}) \times \Lambda(Z_{n})\right]}{\mathbb{E}\left[\Lambda(Z_{n})\right]}\right|=0. 
\label{inter:bound}
\end{equation}
In the remainder of the proof, we bound the left-hand side display of \eqref{inter:bound} in terms of the difference terms in our assertion. 
Note that, using the triangle inequality, \eqref{inter:bound} can be bounded by

$$
\mathrm{term}_1+\mathrm{term}_2
$$
where
\vspace{-1cm}
\begin{align*}
&\mathrm{term}_1=\left|\frac{\mathbb{E}\left[h \circ P\left(\zeta_{n}\right) \times \Lambda\left(\zeta_{n}\right) \right]}{\mathbb{E}\left[\Lambda\left(\zeta_{n}\right)\right]}-\frac{\mathbb{E}\left[h \circ P\left(\zeta_{n}\right) \times \Lambda\left(\zeta_{n}\right) \right]}{\mathbb{E}\left[\Lambda(Z_{n})\right]}\right|, \\
& \mathrm{term}_2=\left|\frac{\mathbb{E}\left[h \circ P\left(\zeta_{n}\right) \times \Lambda\left(\zeta_{n}\right) \right]}{\mathbb{E}\left[\Lambda(Z_{n})\right]}-\frac{\mathbb{E}\left[h \circ P(Z_{n}) \times \Lambda(Z_{n}) \right]}{\mathbb{E}\left[\Lambda(Z_{n})\right]}\right|.
\end{align*}

It is immediate that $\mathrm{term}_2$ is equal to $\RD_n^{(2)}$.
To complete the proof, observe that

$$
\begin{aligned}
\mathrm{term}_1 & \leq\left|\mathbb{E}\left[h \circ P\left(\zeta_{n}\right) \times \Lambda\left(\zeta_{n}\right) \right]\right| \times\left|\frac{1}{\mathbb{E}\left[\Lambda\left(\zeta_{n}\right)\right]}-\frac{1}{\mathbb{E}\left[\Lambda\left(Z_{n}\right)\right]}\right| \\
& \leq \sup _{h \in \mathbb{C}^3(\mathbb{R})}|h| \times\left|\frac{\mathbb{E}\left[\Lambda\left(\zeta_{n}\right) \right]}{\mathbb{E}\left[\Lambda\left(\zeta_{n}\right)\right]}-\frac{\mathbb{E}\left[\Lambda\left(\zeta_{n}\right) \right]}{\mathbb{E}\left[\Lambda\left(Z_{n}\right)\right]}\right|=\sup _{h \in \mathbb{C}^3(\mathbb{R})}|h| \times\RD_n^{(1)}.
\end{aligned}
$$
Thus, we conclude that the left-hand side display of \eqref{inter:bound} is bounded as

$$
\begin{aligned}
& \lim_{[\bn]  \to \infty}\left|\frac{\mathbb{E}\left[h \circ P\left(\zeta_{n}\right) \times \Lambda\left(\zeta_{n}\right)\right]}{\mathbb{E}\left[\Lambda\left(\zeta_{n}\right)\right]}-\frac{\mathbb{E}\left[h \circ P(Z_{n}) \times \Lambda(Z_{n})\right]}{\mathbb{E}\left[\Lambda(Z_{n})\right]}\right| \\
& \leq \sup _{h \in \mathbb{C}^3(\mathbb{R})}|h| \times \lim_{[\bn]  \to \infty} \RD_n^{(1)}+\lim_{[\bn]  \to \infty} \RD_n^{(2)}, 
\end{aligned}
$$
from which our claim follows directly.
\end{proof}

\begin{proof}[Proof of Theorem \ref{thm:plugin}]
\noindent{\textbf{Proof of (i)}}. \quad To prove the first part of the claim, fix $\delta > 0$, and note that, using our notation, we can write
\vspace{-0.5cm}
\[
\mathbb{P}\!\left(\,|\hat\sigma_{n,1}^2-\sigma_{n,1}^2|>\delta\ \big|\ \bS_1=\bs_1\,\right)
=\int \mathbbm{1}\!\left\{|\hat\sigma_{n,1}^2-\sigma_{n,1}^2|>\delta\right\}
\frac{\Lambda(\zeta_n)}{\Ee[\Lambda(\zeta_n)]}\,d\mathbb{P}\rbracket{\zeta_n}.
\]
Now, observe that
 \vspace{-0.5cm}
 \begin{equation*}
 \begin{aligned}
 \mathbb{E}\left[\Lambda(\zeta_n)\right]  &= \mathbb{E}\left[\exp \left(f\nbracket{\Sigma_{n}^{1/2}\zeta_n+\Pi_{n}}\right)\right]    \\
 &= \exp f\nbracket{\Pi_{n}}\times \mathbb{E}\left[\exp \nbracket{ f\nbracket{\Sigma_{n}^{1/2}\zeta_n+\Pi_{n}}-f\nbracket{\Pi_{n}}}\right]  \\
 &\geq \exp{f\nbracket{\Pi_{n}}}\times \mathbb{E}\rbracket{\exp (-L\norm{\Sigma_{n}^{1/2}\zeta_n})}\\
  &\geq \exp{f\nbracket{\Pi_{n}}}\times \mathbb{E}\rbracket{\exp (-LB\norm{\zeta_n})}\\
  &\geq \exp{f\nbracket{\Pi_{n}}}\times \exp\left(-LB \; \mathbb{E}\rbracket{\norm{\zeta_n}}\right),
  \end{aligned}
  \end{equation*}
where the first inequality follows from the fact that the log-probability of the exponential mechanism, $f(\cdot)$, is Lipschitz with constant $L$, as established in Proposition~\ref{prop:supporting1}, the following inequality relies on the bound for $||\Sigma_{n}^{1/2}||$ as stated in Remark~\ref{remark:ass:3} (based on Assumption~\ref{ass:3}), and the final display is obtained by applying Jensen's inequality to the convex function $f(x)= \exp(-LB\; x)$.
Since $\zeta_n$ is standardized with unit variance, using  $\mathbb{E}\rbracket{\norm{\zeta_n}} \leq \mathbb{E}\rbracket{\norm{\zeta_n}^2} +1=2$, it follows directly that
\vspace{-0.5cm}
$$
\Ee[\Lambda(\zeta_n)] \ge \exp (-2LB) \times \exp f(\Pi_n).
$$ Moreover, since $f(\cdot)$ is Lipschitz with constant $L$, we have
\vspace{-0.5cm}
\begin{equation*}
\begin{aligned}
\Lambda(\zeta_n)& =\exp{f(\Sigma_n^{1/2}\zeta_n+\Pi_n)}\le \exp{f(\Pi_n)}\cdot \exp{L\|\Sigma_n^{1/2}\|\|\zeta_n\|}
\le \exp{f(\Pi_n)}\cdot \exp{LB\|\zeta_n\|}.
\end{aligned}
\end{equation*}

That is, we have
\vspace{-0.5cm}
$$
\frac{\Lambda(\zeta_n)}{\Ee[\Lambda(\zeta_n)]} \leq \exp(2 LB)\cdot\exp{LB\|\zeta_n\|}\leq \text{const}\times \exp{LB\|\zeta_n\|}. \vspace{-0.5cm}
$$
with the above-stated display is obtained by substituting the derived lower bound of $\mathbb{E}[\Lambda(\zeta_n)]$.
Then, note that
\vspace{-0.5cm}
\begin{equation*}
\begin{aligned}
\mathbb{P}\!\left(\,|\hat\sigma_{n,1}^2-\sigma_{n,1}^2|>\delta\ \big|\ \bS_1=\bs_1\,\right)
&\leq \text{const}\times\int \mathbbm{1}\!\left\{|\hat\sigma_{n,1}^2-\sigma_{n,1}^2|>\delta\right\}
 \times\exp{LB\|\zeta_n\|}\,d\mathbb{P}\rbracket{\zeta_n} \\
 &\leq \text{const}\times \Ee\rbracket{\mathbbm{1}\!\left\{|\hat\sigma_{n,1}^2-\sigma_{n,1}^2|>\delta\right\}
 \times\exp{LB\|\zeta_n\|}}\\
 &\leq  \text{const}\times\mathbb{P}\nbracket{|\hat\sigma_{n,1}^2-\sigma_{n,1}^2|>\delta}^{1/2} \times \Ee\rbracket{\exp{2LB\|\zeta_n\|}}^{1/2},\vspace{-0.5cm}
\end{aligned}
\end{equation*}
where we apply the Cauchy–Schwarz inequality to obtain the final inequality from the penultimate step.
We complete our proof by showing that $\Ee\rbracket{\exp{2LB\|\zeta_n\|}} <\infty$ for all large $n$, using the fact that $\zeta_n$ is a sub-Gaussian variable.

Recall that $\zeta_{n} = \dfrac{1}{\sqrt{n}} \sum_{i\in [n]} b_{i,n} \tilde{Y}_{i,n}$, where $\tilde{Y}_{i,n}$ are bounded  variables and therefore sub-Gaussian. 
Furthermore, $\max_{i\in[n]}\|b_{i,n}\|_{1}\le \bar C_{+}$ implies  $|u^\top b_{i,n}|\le \|u\|_\infty\|b_{i,n}\|_{1}\le \bar C_{+}$ for any unit vector $u$. 
Then any projection $u^\top \zeta_n = n^{-1/2}\sum_{i=1}^n u^\top b_{i,n}\tilde Y_{i,n}$
is a linear combination of independent sub-Gaussian variables and is therefore sub-Gaussian. 
Consequently, $\zeta_n$ is sub-Gaussian as well, and it holds that 
\vspace{-0.5cm}
$$
\sup_n\Ee\rbracket{\exp{2LB\|\zeta_n\|}}<\infty.
\vspace{-0.5cm}
$$
Therefore, letting $\tilde{\text{const}}=\text{const}\times\displaystyle\sup_n\Ee\rbracket{\exp{2LB\|\zeta_n\|}}^{1/2}$, we obtain
\[
\mathbb{P}\!\left(\,|\hat\sigma_{n,1}^2-\sigma_{n,1}^2|>\delta\ \big|\ \bS_1=\bs_1\,\right)
\le \tilde{\text{const}}\times
\mathbb{P}\!\left(|\hat\sigma_{n,1}^2-\sigma_{n,1}^2|>\delta\right)^{1/2}.
\]
The display on the right-hand side of the inequality converges to $0$ in the limit, since $\hat\sigma_{n,1}^2$ is consistent with respect to the pre-conditional distribution. This proves our claim.

\noindent\textbf{Proof of (ii)}. 
When the true variance $\sigma_{n,1}^{2}$ is known, it has already been established in Theorem~\ref{thm: asymptotic validity} that
\vspace{-0.5cm}
\[
\mathrm{Pivot}(\hat T_{n,1};\Pi_{n,1},\sigma_{n,1}^{2})
\;\big|\;\cbracket{\bS_{1}=\bs_{1}}
\;\indist\; \mathrm{Uniform}[0,1].\vspace{-0.5cm}
\] 
To show this part of the claim, we prove that the pivot is Lipschitz in the variance parameter $\sigma_{n,1}^{2}$, i.e.,
\vspace{-0.5cm}
$$
\bigl|\mathrm{Pivot}(\hat T_{n,1};\Pi_{n,1},\hat\sigma_{n,1}^{2})-
\mathrm{Pivot}(\hat T_{n,1};\Pi_{n,1},\sigma_{n,1}^{2})\bigr|
\leq L_0\cdot |\hat\sigma_{n,1}^{2}-\sigma_{n,1}^{2}|,
\vspace{-0.5cm}
$$
for some constant $L_0$. 
Therefore, the pivot is uniformly continuous in the variance parameter $\sigma_{n,1}^{2}$, and by applying Slutsky’s theorem together with the result in Theorem~\ref{thm: asymptotic validity}, we obtain
\[
\mathrm{Pivot}(\hat T_{n,1};\Pi_{n,1},\hat\sigma_{n,1}^{2})
\;\big|\; \{\bS_{1}=\bs_{1}\}
\;\indist\; \mathrm{Uniform}[0,1].
\] 

Firstly, note that by Assumption~\ref{ass:2}, $\pi_{n,1}\in[\delta_{1},1-\delta_{1}]$ for some $\delta_{1}>0$, so we get a bounded support for
$\sigma_{n,1}^{2}=\pi_{1,n}(1-\pi_{1,n})\in[\delta_{1}(1-\delta_{1}),1/4]$. 
We now verify that the pivot is Lipschitz in $\sigma_{n,1}^2$ by showing that the derivative of the pivot with respect to $\sigma_{n,1}^2$ is uniformly bounded over $\sigma_{n,1}^2 \in [\delta_{1}(1-\delta_{1}),1/4]$.
Since $\lambda_h$ does not depend on $\sigma_{n,1}^{2}$, taking derivative of the pivot with respect to $\sigma_{n,1}^2$ yields
\[
\frac{\partial}{\partial \sigma_{n,1}^{2}}
\mathrm{Pivot}(\hat T_{n,1};\Pi_{n,1},\sigma_{n,1}^{2})
=
\frac{N'(\sigma_{n,1}^{2})D(\sigma_{n,1}^{2})
      -N(\sigma_{n,1}^{2})D'(\sigma_{n,1}^{2})}
     {D(\sigma_{n,1}^{2})^{2}},
\]
where \(N(\sigma_{n,1}^{2})\) and \(D(\sigma_{n,1}^{2})\) denote the numerator and denominator integrals in the pivot, respectively, viewed as functions of $\sigma_{n,1}^{2}$.
Because $0\le \lambda_h\le 1$ and $\phi(\cdot;\Pi_{n,1},\sigma_{n,1}^{2})$ is a density, both $N(\sigma_{n,1}^{2})$ and $D(\sigma_{n,1}^{2})$ are uniformly bounded above on
$\sigma_{n,1}^{2}\in[\delta_1(1-\delta_1),\,1/4]$.
Moreover, $D(\sigma_{n,1}^{2})$ is bounded away from zero with high probability, and $N'(\sigma_{n,1}^{2})$ and $D'(\sigma_{n,1}^{2})$ are also uniformly bounded over $\sigma_{n,1}^{2}\in[\delta_{1}(1-\delta_{1}),1/4]$.
Therefore, the right-hand side is uniformly bounded, and the pivot is Lipschitz in $\sigma_{n,1}^{2}$.

\end{proof}

\section{Supporting results \& proofs}
\label{appendix:supporting}
\begin{lemma}
\label{lemma:lindeberg}
For functions $\Psi\in\cbracket{ \Psi^{(1)}(z)=\Lambda(z), \ \Psi^{(2)}(z)= h \circ P\left(z\right) \times \Lambda(z)}$ and

\[
U_{(i)}=(\tilde{Y}_{1,n},\ldots,\tilde{Y}_{i,n},Z_{i+1,n},\ldots,Z_{n,n})^\top,\quad
V_{(i)}=(\tilde{Y}_{1,n},\ldots,\tilde{Y}_{i-1,n},0,Z_{i+1,n},\ldots,Z_{n,n})^\top,
\]
with $\{Z_{i,n}: i\in [n]\}$ and $\{\tilde{Y}_{i,n}: i\in [n]\}$ as defined earlier, we have
\begin{equation*}
\begin{aligned}
    &\left |\Ee\rbracket{\Psi\nbracket{\frac{1}{\sqrt{n}}B_{(n)}^\top\tilde{Y}_{(n)}}}-  \Ee\rbracket{\Psi\nbracket{\frac{1}{\sqrt{n}}B_{(n)}^\top \tilde{Z}_{(n)}}}\right|\\
    & \;\;\;\;\;\;\;\;\;\;\;\;\;\;\;\;\;\;\;\;\;\;\;\;\;\;\;\;\;\;\leq 2 \frac{\|B_{(n)}\|_{\infty}^3}{\sqrt{n}}\max_{i\in [n]}\sup_{\alpha\in(0,1)}\Ee\!\Big[\norm{\nabla^3 \Psi\big(\tfrac{1}{\sqrt{n}} B_{(n)}^\top \big((1-\alpha) V_{(i)}+\alpha U_{(i)}\big)\big)}\Big].
\end{aligned}
\end{equation*} 
\end{lemma}

\begin{proof}
We proceed using Lindeberg's telescoping idea as outlined in \citep{chatterjee2005simpleinvariancetheorem, Lindeberg1922}, writing the difference by replacing the coordinates one at a time and introducing

\begin{align*}
T_{i,n}&:= \Ee\, \Psi\Big(\tfrac{1}{\sqrt{n}}B_{(n)}^\top U_{(i)}\Big)
      - \Ee\, \Psi\Big(\tfrac{1}{\sqrt{n}}B_{(n)}^\top V_{(i)}\Big),\
\bar{T}_{i,n}&:= \Ee\, \Psi\Big(\tfrac{1}{\sqrt{n}}B_{(n)}^\top U_{(i-1)}\Big)
      - \Ee\, \Psi\Big(\tfrac{1}{\sqrt{n}}B_{(n)}^\top V_{(i)}\Big). 
\end{align*}
Then, we have that

\[
\Big|\Ee\,\Psi\!\big(\tfrac{1}{\sqrt n}B_{(n)}^\top\tilde{Y}_{(n)}\big)
-\Ee\,\Psi\!\big(\tfrac{1}{\sqrt n}B_{(n)}^\top \tilde{Z}_{(n)}\big)\Big|
\;\le\; \sum_{i=1}^n |T_{i,n}-\bar T_{i,n}|.
\]

\noindent For each $i$, define

\[
g_i(t):=\Psi\!\Big(\tfrac{1}{\sqrt{n}}B_{(n)}^\top W_{(i)}(t)\Big),\quad 
W_{(i)}(t)=(\tilde{Y}_{1,n},\ldots,\tilde{Y}_{i-1,n},t,Z_{i+1,n},\ldots,Z_{n,n})^\top.
\]
Then, observe that $T_{i,n}-\bar T_{i,n}=\Ee[g_i(\tilde{Y}_{i,n})-g_i(Z_{i,n})]$. 

A second–order Taylor expansion of $g_i(t)$ around $t=0$ gives

\[
g_i(t)=g_i(0)+g_i^{(1)}(0)t+\tfrac{1}{2}g_i^{(2)}(0)t^2+R_i(t),
\quad \text{where }\ 
R_i(t)=\dfrac{1}{2}\!\int_0^t(t-u)^2g_i^{(3)}(u)\,du.
\]
Because $\Ee[\tilde{Y}_{i,n}]=\Ee[Z_{i,n}]=0$ and $\Ee[\tilde{Y}_{i,n}^2]=\Ee[Z_{i,n}^2]=1$, the linear and quadratic terms cancel in expectation, leaving only the remainders

\[
T_{i,n}-\bar T_{i,n}
=\Ee[R_i(\tilde{Y}_{i,n})-R_i(Z_{i,n})]
=\Ee[R_{i,n}-\bar R_{i,n}],
\]
where $R_{i,n}= R_i(\tilde{Y}_{i,n})$ and $\bar R_{i,n}= R_i(Z_{i,n})$.
Then for the remainder terms, applying the mean-value theorem and taking expectations, we obtain that

\begin{align*}
&|\Ee[R_{i,n}]|, |\Ee[\bar R_{i,n}]|
\le \frac{1}{6}\Big(\Ee|\tilde{Y}_{i,n}|^3+\Ee|Z_{i,n}|^3\Big)
\sup_{\alpha\in[0,1]}
\Ee\!\Big[\big|\partial_i^3 \Psi\big(\tfrac{1}{\sqrt{n}} B_{(n)}^\top \big((1-\alpha) V_{(i)}+\alpha U_{(i)}\big)\big)\big|\Big]\\
&\le \frac{1}{6}\Big(\Ee|\tilde{Y}_{i,n}|^3+\Ee|Z_{i,n}|^3\Big)
\left(\dfrac{||b_{i,n}||}{\sqrt n}\right)^{\!3}
\sup_{\alpha\in[0,1]}
\Ee\!\Big[\norm{\nabla^3 \Psi\big(\tfrac{1}{\sqrt{n}} B_{(n)}^\top \big((1-\alpha) V_{(i)}+\alpha U_{(i)}\big)\big)}\Big].
\end{align*}

Now, $|T_{i,n}-\bar T_{i,n}|\le |\Ee[R_{i,n}]|+|\Ee[\bar R_{i,n}]|$. 
Using the bound that $\Ee|Z_{i,n}|^3=2\sqrt{2/\pi}<2$ and the uniform bound on $\Ee|\tilde{Y}_{i,n}|^3$ for the standardized Bernoulli variable, observe that $\dfrac{1}{6}\big(\Ee|\tilde{Y}_{i,n}|^3+\Ee|Z_{i,n}|^3\big)$ is bounded above by $2$. Finally, summing over $i\in [n]$, we have that

\[
\sum_{i=1}^n|T_{i,n}-\bar T_{i,n}|\leq 2 \frac{\|B_{(n)}\|_{\infty}^3}{\sqrt{n}}\max_{i\in [n]}\sup_{\alpha\in(0,1)}\Ee\!\Big[\norm{\nabla^3 \Psi\big(\tfrac{1}{\sqrt{n}} B_{(n)}^\top \big((1-\alpha) V_{(i)}+\alpha U_{(i)}\big)\big)}\Big]. 
\] This yields the stated inequality.
\end{proof}

\begin{proposition}
\label{prop:supporting0}
Under Assumptions~\ref{ass:1},\ref{ass:2},\ref{ass:3}, $\norm{b_{i,n}}$ are bounded for each $i\in [n]$. Moreover, there exists a constant  $\bar{C}_{+}>0$ such that $\|B_{(n)}\|_{\infty}\leq \bar{C}_{+}$ for all $n$. 
\end{proposition}

\begin{proof}

Recall $b_{i,n} = \sqrt{\theta_{i,n}(1-\theta_{i,n})}\Sigma_{n}^{-1/2}a_{i,n}$ and the entries of $a_{i,n}$ are of the following form:
\begin{align*}
a_{i,n,1} =  \sqrt{\frac{n}{n_1}}\mathbbm{1}[X_{i,n}\in l_{1}],& \ a_{i,n,\sminus 1}^{h} =  \sqrt{\frac{n}{n_{h,0}}}\nbracket{\mathbbm{1}[X_{i,n}\in p_{h}]-\mathbbm{1}[X_{i,n}\in l_{1}]}\\
a_{i,n,\sminus 1}^{h,k} &=  \sqrt{\frac{n}{n_{h,k}}}\nbracket{\mathbbm{1}[X_{i,n}\in p_{h}^{k}]-\frac{n_{h,k}^{1}}{n_1}\mathbbm{1}[X_{i,n}\in l_{1}]}.
\end{align*} 
Under the assumptions~\ref{ass:1},\ref{ass:2}, the ratios are uniformly bounded implying that $\norm{a_{i,n}}$ for each $i\in[n]$ is uniformly bounded over $n$.

Furthermore, note that $\limsup_{[\bn] \to \infty} \theta_{i,n}(1-\theta_{i,n}) < 1$, and Remark~\ref{remark:ass:3} (which follows from Assumption~\ref{ass:3}) guarantees that the norms of the root-inverse covariance matrices $\Sigma_{n}^{-1/2}$ are uniformly bounded. Consequently, the $\ell_2$-norm of $b_{i,n}$ is also uniformly bounded for each $i \in [n]$. Specifically, there exists a constant $\bar{C}_{+} > 0$, not dependent on $n$, such that:

\begin{equation*}
\begin{aligned}
\|b_{i,n}\|_1
   &\leq \sqrt{\bar{k}}\,\|b_{i,n}\|_2 = \sqrt{\bar{k}}\,\sqrt{\theta_{i,n}(1-\theta_{i,n})}\,\|\Sigma_n^{-1/2} a_{i,n}\|_2 \\
   &\leq \sqrt{\bar{k}}\,\sqrt{\theta_{i,n}(1-\theta_{i,n})}\,\|\Sigma_n^{-1/2}\|\|a_{i,n}\|_2 \leq \bar{C}_{+},  
\end{aligned}
\end{equation*}
for all $i\in [n]$, and therefore, it holds that
$
\|B_{(n)}\|_{\infty} = \max_{i\in[n]} \|b_{i,n}\|_{1} \leq \bar{C}_{+}.
$
\end{proof}

\begin{proposition}
\label{prop:supporting1}
Under Assumptions~\ref{ass:1},\ref{ass:2}, and \ref{ass:infodevbound}, there exist positive constants $L, b_{1},b_{2},b_{3}$ such that the following statements hold:
\begin{enumerate}
    \item The gain functions $\cG^{k}_{h}(v)$ and their first- to third-order derivatives order  $\nabla \cG^{k}_{h}(v)$,$\nabla^2 \cG^{k}_{h}(v)$, $\nabla^3 \cG^{k}_{h}(v)$ are bounded for all $v\in \tilde{D}_n$.
    \item The function $f(.)$, derived from the logarithm of the sampling probabilities based on the gain functions, is $L$-Lipschitz, i.e. $\abs{f(v_1)-f(v_2)}\leq L\norm{v_1 - v_2}$ for all $v_1, v_2 \in \tilde{D}_n$
    \item Furthermore, $\norm{\nabla f(v)}\leq b_{1},\norm{\nabla^{2} f(v)}\leq b_{2}, \norm{\nabla^{3} f(v)}\leq b_{3}$ for all $v\in \tilde{D}_n$ 
\end{enumerate}
\end{proposition}

\begin{proof}

\noindent
\textbf{Part 1.} 
Recall
$\cG^{k}_{h}\nbracket{v}:= g_{\eta^{k}_{h}}\nbracket{v_{1}, v_{\sminus 1}^{h}, v_{\sminus 1}^{h,k}}$ where 

\[
  g_{\eta}\nbracket{x,y,z} := \cI\nbracket{\eta_{1}x+\eta_{2}y} -\eta_{8}\cI\nbracket{\eta_{3}x+\eta_{4}z}-(1-\eta_{8})\cI\nbracket{\eta_{5}x+\eta_{6}y+\eta_{8}z}. 
\]
Firstly, note that the absolute value of each component of $\eta^{k}_{h} \in \mathbb{R}^{8}$ is uniformly bounded with respect to $n$. Under Assumption~\ref{ass:infodevbound}, the function $\cI(\cdot)$ is bounded and possesses bounded first-, second-, and third-order derivatives. The function 
\( g_\eta(x, y, z) \) is therefore bounded, as it is a linear combination of three terms, each being the function \( \cI(.) \) applied to an affine transformation of \( (x, y, z) \). 
Each affine argument involves linear combinations of the inputs (defined on the bounded domain $D(\eta)$) with coefficients depending on \( \eta =\eta^{k}_{h}\), which are themselves uniformly bounded. Consequently, these affine maps---denoted $a_\eta(x, y, z)$---are smooth and have uniformly bounded derivatives of all orders.

 Now, since \( \cI(.) \) is three times differentiable with uniformly bounded derivatives, the chain rule implies that the derivatives (up to the third order) of the composed functions \( \cI(a_\eta(x, y, z)) \) are uniformly bounded as well. Furthermore, the coefficients \( \eta_8 \) and \( 1 - \eta_8 \) are bounded. Therefore, the partial derivatives of \( g_{\eta} \), being linear combinations of bounded functions and their bounded derivatives, are uniformly bounded up to the third order.


\noindent
\textbf{Part 2.} To prove that $f\nbracket{v}$ is Lipschitz continuous, it suffices to show that $\norm{\nabla f(v)}$ is bounded, which follows from Part 3 (proved next).

\noindent
\textbf{Part 3.} We can write

\[
f(v)=\sum_{h=1}^d f_{h}(v),\qquad \text{where }\ 
f_{h}(v):=\epsilon_{h}^{-1}\,\cG_h^{1}(v)-\log\!\Big(\sum_{k=1}^{k_h} e^{\epsilon_{h}^{-1}\,\cG_h^{k}(v)}\Big),
\]
and define the softmax weights

\[
w_{h,k}(v):=\frac{e^{\epsilon_{h}^{-1}\,\cG_h^{k}(v)}}{\sum_{\ell=1}^{k_h} e^{\epsilon_{h}^{-1}\,\cG_h^{\ell}(v)}},\qquad \text{where }\ 
0\le w_{h,k}\le 1,\ \text{ and } \sum_k w_{h,k}=1.
\]

\smallskip
\noindent\textit{Gradient.}
By standard algebra, we obtain

\[
\nabla f_{h}(v)=\epsilon_{h}^{-1}\Big(\nabla\cG_h^{1}(v)-\sum_{k} w_{h,k}(v)\,\nabla\cG_h^{k}(v)\Big),
\]
which is a difference between a bounded vector and a convex combination of bounded vectors, given that the derivatives of the gain functions are uniformly bounded.
Hence, $\|\nabla f_{h}(v)\|$ is uniformly bounded. 
Summing over $h$ then gives $\|\nabla f(v)\| \le b_1$ for some finite constant $b_1 > 0$.

\smallskip
\noindent\textit{Hessian.}
Using $\nabla w_{h,k}(v)=\epsilon_h\,w_{h,k}(v)\big(\nabla\cG_h^{k}(v)-\sum_\ell w_{h,\ell}(v)\nabla\cG_h^{\ell}(v)\big)$, we obtain $\nabla^2f_{h}(v)$ to be
\[
\epsilon_h\underbrace{\Big(\nabla^2\cG_h^{1}-\sum_{k}w_{h,k}\nabla^2\cG_h^{k}\Big)}_{\text{(T1)}}
-\epsilon_h^2\underbrace{\sum_{k}w_{h,k}\Big(\nabla\cG_h^{k}-\sum_{\ell}w_{h,\ell}\nabla\cG_h^{\ell}\Big)\!\Big(\nabla\cG_h^{k}-\sum_{\ell}w_{h,\ell}\nabla\cG_h^{\ell}\Big)^\top}_{\text{(T2)}},
\]
Since (T1) is the difference of a bounded matrix and a convex combination of bounded matrices, and (T2) is a finite sum of outer products of bounded vectors, both terms are bounded. Hence, $\|\nabla^2 f_{h}(v)\|$ is uniformly bounded, implying $\|\nabla^2 f(v)\| \le b_2$ for some $b_2 > 0$.

\smallskip
\noindent\textit{Third derivative.}
Upon differentiating the Hessian once more, we obtain finite sums of tensors of the following types: 
(i) $\epsilon_h\big(\nabla^3\cG_h^{1} - \sum_k w_{h,k}\nabla^3\cG_h^{k}\big)$;
(ii) products $(\nabla w_{h,k})\otimes\nabla^{2}\cG_h^{k}$, $(\nabla^{2}w_{h,k})\otimes\nabla\cG_h^{k}$; and
(iii)  cubic outer products $(\nabla\cG_h^{k}-\sum_{\ell}w_{h,\ell}\nabla\cG_h^{\ell})^{\otimes3}$ multiplied by bounded weights.

Because all softmax weights satisfy $0 \le w_{h,k} \le 1$ and their gradients $\nabla w_{h,k}$ are bounded, being products of $w_{h,k}$ with bounded gradients, and since $\nabla \cG_h^k$, $\nabla^2 \cG_h^k$, and $\nabla^3 \cG_h^k$ are bounded, it follows that each term described above involves only finite sums and products of bounded quantities. Consequently, $\|\nabla^3 f_{h}(v)\|$ is uniformly bounded. As the tree-depth $d$ is fixed, summing over $h \in [d]$ yields $\|\nabla^3 f(v)\| \le b_3$ for some finite $b_3 > 0$.

\end{proof}

\begin{proposition}{}{(Weight \& Pivot Derivative Bounds)}
\label{prop:supporting2}
Under Assumptions~\ref{ass:1}, \ref{ass:2}, \ref{ass:3}, and \ref{ass:infodevbound}, there exist positive constants $C_{1},C_{2}$ such that:
\begin{enumerate}
    \item $\displaystyle\sup_{v \in D_{n}}\norm{\nabla^{3}\Lambda(v)}\leq C_1 \exp f\nbracket{\Pi_{n}}$
    \item $\displaystyle\sup_{v \in D_{n}}\norm{\nabla^{3} h \circ P(v) \times \Lambda(v)}\leq C_2 \exp f\nbracket{\Pi_{n}}$ 
\end{enumerate} for any arbitrary function $h \in \mathbb{C}^3(\mathbb{R})$ with bounded derivatives up to third order.
\end{proposition}

\begin{proof}\textbf{Part 1.} Since $\log \Lambda(v)=f(\Sigma_n^{1/2}v+\Pi_n)$, the chain rule gives,

\[
\|\nabla \log \Lambda(v)\|\le \|\Sigma_n^{1/2}\|\, b_1,\quad
\|\nabla^2 \log \Lambda(v)\|\le \|\Sigma_n^{1/2}\|^2\, b_2,\quad
\|\nabla^3 \log \Lambda(v)\|\le \|\Sigma_n^{1/2}\|^3\, b_3,
\]
using the bounds from Proposition~\ref{prop:supporting1}.
Since $\|\Sigma_n^{1/2}\|$ is uniformly bounded ( Remark~\ref{remark:ass:3}), the first- to third-order derivatives of $\log \Lambda(\cdot)$ are uniformly bounded over $D_n$.

For $\Lambda=e^{\log \Lambda}$, applying the product rule repeatedly via the Faà di Bruno formula gives

\[
\nabla^3 \Lambda
=\Lambda\big(\nabla^3\!\log \Lambda+3\,\mathrm{sym}(\nabla^2\!\log \Lambda\otimes \nabla\!\log \Lambda)
+(\nabla\!\log \Lambda)^{\otimes 3}\big),
\] 
Thus, by submultiplicativity, we obtain

\[
\|\nabla^3 \Lambda(v)\|
\le \Lambda(v)\Big(\|\nabla^3\!\log \Lambda(v)\|+3\|\nabla^2\!\log \Lambda(v)\|\,\|\nabla\!\log \Lambda(v)\|
+\|\nabla\!\log \Lambda(v)\|^3\Big).
\]
Using the previous bounds, we get $
\|\nabla^3 \Lambda(v)\|\ \le\ c_0\Lambda(v)$,
uniformly for $v \in D_n$, with $c_0$ depending only on $b_1, b_2, b_3$ and the uniform bound on $\|\Sigma_n^{1/2}\|$. Since $f$ is $L$-Lipschitz, $D_n$ is bounded, and $\|\Sigma_n^{1/2}\|$ is uniformly bounded (Remark~\ref{remark:ass:3}), it follows that

\[
\Lambda(v)=e^{f(\Sigma_n^{1/2}v+\Pi_n)}\le e^{f(\Pi_n)}\cdot e^{L\|\Sigma_n^{1/2}\|\|v\|}
\le c_1\,e^{f(\Pi_n)}, \quad 
\]
for all $v\in D_n$, and for a constant $c_1$. 
Combining,

\[
\sup_{v\in D_n}\|\nabla^3 \Lambda(v)\|
\le c_0\,c_1\,e^{f(\Pi_n)}
=: C_1\,e^{f(\Pi_n)}.
\]

\noindent
\textbf{Part 2.} 
For this part of the proof, we begin by rewriting the pivot as $P(v)=\frac{N(v)}{D(v)}$ where
\begin{align*}
N(v)&:=\!\!\int_{-\infty}^{\sigma_{n,1} v_1+\Pi_{n,1}}\!\!\phi(u;\Pi_{n,1},\sigma_{n,1}^{2})\,
\prod_{h=1}^{d}\lambda_h\big(u,\Sigma_{n, \sminus 1}^{1/2}v_{\sminus1}+\Pi_{n,\sminus1}\big)\,du \\
&= \int_{-\infty}^{v_1}\!\!\phi(u;0,1)\,
\Lambda(u,v_{-1})\,du,  \\
D(v)&:=\!\!\int_{-\infty}^{\infty}\!\phi(u;\Pi_{n,1},\sigma_{n,1}^{2})\,
\Lambda\big(u,\Sigma_{n, \sminus 1}^{1/2}v_{\sminus1}+\Pi_{n,\sminus1}\big)\,du = \int_{-\infty}^{\infty}\!\!\phi(u;0,1)\,
\Lambda(u,v_{-1})\,du,
\end{align*}
For $h\in \mathbb{C}^3(\mathbb{R})$, let $H_\ell := \sup_{x}\!\big|h^{(\ell)}(x)\big|$ for $\ell = 0,1,2,3$, with $H_\ell < \infty$. Since the integrands are positive, it follows that $0 \le P(v) \le 1$ for all $v$.

Firstly, we note that on $D_n$, the restricted map $v=(v_1,v_{\sminus1}) \mapsto \Lambda(u,v_{\sminus1})$ is $\mathbb{C}^3$. Similar to the previous bound $\sup_{w\in D_n}\|\nabla^{3}\Lambda(w)\|\le C_1 e^{f(\Pi_n)}$
using the chain rule and the uniform bounds on $\|\Sigma_n^{\pm 1/2}\|$, check that all $v$-derivatives of $\Lambda$ up to the third order are uniformly bounded on $D_n$ by a constant multiple of $e^{f(\Pi_n)}$.

Since $\phi(\cdot;0,1)$ and its first three derivatives are bounded, the Leibniz rule yields, for any partial derivative $\partial_v^\alpha$ with multi-index $\alpha=(\alpha_1, \alpha_{\sminus1}), \ |\alpha|\le 3$,
\begin{align*}
\partial^{\alpha}N(v)&=
\mathbbm{1}_{\{\alpha_1=0\}}
\int_{-\infty}^{v_1}
\phi(u;0,1)\,
\partial_{v_{-1}}^{\alpha_{-1}}
\Lambda(u,v_{-1})\,du
\\[0.5em]
&\quad+
\mathbbm{1}_{\{\alpha_1\ge 1\}}
\sum_{r+t=\alpha_1-1}
\binom{\alpha_1-1}{r}\,
\phi^{(r)}(v_1;0,1)\,
\partial_u^{\,t}
\partial_{v_{-1}}^{\alpha_{-1}}
\Lambda(u,v_{-1})
\Big|_{u=v_1},
\label{eq:Nprime}
\end{align*}
where all derivatives of $\Lambda$ above are of total order at most $3$. All such terms are uniformly bounded on $D_n$ by a constant multiple of $e^{f(\Pi_n)}$. The same argument (without a moving upper limit) gives the identical bound for $\partial^\alpha D(v)$, and $D(v)$ itself is bounded above and away from $0$ on $D_n$ by positivity/continuity.

For $P(v)=N(v)/D(v)$, the quotient and Faà di Bruno rules imply that every $\partial^{\alpha}P(v)$ with $|\alpha|\le3$ satisfies

\[
\partial^{\alpha}P(v)=
\sum_{\substack{\alpha_1+\cdots+\alpha_m=\alpha\\1\le m\le3}}
c_{\alpha_1,\ldots,\alpha_m}\,
\frac{\prod_{j=1}^m\partial^{\alpha_j}N(v)\,\partial^{\beta}D(v)}{D(v)^{m+1}},
\]
where each $|\alpha_j|,|\beta|\le3$ and $c_{\alpha_1,\ldots,\alpha_m}$ are finite constants.  
Because $D(v)$ is bounded away from zero and its derivatives, as well as those of $N(v)$, are bounded by $O(e^{f(\Pi_n)})$, we have

\[
\sup_{v\in D_n}\|\nabla^{j}P(v)\|\le C_P\,e^{f(\Pi_n)},\qquad j\le3.
\] Further, as $0\le P\le 1$ and $h^{(\ell)}$ are bounded, Faà di Bruno again gives that
$\partial^\alpha\!\big(h\!\circ\! P\big)(v)$ is uniformly bounded on $D_n$ for $|\alpha|\le 3$.

Finally, let $A(v):=h(P(v))$ and $B(v):=\Lambda(v)$, then by the (multi-index) product rule,

\[
\partial^\alpha\!\big(A B\big)
=\sum_{\beta\le\alpha}\binom{\alpha}{\beta}\,
\partial^{\alpha-\beta}A\,\partial^{\beta}B, \qquad |\alpha|=3.
\]
Each factor $\partial^{\alpha-\beta}A$ is uniformly bounded on $D_n$ by Step~2, while $\partial^\beta B$ with $|\beta|\le 3$ is bounded on $D_n$ by a constant multiple of $e^{f(\Pi_n)}$ (for $|\beta|=3$ this is the assumption; for $|\beta|\le 2$ it follows by integrating the third derivative along line segments within $D_n$ and using the already established bounds). Consequently,

\[
\sup_{v\in D_n}\|\nabla^{3}\big(h\circ P\cdot \Lambda\big)(v)\|
\ \le\ C_2\,e^{f(\Pi_n)} 
\]
for some finite constant $C_2$ depending on bounds $C_1,H_0,H_1,H_2,H_3$, and on the uniform bounds for $\Sigma_n^{\pm1/2}$.
\end{proof}

\begin{lemma}
\label{lemma:supporting1}
For the standardized variables $\zeta_n$ and $Z_n$, the following statements hold:
\begin{enumerate}
\item $\mathbb{E}  \left[h \circ P\left(\zeta_{n}\right) \mid \cbracket{\bS_{1}=\bs_{1}} \right] =\dfrac{\mathbb{E}\left[h \circ P\left(\zeta_{n}\right) \times \Lambda\left(\zeta_{n}\right)\right]}{\mathbb{E}\left[\Lambda\left(\zeta_{n}\right)\right]}$
\item $
\mathbb{E}\left[h \circ P(Z_{n}) \mid\cbracket{\bS_{1}=\bs_{1}}\right] =\dfrac{\mathbb{E}\left[h \circ P(Z_{n}) \times \Lambda(Z_{n})\right]}{\mathbb{E}\left[\Lambda(Z_{n})\right]}$
\end{enumerate}
\end{lemma}

\begin{proof}[Proof of Lemma~\ref{lemma:supporting1}]
To prove (1), first recall that, by construction,
$\Lambda(\zeta_n)=\mathbb{P}\nbracket{\bS_{1}=\bs_{1} \mid \zeta_n}$. 
Applying the tower property of expectation, we obtain 
\begin{align*}
\begin{gathered}
\mathbb{E}\!\left[h\!\circ\! P(\zeta_n)\,\middle|\,\bS_{1}=\bs_{1}\right]
=\frac{\mathbb{E}\!\left[h\!\circ\! P(\zeta_n)\,\mathbbm{1}\nbracket{\bS_{1}=\bs_{1}}\right]}
{\mathbb{P}\rbracket{\bS_{1}=\bs_{1}}}
=\frac{\mathbb{E}\!\left[\mathbb{E}\!\left(h\!\circ\! P(\zeta_n)\,\mathbbm{1}\nbracket{\bS_{1}=\bs_{1}}\mid \zeta_n\right)\right]}
{\mathbb{E}\!\left[\mathbb{E}\!\left(\mathbbm{1}\nbracket{\bS_{1}=\bs_{1}}\mid \zeta_n\right)\right]}.
\end{gathered}
\end{align*}
Now, note that

\begin{align*}
\begin{gathered}
\mathbb{E}\!\left[h\!\circ\! P(\zeta_n)\,\mathbbm{1}\nbracket{\bS_{1}=\bs_{1}}\mid \zeta_n\right]
= h\!\circ\! P(\zeta_n)\,\mathbb{E}\!\left[\mathbbm{1}\nbracket{\bS_{1}=\bs_{1}}\mid \zeta_n\right]
= h\!\circ\! P(\zeta_n)\,\Lambda(\zeta_n), \ \text{ and } \\
\mathbb{E}\!\left[\mathbbm{1}\nbracket{\bS_{1}=\bs_{1}}\mid \zeta_n\right]=\Lambda(\zeta_n).
\end{gathered}
\end{align*}
Therefore,
\[
\mathbb{E}\!\left[h\!\circ\! P(\zeta_n)\,\middle|\,\bS_{1}=\bs_{1}\right]
=\frac{\mathbb{E}\!\left[h\!\circ\! P(\zeta_n)\,\Lambda(\zeta_n)\right]}
{\mathbb{E}\!\left[\Lambda(\zeta_n)\right]}.
\]
This proves (1). 
The proof of (2) is identical, with \(\zeta_n\) replaced by \(Z_n\) throughout.
\end{proof}

\begin{lemma}
\label{lemma:supporting2}
Under Assumptions~\ref{ass:1},\ref{ass:2},\ref{ass:3}, there exists a constant
$C_->0$, not depending on $n$, such that for all sufficiently large $n$,

\[
\mathbb{E}\Big[\exp\!\big(-LB\|Z_n\|\big)\,\mathbbm{1}_{D_n}(Z_n)\Big]\;\ge\; C_-.
\] for constants $L>0$ and $B>0$.
\end{lemma}

\begin{proof}
For any $R>0$, we have

\[
\exp\!\big(-LB\|Z_n\|\big)\,\mathbbm{1}_{D_n}(Z_n)
\;\ge\; e^{-LBR}\,\mathbbm{1}_{\{\|Z_n\|\le R\}\cap D_n}(Z_n).
\]

Taking expectations and applying the union bound, we observe that
\begin{align*}
\mathbb{E}\!\left[\exp\!\big(-LB\|Z_n\|\big)\,\mathbbm{1}_{D_n}(Z_n)\right]
&\ge e^{-LBR}\,\mathbb{P}\!\left(\{\|Z_n\|\le R\}\cap \{Z_n \in D_n\}\right)\\
&\ge e^{-LBR}\,\Big(\mathbb{P}\rbracket{Z_n \in D_n}-\mathbb{P}\rbracket{\|Z_n\|>R}\Big).
\end{align*}
Since $Z_n\sim\mathcal{N}_{\bar{k}}(0,I_{\bar{k}})$, we have $\mathbb{E}\|Z_n\|^2=\bar{k}$.
By Chebyshev’s inequality,
\[
\sup_{n}\,\mathbb{P}\rbracket{\|Z_n\|>R}\;\le\;\frac{\bar{k}}{R^2}.
\] In particular, choose $R = 2\sqrt{\bar{k}}$ so that $\sup_n \mathbb{P}\rbracket{\|Z_n\| > R} \le \tfrac{1}{4}$. Now, since each $b_{i,n}$ is uniformly bounded, i.e., $\sup_{n,i}|b_{i,n}| \le \bar{C}_{+}$ (by Proposition~\ref{prop:supporting0}), and the uniform multivariate Berry–Esseen bound in \cite{BENTKUS2003385} applies, it follows that 
\[
\sup_{A\in\mathcal{C}_{\bar{k}}}
\big|\mathbb{P}\rbracket{\zeta_n\in A}-\mathbb{P}\rbracket{Z_n\in A}\big|
\;\lesssim\;
\bar{k}^{1/4}\,\frac{\mathbb{E}\|b_{i,n}\tilde Y_{i,n}\|_2^3}{\sqrt{n}}
=O(n^{-1/2}),
\] where the supremum is taken over convex sets in $\mathbb{R}^{\bar{k}}$, denoted $\mathcal{C}_{\bar{k}}$.
Because $D_n$ is convex and $\mathbb{P}\rbracket{\zeta_n\in D_n}=1$, applying the bound with $A=D_n$ yields
\[
\mathbb{P}\rbracket{Z_n\in D_n}
\;\ge\;
\mathbb{P}\rbracket{\zeta_n\in D_n}-O(n^{-1/2})
\;=\;
1-O(n^{-1/2}),
\] which implies that there exists $n_0$ such that for all $n\ge n_0$, $\mathbb{P}\rbracket{Z_n\in D_n} \ge \tfrac{3}{4}$. Hence, for all sufficiently large $n$,
\[
\mathbb{E}\!\left[\exp\!\big(-LB\|Z_n\|\big)\,\mathbbm{1}_{D_n}(Z_n)\right]
\;\ge\; e^{-LB(2\sqrt{\bar{k}})}\!\left(\tfrac{3}{4}-\tfrac{1}{4}\right)
\;=\;\tfrac{1}{2}\,e^{-2LB\sqrt{\bar{k}}}
\;=:\; C_- \;>\;0,
\]where $C_{-}$ does not depend on $n$.

\end{proof}

\section{Adjusting for data-adaptive temperature parameters}
\label{appendix:dataadaptivetemp}

Recall from Equation~\ref{eq:dataadaptivetemp}, for $h \in [d]$, the data-adaptive temperatures with a  fixed \emph{temperature scale} hyperparameter $\tau>0$ follow
\begin{align*}
\overline{\epsilon}_h(t)
\;:=\;
\tau\,\frac{1}{k_h}\sum_{k=1}^{k_h}\cG_h^k(t), \
\forall t\in\tilde D_n; \qquad \overline{\epsilon}_{h}(\hat{T}_{n})
=\tau\times\frac{\sum_{k=1}^{k_h}\cG^{k}_{h}(\hat{T}_{n})}{ k_h}.
\end{align*} The adjusted pivotal quantity in Equation~\ref{pivot:adjusted} can be expressed in terms of the standardized variables as $ \overline{\text{\normalfont Pivot}}(\hat{T}_{n,1}; \Pi_{n,1})=\overline{P}(\zeta_{n})$, where
$$\overline{P}(\zeta_{n}):=\overline{P}(\zeta_{n}; \Pi_n)= \frac{\int_{-\infty}^{\sigma_{n,1} \zeta_{n,1}+\Pi_{n,1}}\phi(u;\Pi_{n,1},\sigma_{n,1}^{2})\prod_{h=1}^{d}\overline{\lambda}_{h}(u,\Sigma_{ n,\sminus 1}^{1/2}\zeta_{n, \sminus 1} + \Pi_{n, \sminus 1} )  du}{\int_{-\infty}^{\infty}\phi(u;\Pi_{n,1},\sigma_{n,1}^{2})\prod_{h=1}^{d}\overline{\lambda}_{h}(u,\Sigma_{ n,\sminus 1}^{1/2}\zeta_{n, \sminus 1} + \Pi_{n, \sminus 1})  du}.$$Analogously, define adjusted functions $\overline{f}:\tilde{D}_{n}\to \mathbb{R}$ and  $\overline{\Lambda}: D_n \to (0,1)$ as
\vspace{-0.5cm}
\begin{align*}
\bar f(t)
& :=
\sum_{h=1}^{d}
\Bigg[
\overline{\epsilon}_h(t)^{-1}\,\cG_h^{1}(t)
\;-\;
\log\!\Bigg(\sum_{k=1}^{k_h}
\exp\!\big\{\overline{\epsilon}_h(t)^{-1}\,\cG_h^{k}(t)\big\}\Bigg)
\Bigg]\\
\overline{\Lambda}(\zeta) &:=  \exp \overline{f}(\Sigma_{n}^{1/2}\zeta+\Pi_{n}),
\end{align*}

Recall Assumption~\ref{ass:extra} is an extra assumption necessary to extend the same asymptotic guarantees for the adjusted pivot. Specifically, Assumption~\ref{ass:extra} imposes a mild condition on the adaptive temperatures $\{\overline{\epsilon}_h:{h\in[d]}\}$ with scale $\tau>0$ that there exists a constant $m_\tau>0$ such that $\overline{\epsilon}_h(v)\;\ge\; m_\tau$
for all  $v\in\tilde D_n$  and all $h\in[d]$.

Firstly we prove the result bellow which is analogous to Proposition~\ref{prop:supporting1} for the adjusted function $\overline{f}(.)$.
\begin{proposition}
Under Assumptions~\ref{ass:1}, \ref{ass:2}, \ref{ass:3}, \ref{ass:infodevbound} and \ref{ass:extra}, there exist finite constants $\bar{L}, \bar{b}_1,\bar{b}_2,\bar{b}_3>0$ such that the function $\overline{f}(.)$ is $\bar{L}$-Lipchitz and 
\[
\|\nabla \bar f(v)\|\le \bar{b}_1,
\qquad
\|\nabla^2 \bar f(v)\|\le \bar{b}_2,
\qquad
\|\nabla^3 \bar f(v)\|\le \bar{b}_3,
\qquad
\text{for all } v\in\tilde D_n .
\]
\end{proposition}

\begin{proof}

Write $\bar f(v)=\sum_{h=1}^d \bar f_h(v)$, where
\[
\bar f_h(v)
=
\overline{\epsilon}_h(v)^{-1}\cG_h^{1}(v)
-
\log\!\Big(\sum_{k=1}^{k_h}e^{\overline{\epsilon}_h(v)^{-1}\cG_h^{k}(v)}\Big).
\]
Recall the gains $\cG_h^k(v)$ are always non-negative and also uniformly bounded over $\tilde D_n$. Moreover, since  $\overline{\epsilon}_h(v)\ge m_\tau$ (from Assumption~\ref{ass:extra}), the following uniform bound is satisfied 
\[
\sup_{v\in\tilde D_n}\overline{\epsilon}_h(v)^{-r}\cG_h^k(v)<\infty,
\qquad r=1,2,3,4.
\] Let's denote the soft-max weights as
\[
\overline{w}_{h,k}(v)
=
\frac{e^{\overline{\epsilon}_h(v)^{-1}\cG_h^k(v)}}
{\sum_{\ell=1}^{k_h}e^{\overline{\epsilon}_h(v)^{-1}\cG_h^\ell(v)}}.
\]

\noindent\emph{Gradient.}
By the chain rule,
\[
\nabla\bar f_h(v)
=
\nabla\!\big(\overline{\epsilon}_h^{-1}\cG_h^1\big)
-
\sum_{k=1}^{k_h}\overline{w}_{h,k}(v)\,
\nabla\!\big(\overline{\epsilon}_h^{-1}\cG_h^k\big) \ \text{ where } \ \overline{w}_{h,k}(v)
=
\frac{e^{\overline{\epsilon}_h(v)^{-1}\cG_h^k(v)}}
{\sum_{\ell=1}^{k_h}e^{\overline{\epsilon}_h(v)^{-1}\cG_h^\ell(v)}}.
\]
Now, using product rule to get the gradient as 
\[
\nabla\!\big(\overline{\epsilon}_h^{-1}\cG_h^k\big)
=
\overline{\epsilon}_h^{-1}\nabla\cG_h^k
-
\overline{\epsilon}_h^{-2}\cG_h^k\nabla\overline{\epsilon}_h,
\]
in which all terms are uniformly bounded on $\tilde D_n$, hence
$
\sup_{v\in\tilde D_n}\|\nabla\bar f_h(v)\|<\infty.$ Note summing over $h$ the above yields $\sup_{v\in\tilde D_n}\|\nabla\bar f_h\|<\infty$ which further implies that $\bar f_h$ is also $\bar{L}$-Lipchitz function for some constant $\bar{L}$.

\noindent\emph{Higher derivatives.}
Applying the chain rule repeatedly,
$\nabla^2\bar f_h$ and $\nabla^3\bar f_h$ can besimilarly  written as finite sums of terms involving
\[
\nabla^j\cG_h^k,\quad
\nabla^j\overline{\epsilon}_h,\quad
\overline{w}_{h,k},\quad
\nabla^j \overline{w}_{h,k},
\qquad j\le 3,
\]
multiplied by powers $\overline{\epsilon}_h^{-r}$ with $r\le4$.
Since $0\le \cG_h^k\le C_G$, $\overline{\epsilon}_h\ge m_\tau$, and all derivatives of
$\cG_h^k$ up to order three are uniformly bounded on $\tilde D_n$,
each such term is uniformly bounded. Therefore,
\[
\sup_{v\in\tilde D_n}\|\nabla^2\bar f_h(v)\|<\infty,
\qquad
\sup_{v\in\tilde D_n}\|\nabla^3\bar f_h(v)\|<\infty.
\]Since the depth $d$ is fixed, summing over $h$ yields the desired bounds for
$\bar f$. 
\end{proof}

Note that other supporting results like weight and pivot derivative bounds in Proposition~\ref{prop:supporting2} can be similarly reproduced for their adjusted versions $\overline{\Lambda}, \overline{P}$ to get bounds like:
\begin{enumerate}
    \item $\displaystyle\sup_{v \in D_{n}}\norm{\nabla^{3}\overline{\Lambda}(v)}\leq \overline{C}_1 \exp \overline{f}\nbracket{\Pi_{n}}$
    \item $\displaystyle\sup_{v \in D_{n}}\norm{\nabla^{3} h \circ \overline{P}(v) \times \overline{\Lambda}(v)}\leq \overline{C}_2 \exp \overline{f}\nbracket{\Pi_{n}}$ 
\end{enumerate} for any arbitrary function $h \in \mathbb{C}^3(\mathbb{R})$ with bounded derivatives up to third order and some positive constants $\overline{C}_{1},\overline{C}_{2}$.

To prove the main result in Theorem~\ref{thm: adjusted asymptotic validity}, weak limit of the adjusted pivot as a $\text{Uniform}(0,1)$ random variable, it is sufficient to show the following:
\vspace{-0.5cm}
\begin{equation}
\label{adjpivot:WeakConv}
\lim_{[\bn]  \to \infty} \Big \lvert \mathbb{E}  \left[h \circ \overline{P}\left(\zeta_{n}\right) \mid \cbracket{\bS_{1}=\bs_{1}} \right]  -\mathbb{E}\left[h \circ \overline{P}(Z_{n}) \mid\cbracket{\bS_{1}=\bs_{1}}\right] \Big \rvert=0, \vspace{-0.5cm}
\end{equation}
for every real-valued $h \in \Bigl\{\tilde{h} \in \mathbb{C}^3(\mathbb{R}): \; \sup_{x \in \mathbb{R}} \big| \grad^{(k)}\tilde{h}(x) \big| < \infty \;\; \text{for } k = 0,1,2,3 \Bigr\}$. 

\begin{proposition}{}{(Adjusted Relative Differences)}
Define the difference terms
\vspace{-0.5cm}
\begin{align*}
\begin{gathered}
\overline{\RD}_n^{(1)}=\frac{\left|\mathbb{E}\left[\overline{\Lambda}(\zeta_n)\right]-\mathbb{E}\left[\overline{\Lambda}(Z_n) \right]\right|}{\mathbb{E}\left[\overline{\Lambda}(Z_{n})\right]} \\
\overline{\RD}_n^{(2)}=\frac{\left|\mathbb{E}\left[h \circ \overline{P}\left(\zeta_{n}\right) \times \overline{\Lambda}(\zeta_n) \right]-\mathbb{E}\left[h \circ \overline{P}(Z_{n}) \times \overline{\Lambda}(Z_n) \right]\right|}{\mathbb{E}\left[\overline{\Lambda}(Z_n)\right]}.
\vspace{-0.5cm}
\end{gathered}
\end{align*} 
where $h\in \mathbb{C}^3(\mathbb{R})$ and $\overline{P}$,$\overline{\Lambda}$ are as defined above.
Then, if 
\vspace{-0.5cm}
$$
\lim_{[\bn]  \to \infty} \overline{\RD}_n^{(1)}=0, \quad \lim_{[\bn]  \to \infty}  \overline{\RD}_n^{(2)}=0,
\vspace{-0.5cm}
$$
the weak convergence in \eqref{adjpivot:WeakConv} holds.
\end{proposition}

\begin{theorem}
   Under Assumptions~\ref{ass:1}, \ref{ass:2}, \ref{ass:3}, \ref{ass:infodevbound} and \ref{ass:extra}, we have 
   \vspace{-0.5cm}
   $$
   \lim_{[\bn]  \to \infty} \overline{\RD}_n^{(1)}=0, \quad \lim_{[\bn]  \to \infty}  \overline{\RD}_n^{(2)}=0. 
   \vspace{-0.5cm}
   $$ 
   \label{adjthm: verifying suff cond}
\end{theorem}

The proofs of the above results and other supporting results are omitted as they follow the exact same proof steps as before but now replacing all relevant functions (and bounds) with their corresponding adjusted function (and bounds). 

\paragraph{General adaptive temperatures:} Arguably, there is no uniformly best choice of data-adaptive temperature parameter $\epsilon_h=\bar\epsilon_h(\hat{T}_n)$. Assuming regularity conditions as described bellow (Assumption~\ref{ass:temp_regular}) on the data-adaptive temperature functions $\cbracket{\bar\epsilon_h(.): \tau>0}_{h\in[d]}$ allows us to extend the same theoretical guarantees for asymptotic inference obtained from our method. 

\begin{assumption}[Regularity of adaptive temperatures]
\label{ass:temp_regular}
For any $\tau>0$, the adaptive temperature functions
$\bar\epsilon_h(.):\tilde D_n\to(0,\infty)$ for all $h\in[d]$ with temperature scale $\tau$ satisfies the following:

\begin{enumerate}
\item (Uniform positivity.)  
There exists a constant $m_\tau>0$ such that $\inf_{v\in\tilde D_n} \bar\epsilon_h(v) \;\ge\; m_\tau$, for all $h\in [d]$.

\item (Smoothness and bounded derivatives.)
The function $\bar\epsilon_h(.)$ is three times continuously differentiable on
$\tilde D_n$, and there exist finite constants
$E_{\tau,1},E_{\tau,2},E_{\tau,3}>0$ such that
\[
\sup_{v\in\tilde D_n}\|\nabla \bar\epsilon_h(v)\|\le E_{\tau,1},\quad
\sup_{v\in\tilde D_n}\|\nabla^2 \bar\epsilon_h(v)\|\le E_{\tau,2},\quad
\sup_{v\in\tilde D_n}\|\nabla^3 \bar\epsilon_h(v)\|\le E_{\tau,3} \quad \forall h\in [d].
\]
\end{enumerate}
\end{assumption}To conclude, we note  that there is nothing special about the specific form of temperature (Equation~\ref{eq:dataadaptivetemp}) we used in our data experiments and our inference method can accommodate other reasonable data-adaptive choices as well. 

\section{Randomized cost-complexity pruning}
\label{appendix:ccp}

\begin{algorithm}[h]
\setstretch{1.3}
\caption{Randomized tree construction, pruning, and selective inference}
\begin{algorithmic}[1]

\State \textbf{Input:} Training data $(X,y)$; maximum depth $d_{\max}$; minimum node size $n_{\min}$; 
splitting temperature parameter(s) $\epsilon$; pruning complexity $\alpha$; pruning temperature $\rho$.

\vspace{0.15cm}
\Statex \textbf{(A) Grow a fully randomized tree}
\State Use Algorithm~\ref{alg:RCT} with inputs $(X,y,d_{\max},n_{\min},\epsilon)$ to obtain a fully randomized tree $\cT=\ct$.

\vspace{0.15cm}
\Statex \textbf{(B) Randomized cost-complexity pruning}
\State Enumerate all subtrees $\mathrm{Subtrees}(\ct)$.
\State Compute the penalized cost $C(\cT';y)$ for each subtree (Section~\ref{appendix:ccp}).
\State Draw a pruned subtree $\cT_\alpha=\ct_\alpha$ according to the randomized CCP distribution in Section~\ref{appendix:ccp}.

\vspace{0.15cm}
\Statex \textbf{(C) Selective inference on pruned leaves}
\For{each terminal node $l$ of the pruned tree $\ct_\alpha$}
    \State Construct the contrast vector $\nu_l$ and statistics $(\hat{T}_{n,l},\hat{T}_{n,\sminus l})$.
    \State Reconstruct $y(u,v)$ as in Section~\ref{appendix:ccp}.
    \State Evaluate the functions 
    $\lambda_{\text{grow}}(u,v)$ and $\lambda_{\text{CCP}}(u,v)$,  
    and their product $\lambda_{\text{prune}}(u,v)$.
    \State Form the selective pivot $\text{Pivot}_{\text{prune}}$ (Section~\ref{appendix:ccp}) 
    and numerically invert it to obtain a confidence interval for $\Pi_{n,l}$ (and hence $\pi_l$).
\EndFor

\vspace{0.15cm}
\State \textbf{Return:} Pruned tree $\ct_\alpha$ and selective confidence intervals for its terminal-region parameters.
\end{algorithmic}
\label{alg:RCT-CCP-compact}
\end{algorithm}

Given the fully grown classification tree $\cT = \ct$ obtained from the exponential randomization based splitting described in Section~\ref{sec:treegrowing}, we now describe a randomized analogue of the classical cost-complexity pruning (CCP) procedure for model selection over subtrees of $\ct$ and demonstrate subsequent inference procedure. Let 
\[
\mathrm{Subtrees}(\cT) = \{ \cT' : \cT' \text{ is a subtree of } \cT \},
\]
and denote by $\cL(\cT')$ the collection of terminal nodes (leaves) of a subtree $\cT'$. 

In the standard CCP framework \citep{breiman1984classification}, a subtree is selected by minimizing a penalized empirical risk criterion
\[
\cT_\alpha^{\mathrm{CCP}}
= 
\argmin_{\cT' \in \mathrm{Subtrees}(\cT)} 
C(\cT'; y)
\quad\text{where}\quad
C(\cT'; y) := R(\cT';y) + \alpha |\cL(\cT')|.
\]
Here, $\alpha \ge 0$ is the complexity parameter controlling the trade-off between goodness of fit and model size (number of leaves), and 
\[
R(\cT';y) := \sum_{l \in \cL(\cT')} I(l;y)
\]
denotes the empirical risk of $\cT'$, expressed as the sum of node impurities $I(l;y)$ across its leaves.

\vspace{1em}
\noindent
\textbf{Randomized Pruning.}  
We extend CCP to a randomized selection mechanism that introduces an additional layer of stochasticity after the fully randomized tree $\cT$ has been grown. Conditional on the observed fully grown tree $\cT=\ct$, we define a random subtree $T_\alpha$ drawn from a discrete distribution over $\mathrm{Subtrees}(\ct)$ with probabilities given by:
\[
\mathbb{P}\rbracket{\cT_\alpha = \ct_\alpha \mid \cT= \ct,Y= y}
=
\frac{
\exp\{-\rho^{-1} \, C(\ct_\alpha; y)\}
}{
\sum_{\ct' \in \mathrm{Subtrees}(\ct)} 
\exp\{-\rho^{-1} \, C(\ct'; y)\}
},
\quad \rho^{-1} > 0.
\]
Here, $\rho$ controls the strength of the randomization: as $\rho \to 0$, the selection concentrates on the classical minimizer $\cT_\alpha^{\mathrm{CCP}}$, whereas larger values of $\rho$ induce greater exploration among competing subtrees with similar penalized risks. 

\paragraph{Inference on Pruned Tree:} Here we extend our proposed method (for fixed depth trees in Section~\ref{sec:inference}) to obtain inference for the terminal-node mean parameters obtained of the pruned tree model. To ensure valid inference we condition on the final selection event, 
\[
\{ \cT_\alpha = \ct_\alpha, \; \cT = \ct \},
\] which includes the randomized pruning step.

\paragraph{Key Statistics:} Say, our inferential target is the scaled mean parameter $\Pi_{n,l}=\sqrt{n_l}\pi_l$ for the observed terminal node $l$ in the pruned tree $\cT_\alpha=\ct_\alpha$ and its natural estimator is $\hat{T}_{n,l}=\sqrt{n_l}\,\hat{\pi}_{l}$.
However, selection of $l$ as a leaf in the pruned tree does not depend solely on the subtree of $l$ but also depends on the fully grown tree and its subtrees (in contrast to an unpruned tree with leaf $l$ where accounting only for the upstream splits that lead to $l$ suffice). Instead, the pruning step evaluates the empirical risks of all nodes in the fully grown tree $\cT = \ct$, and hence involves the sample proportions $\{\hat{\pi}_\cR: \cR\ \text{is a region in the fitted tree} \ \ct\}$. Thus, valid selective inference for $\pi_l$ requires a representation that simultaneously captures $\hat{T}_{n,l}$ and all $\hat{\pi}_\cR$ appearing in the pruning criterion.

To obtain such a representation, for any region $\cR\subseteq\cX$ define the normalized contrast vector
\[
\nu_\cR
:=\frac{1}{\sqrt{n_\cR}}\bigl[\mathbbm{1}(X_i\in \cR):\, i\in[n]\bigr]^\top,
\quad\text{so that}\quad
\nu_\cR^\top Y=\sqrt{n_\cR}\,\widehat{\pi}_\cR.
\]
For the target leaf $l$, introduce the key statistics
\[
\hat{T}_{n,l} := \nu_l^\top Y = \sqrt{n_l}\hat{\pi}_l,
\qquad
\hat{T}_{n,\sminus l} := Y - \nu_l\nu_l^\top Y \in \mathbb{R}^n,
\]
where $\hat{T}_{n,\sminus l}$ contains all residual information orthogonal to $\nu_l$, which serves as the nuisance statistics that we would condition on.  The next lemma shows that every quantity of the form 
\[
\widehat{\pi}_{\sminus l}(\cR)
:= \sqrt{n_{\cR}}\widehat{\pi}_{\cR}
   - \frac{n_{\cR\cap l}}{\sqrt{n_{l}n_{\cR}}}\,\sqrt{n_l}\widehat{\pi}_{l}, \quad \text{(analogous to $\widehat{\pi}_{\sminus 1}(\cR)$ for $l=l_1$)}
\]
used previously to construct the nuisance statistics $\hat{T}_{n,\sminus 1}$ in Section~\ref{subsec:keystatistics}, is exactly a linear contrast of $\hat{T}_{n,\sminus l}$.

\begin{lemma}
For any region $\cR\subseteq\cX$, $
\nu_\cR^\top \hat{T}_{n,\sminus l}
=
\widehat{\pi}_{\sminus l}(\cR)$.
\end{lemma}

\begin{proof}
We compute
\[
\nu_\cR^\top \hat{T}_{n,\sminus l}
=
\nu_\cR^\top Y
-
(\nu_\cR^\top \nu_l)\,(\nu_l^\top Y).
\]
The first term is $\nu_\cR^\top Y=\sqrt{n_\cR}\hat{\pi}_\cR$. The inner product of contrasts equals
\[
\nu_\cR^\top\nu_l
=
\frac{1}{\sqrt{n_\cR n_l}}
\sum_{i=1}^n
\mathbbm{1}(X_i\in\cR)\mathbbm{1}(X_i\in l)
=
\frac{n_{\cR\cap l}}{\sqrt{n_\cR n_l}},
\]
and the last term is $\nu_l^\top Y=\sqrt{n_l}\hat{\pi}_l$. Substituting yields the claimed expression.
\end{proof}

Note based on the orthogonal decomposition $Y=\nu_l\hat{T}_{n,l}+\hat{T}_{n,\sminus l}$ one can reconstruct the response, as $Y=y(\hat{T}_{n,l},\hat{T}_{n,\sminus l})$ where $y(u,v):=\nu_l u + v$ is the reconstruction map. So, conditioning on $\cbracket{\hat{T}_{n,l}=u, \hat{T}_{n,\sminus l}=v}$ is same as conditioning on value the response being $Y=y(u,v)$. This establishes that all statistics required for characterizing the pruning selection event are captured by the key statistics $\hat{T}_{n,l}, \hat{T}_{n,\sminus l}$.

\paragraph{Pivot Construction.}
Let $\{S_1=s_1^1,\ldots,S_l=s_l^1\}$ denote the sequence of all realized splits on internal regions $\{p_1,\ldots,p_l\}$ that produce the fully grown tree $\cT=\ct$, ordered so that $p_h$ is an ancestor of $p_{h'}$ whenever $h<h'$.  
To construct the selective pivot for the pruned tree’s terminal-node parameter $\Pi_{n,l}$, we outline below the key ingredients.  


\begin{enumerate}
\item \textbf{Selection likelihood for tree growing.}
For each internal region $p_i$, the exponential-mechanism split rule with temperature $\epsilon_{p_{i}}$ implies
\[
\mathbb{P}\rbracket{S_i=s_i^1 \mid \hat{T}_{n,l}=u, \hat{T}_{n,\sminus l}=v}
=
\frac{\exp \nbracket{\epsilon_{p_{i}}^{-1} G_{p_{i}}(s_i^1; y(u,v))}}{\displaystyle\sum_{s'\in \cK(p_{i})}\exp \nbracket{\epsilon_{p_{i}}^{-1} G_{p_{i}}(s'; y(u,v))}}.
\]
Hence the probability of regenerating the entire observed tree $\ct$ under $(u,v)$ is 
\begin{align*}
\mathbb{P}\rbracket{ \cT= \ct \mid \hat{T}_{n,l}=u, \hat{T}_{n,\sminus l}=v}=\prod_{i=1}^{l}
\frac{\exp\!\big(\epsilon_{p_i}^{-1}\,G_{p_i}(s_i^1; y(u,v))\big)}{\sum_{s'\in\cK(p_i)}\exp\!\big(\epsilon_{p_i}^{-1}\,G_{p_i}(s'; y(u,v))\big)}=:\lambda_{\mathrm{grow}}(u,v)
\end{align*}

\item \textbf{Selection likelihood for cost-complexity pruning.}
Conditional on having grown the full tree $\ct$, the random pruning rule selects a subtree $\ct'$ with probability proportional to $\exp\{-\rho^{-1} C(\ct'; y(u,v))\}$.  
Thus the conditional probability of recovering the observed pruned tree $\ct_\alpha$, $\mathbb{P}\rbracket{\cT_\alpha=\ct_\alpha \mid \cT=\ct,\hat{T}_{n,l}=u,\hat{T}_{n,\sminus l}=v }$ is
\begin{align*}
 \lambda_{\mathrm{CCP}}(u,v)
&:=
\frac{\exp\{-\rho^{-1}\, C(\ct_\alpha; y(u,v))\}}
     {\sum_{\ct'\in \mathrm{Subtrees}(\ct)}
      \exp\{-\rho^{-1}\, C(\ct'; y(u,v))\}}
\end{align*}

\item \textbf{Joint selection likelihood.}
Multiplying the two components yields the probability of simultaneously regenerating both the full tree and its pruned version:
\[
\lambda_{\mathrm{prune}}(u,v)
:=\lambda_{\mathrm{grow}}(u,v)\times\lambda_{\mathrm{CCP}}(u,v)
=
\mathbb{P}\rbracket{
\cT=\ct,\;
\cT_\alpha=\ct_\alpha
\;\middle|\;
\hat{T}_{n,l}=u,\hat{T}_{n,\sminus l}=v}.
\]
All quantities in $\lambda_{\mathrm{grow}}$ and $\lambda_{\mathrm{CCP}}$ have closed-form expressions once the pseudo-response $y(u,v)$ is evaluated, since all gains and costs depend only on empirical class proportions in the corresponding regions.

\item \textbf{Selective pivot.}
The selective pivot for the target $\Pi_{n,l}$ is then given by
\[
\mathrm{Pivot}_{\mathrm{prune}}(\hat{T}_{n,l};\Pi_{n,l})
=
\frac{
\displaystyle\int_{-\infty}^{\hat{T}_{n,l}}
\phi(u;\Pi_{n,l},\sigma_l^2)\,
\lambda_{\mathrm{prune}}(u,\hat{T}_{n,\sminus l})\,du
}{
\displaystyle\int_{-\infty}^{\infty}
\phi(u;\Pi_{n,l},\sigma_l^2)\,
\lambda_{\mathrm{prune}}(u,\hat{T}_{n,\sminus l})\,du
},
\]
where $\sigma_l^2=\pi_l(1-\pi_l)$ is the variance of the leaf mean.  
In practice we replace $\sigma_l^2$ by its consistent plug-in estimator $\hat{\sigma}_l^2=\hat{\pi}_l(1-\hat{\pi}_l)$.
\end{enumerate}

This construction yields an exact conditional pivot under the assumed model, combining both the growing and pruning selection stages through the likelihood factor $\lambda_{\mathrm{prune}}$.

\section{Predictive Performance of an Ensemble of Randomized Classification Trees}
\label{appendix: ensembleRCT}

The tree-growing procedure introduced in Algorithm 1 (see Section 3 in the main paper) naturally leads to an ensemble method obtained by combining predictions from multiple independently trained randomized classification trees. We call this approach an \emph{ensemble RCT}. Similar to random forests, predictions are obtained by averaging across trees. However, unlike standard random forests, each tree is grown using the proposed randomized splitting rule instead of greedy CART splits.

We examine the predictive performance of the ensemble RCT and compare it with that of a standard random forest based on bootstrap aggregation and feature subsampling. The goal of this experiment is to evaluate how the alternative form of randomness in split selection introduced by the RCT approach compares with that of the conventional ensemble method. 

Data are generated from the high-dimensional extension of the logistic model we used in simulations (Section~\ref{sec:simulations}). Specifically, $Y_i \stackrel{\mathrm{ind}}{\sim} \mathrm{Bernoulli}(\theta_i)$ with 
$$\theta_i = \sigma\!\nbracket{\nbracket{s X_{i1} - s X_{i2} + 0.3 s X_{i1}X_{i2} + \gamma \sin(2\pi X_{i1})}},$$ 
where where $X_{i1}, X_{i2} \sim \text{Uniform}(-1,1)$ are i.i.d covariates and $\sigma(t) = (1+e^{-t})^{-1}$ is the logistic sigmoid function. Along with the $2$ signal variables $X_{i1},X_{i2}$, we augment the feature vector with $18$ nuisance covariates $Z_{ij} \stackrel{\mathrm{i.i.d.}}{\sim} \mathcal{N}(0,1)$, independent of $(X_{i1},X_{i2})$, yielding $X_i = (X_{i1},X_{i2},Z_{i1},\dots,Z_{i,18})$. We set $s=3$, $\gamma=0.5$, and $m=0.5$.

For each of the $100$ Monte Carlo replicates, a sample of size $n=200$ is randomly split into training and test sets of equal size. Test-set accuracy is evaluated for ensemble sizes $B \in \{10,50,100,200\}$. In both types of ensembles, predictions are formed by aggregating $B$ trees grown to a fixed maximum depth of $3$ with minimum leaf size $10$. Class labels are determined by majority vote, and class probabilities are obtained by averaging across trees.

\begin{figure}[h]
    \centering
    \includegraphics[width=13.5cm]{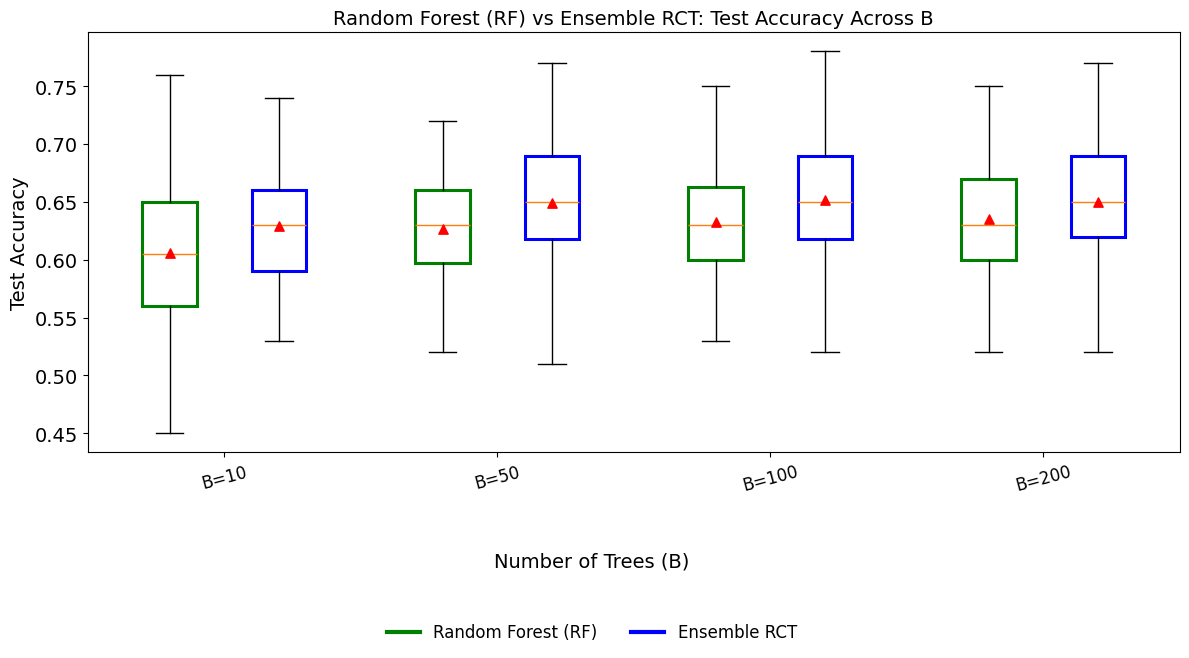}
    \caption{Test-set accuracy of the ensemble RCT and the standard random forest across Monte Carlo replicates for varying ensemble sizes $B$. Red triangles denote Monte Carlo means.}
    \label{fig:rfcomp}
\end{figure}

The ensemble RCT differs from the standard random forest in how randomness is introduced during split selection. Each randomized tree considers all covariates at every split and selects among candidate split points using an exponential mechanism based on impurity reduction. In contrast, the random forest grows CART trees on bootstrap resamples and, at each split, greedily selects the best split from a uniformly sampled subset of $\sqrt{p}$ covariates. 
 
 Figure~\ref{fig:rfcomp} summarizes the resulting distributions of test accuracy for both types of ensembles.
From a purely predictive standpoint, the ensemble RCT performs comparably to the standard random forest across all ensemble sizes considered. With relatively few trees, the ensemble RCT achieves slightly higher average accuracy. As the number of trees increases, the performance of the two methods becomes similar.


\section{Additional data application \& details}

\subsection{QSAR Biodegradation Data preprocessing}
\label{app:QSAR_preproc}
The QSAR Biodegradation dataset was obtained from the UCI Machine Learning Repository and consists of \(1055\) chemical compounds. Each compound is described by \(41\) quantitative molecular descriptors, together with a categorical class label indicating whether the compound is \emph{ready biodegradable} (RB) or \emph{not ready biodegradable} (NRB).

The raw data are provided as a semicolon-separated file with no header. We assign generic feature names \(\{\texttt{QSAR}_1,\ldots,\texttt{QSAR}_{41}\}\) to the molecular descriptors and treat the final column as the class label. The binary response variable is defined as $Y = \mathbbm{1}\{\text{class} = \text{RB}\}$, so that \(Y=1\) corresponds to ready biodegradable compounds.

All \(41\) predictors are treated as continuous and are standardized to have zero mean and unit variance using the training data. No missing values are present in the dataset, and no observations are removed during preprocessing. The final design matrix has dimension \(1055 \times 41\), with an empirical class prevalence of
$\mathbb{P}(Y=1) \approx 0.337$.

This preprocessing yields a fully continuous feature space, well suited to axis-aligned splits in classification trees while avoiding scale imbalances across molecular descriptors.

\subsection{Additional data application: Predicting and estimating diabetes readmission rates}
\label{app:diabetes}

We include a second data application based on the Diabetes 130-US Hospitals dataset to evaluate the proposed RCT method in a large-sample, heterogeneous clinical setting. The same experimental protocol and methods described in Section~\ref{sec:dataapps} are applied. In contrast to the QSAR example, this dataset reflects a modern electronic health record environment with mixed predictor types, higher noise, and substantially larger sample size, allowing us to assess scalability and empirical stability beyond moderate-sample scientific data.

\paragraph{Data description and preprocessing.}
\label{app:diabetes_preproc}

The Diabetes 130-US Hospitals dataset \citep{uci_diabetes_130} contains \(101{,}766\) hospital encounters from \(130\) U.S. hospitals, with demographic, administrative, and clinical variables. Identifier variables (\texttt{encounter\_id}, \texttt{patient\_nbr}) are removed.

We restrict attention to clinically relevant subpopulation of patients who were readmitted at least once, discarding cases with \texttt{readmitted = NO}. The binary response is defined as 
\(
Y = \mathbbm{1}\{\text{readmitted} < 30 \text{ days}\},
\)
so that \(Y=1\) indicates \emph{early readmission} among readmitted patients.

Missing values encoded as \texttt{"?"} are converted to \texttt{NaN}, and a complete-case analysis is performed on a selected subset of predictors.Specifically, we retain eight continuous utilization variables—
``time\_in\_hospital", ``num\_lab\_procedures", ``num\_procedures",
``num\_medications", ``number\_outpatient", ``number\_emergency",
``number\_inpatient" and ``number\_diagnoses" and eight categorical or
ordinal variables: ``age", ``race", ``gender", ``admission\_type\_id",
``diabetesMed", ``change", ``discharge\_disposition\_id" and
"admission\_source\_id". Categorical predictors are ordinally encoded, with unseen categories mapped to a dedicated code, and continuous predictors are standardized. After preprocessing, the final dataset consists of \(46{,}176\) encounters with \(16\) predictors (8 continuous and 8 categorical), with empirical early-readmission prevalence \(\mathbb{P}(Y=1) \approx 0.242\). Predicting and estimating early (or) rapid readmission supports discharge planning, and interpretable trees yield transparent decision rules based on admission and utilization patterns.




\begin{table}
\centering
\label{tab:diabetes_results}
\begin{tabular}{lcc}
\toprule
Method &  Test accuracy & Avg. CI length \\
\midrule
CART (full) & \textbf{0.760} & - \\
DS ($\rho = 0.1$) & \textcolor{red}{\textbf{0.760}} & \textcolor{red}{\textbf{0.120}} \\
DS ($\rho = 0.3$) & 0.758 & 0.184 \\
DS ($\rho = 0.5$) & 0.756 & 0.093 \\
DS ($\rho = 0.7$) & 0.757 & 0.087 \\
RCT ($\tau=0.1$) & 0.760 & 0.110 \\
RCT ($\tau=0.5$) & 0.760 & 0.065 \\
RCT ($\tau=1$)  & \textcolor{blue}{\textbf{0.760}} & \textcolor{blue}{\textbf{0.059}} \\
\bottomrule
\end{tabular}
\caption{Diabetes Readmission: Comparison of predictive performance (via Test accuracy) and inference power (via average confidence interval length). Here $\rho$ denotes the inference fraction used in splitting for DS, and $\tau$ denotes the temperature scale used for fitting RCT.}
\end{table}

The results for CART, DS, and RCT are summarized in Table~\ref{tab:diabetes_results}.

\paragraph{Interpretation of results.}
In Table~\ref{tab:diabetes_results}, test accuracy is essentially unchanged across DS and RCT over the reported values of \(\rho\) and \(\tau\), indicating that predictive performance is stable across these choices in this large-sample setting. Data splitting maintains accuracy across a wide range of inference fractions, but interval lengths depend strongly on $\rho$. RCT, by contrast, retains full-sample predictive performance while delivering progressively shorter confidence intervals as the temperature scale $\tau$ increases. This behavior reflects the stability of the fitted randomized tree and the reduced cost of conditioning for selective inference at large sample sizes. These results demonstrate that the benefits of randomized fitting extend beyond moderate-sample scientific datasets: RCT scales effectively to noisy, heterogeneous clinical data while providing valid and increasingly efficient uncertainty quantification without sacrificing predictive accuracy.

\end{appendices}
\end{document}